\documentclass[a4paper,USenglish,cleveref, autoref, thm-restate,numberwithinsect]{lipics-v2021}

\newcommand{\remove}[1]{}

\newcommand{\vc}{\textsf{vc}}
\newcommand{\ndd}{\textsf{nd}}
\newcommand{\clp}{\textsf{pc}}
\newcommand{\vi}{\textsf{vi}}

\newcommand{\defproblem}[3]{
  \vspace{3mm}
\noindent\fbox{
  \begin{minipage}{.95\textwidth}
  \begin{tabular*}{\textwidth}{@{\extracolsep{\fill}}lr} #1  \\ \end{tabular*}
  {\bf{Input:}} #2  \\
  {\bf{Objective:}} #3
  \end{minipage}
  }
  \vspace{2mm}
  }

\newcommand{\defdecproblem}[3]{
  \vspace{3mm}
\noindent\fbox{
  \begin{minipage}{.95\textwidth}
  \begin{tabular*}{\textwidth}{@{\extracolsep{\fill}}lr} #1  \\ \end{tabular*}
  {\bf{Input:}} #2  \\
  {\bf{Question:}} #3
  \end{minipage}
  }
  \vspace{2mm}
  }

\newcommand{\cA}{{\mathcal A}}

\newcommand{\CC}{{\mathcal C}}
\newcommand{\DD}{{\mathcal D}}

\newcommand{\FF}{\ensuremath{\mathcal{F}}\xspace}
\newcommand{\GG}{{\mathcal G}}

\newcommand{\OO}{\mathcal{O}}

\newcommand{\UU}{{\mathcal U}}

\newcommand{\YY}{{\mathcal Y}}

\newcommand{\nn}{{\mathbb N}}

\newcommand{\nka}{${\rm NP \subseteq coNP/poly}$}

\newcommand{\sol}{\sf Sol}

\newcommand{\MC}{\sc Matching Cut}

\newcommand{\dcut}{\sc $d$-Cut}

\newcommand{\enumdcut}{\sc Enum $d$-Cut}
\newcommand{\enummaxdcut}{\sc Enum Max-$d$-Cut}
\newcommand{\enummindcut}{\sc Enum Min-$d$-Cut}

\newcommand{\nd}{{\rm nd}}
\newcommand{\mch}{\sf mch}

\newcommand{\sv}[1]{}

\newif\iflong\longtrue

\usepackage{tikz}
\usetikzlibrary{decorations.shapes,decorations.pathreplacing,decorations.pathmorphing}
\usetikzlibrary{arrows,matrix,shapes}
\usetikzlibrary{positioning}
\tikzset{
        stars/.style={star,inner sep=2pt}
    }

\usepackage[disable]{todonotes}
\usepackage[ruled,vlined,linesnumbered]{algorithm2e}
\usepackage{hyperref}
\usepackage{xcolor}
\usepackage{cleveref}

\usepackage{microtype}

\newtheorem{reduction rule}{\bf Reduction Rule}[section]

\newtheorem{obs}{\bf Observation}[section]

\iflong \title{Enumeration Kernels of Polynomial Size for Cuts of Bounded Degree}
\else \title{Enumeration Kernels of Polynomial Size for Cuts of Bounded Degree}
  \fi

\titlerunning{Enumeration Kernels for Bounded Degree Cuts} 

\author{Christian Komusiewicz}{University of Jena, Germany \and \url{https://www.fmi.uni-jena.de/en/6050/prof-dr-christian-komusiewicz}}{c.komusiewicz@uni-jena.de}{https://orcid.org/0000-0003-0829-7032}{}

\author{Diptapriyo Majumdar}{Indraprastha Institute of Information Technology Delhi, New Delhi, India \and \url{https://diptapriyomajumdar.wixsite.com/toto}}{diptapriyo@iiitd.ac.in}{https://orcid.org/0000-0003-2677-4648}{}

\authorrunning{C. Komusiewicz and D. Majumdar} 

\Copyright{Christian Komusiewicz and Diptapriyo Majumdar} 

\ccsdesc[100]{Theory of computation~Parameterized complexity and exact algorithms} 

\keywords{Parameterized Complexity, Enumeration Algorithms, $d$-Cut, Parameterized Enumeration, Polynomial-Delay Enumeration Kernelization} 

\funding{Research of Diptapriyo Majumdar has been supported by Science and Engineering Research Board (SERB) grant SRG/2023/001592.}

\nolinenumbers 

\EventEditors{Jane Open Access and Joan R. Public}
\EventNoEds{2}
\EventLongTitle{42nd Conference on Very Important Topics (CVIT 2016)}
\EventShortTitle{CVIT 2016}
\EventAcronym{CVIT}
\EventYear{2016}
\EventDate{December 24--27, 2016}
\EventLocation{Little Whinging, United Kingdom}
\EventLogo{}
\SeriesVolume{42}
\ArticleNo{23}

\begin{document}

\maketitle

\begin{abstract}
Enumeration Kernelization was first proposed by Creignou et al. [TOCS 2017] and was later refined by Golovach et al. [JCSS 2022, STACS 2021] into two different variants: fully-polynomial enumeration kernelization and polynomial-delay enumeration kernelization.
A central problem in the literature on 
this topic is the {\MC}. 
In this paper, we consider the more general NP-complete {\dcut} problem. 
%
%
%
The decision version of {\dcut} asks if a given undirected graph $G$ has a cut $(A, B)$ such that every $u \in A$ has at most $d$ neighbors in $B$ and every $v \in B$ has at most $d$ neighbors in $A$.
%
In this paper, we study three different enumeration variants of {\dcut} -- {\enumdcut}, {\enummindcut} and {\enummaxdcut} in which one aims to enumerate the $d$-cuts, the {\em minimal} $d$-cuts and the {\em maximal} $d$-cuts, respectively.
We consider various structural parameters of the input graph such as its vertex cover number, neighborhood diversity, and clique partition number.
When vertex cover number (${\vc}$) and neighborhood diversity (${\ndd}$) are considered as the parameter, we provide polynomial-delay enumeration kernels of polynomial size for {\enumdcut} and {\enummaxdcut} and fully-polynomial enumeration kernels of polynomial size for {\enummindcut}.
When clique partition number (${\clp}$) is considered as the parameter, we provide bijective enumeration kernels for each of {\enumdcut}, {\enummindcut}, and {\enummaxdcut}.
\end{abstract}



\newpage

\section{Introduction}
\label{sec:intro}

%
Kernelization \cite{CyganFKLMPPS15,DowneyF13,FominLSZ19} is a central subfield of parameterized complexity that was initially explored for decision problems.
The objective of a kernelization algorithm (or kernel) is to shrink the input instance in polynomial time into an equivalent instance  the size of which is dependent only on the parameter.
The effectiveness of a kernelization (or a preprocessing algorithm) is determined by the size bound of the kernel.
Later, several extensions of kernelization have been adopted for parameterized optimization problems \cite{LokshtanovPRS17}, counting problems \cite{JansenS23,LokshtanovM0Z24}, and enumeration problems \cite{CreignouMMSV17,Damaschke06,GolovachKKL22}.
\iflong In the case of parameterized enumeration problems, a small kernel can be viewed as a ``compact representation'' of all  solutions of the original instance. 
The first attempt was suggested by Damaschke \cite{Damaschke06} where the aim was to shrink the input instance to an output instance containing all minimal solutions of size at most $k$.
This definition, however, is limited to subset minimization problems.
Observing this as a drawback, Creignou et al.~\cite{CreignouKMMOV19,CreignouMMSV17,CreignouKPSV19} introduced a modified notion which asks for two algorithms:
\fi
The first algorithm is the {\em kernelization algorithm} that shrinks the input to an instance (the kernel) whose size is bounded by the parameter.
The second algorithm is the {\em solution-lifting algorithm} that given a solution of the kernel, outputs a collection of solutions to the input instance. 
The definition of Creignou et al.~\cite{CreignouKMMOV19,CreignouMMSV17,CreignouKPSV19} allows the solution-lifting algorithm to run in FPT-delay.

Golovach et al.~\cite{GolovachKKL22} observed that allowing FPT-delay is, however, too weak: In this case, the solution-lifting algorithm may  {\em ignore} the kernel when outputting all solutions to the input instance which implies that many problems with FPT enumeration algorithms have kernels of constant size. Consequently, the instance produced by the kernel is not necessarily a representation of the original instance.  To remedy this, Golovach et al.~\cite{GolovachKKL22} introduced two {\em refined} notions of enumeration kernels: {(i)} {\em fully-polynomial enumeration kernels} and {(ii)} {\em polynomial-delay enumeration kernels} where the solution-lifting algorithms must run in polynomial time and polynomial delay, respectively (see Definitions \ref{defn:poly-delay-enum-kernel}).
  This definition has several desirable consequences~\cite{GolovachKKL22}: First an enumeration problem admits a polynomial-delay (or polynomial-time) enumeration algorithm if and only if it admits a polynomial-delay (fully-polynomial) enumeration kernel of constant size. This is analogous to the fact that a problem admits a polynomial-time algorithm if and only if it admits a kernel of constant size. 
  Second, a parameterized enumeration problem admits an FPT-delay algorithm if and only if it admits a polynomial-delay enumeration kernel. This \iflong automatic enumeration\fi kernel however has exponential size\iflong with respect to the parameter\fi.
As for decision problems, this begs the following question: if a parameterized enumeration problem admits an FPT-delay algorithm, does it admit a polynomial-delay enumeration kernel of {\em polynomial size}?
Such a kernel gives us a \emph{compact} representation of all solutions of the original instance, and a guarantee that a collection of `equivalent' solutions for each kernel solutions can be enumerated very efficiently. In this work we aim to extend the set of enumeration problems for which such kernelizations are known. 

Motivated by the work on polynomial-delay enumeration kernelization for {\MC} by Golovach et al.~\cite{GolovachKKL22} and the work on standard kernelization for the more general {\dcut} problem by Gomes and Sau \cite{GomesS21}, we initiate a systematic study of the {\dcut} problem from the perspective of enumeration kernelizations, with the focus of providing fully-polynomial enumeration kernelizations and polynomial-delay enumeration kernelizations. The {\dcut} problem can be formalized as follows.
Given a graph $G = (V, E)$, a {\em cut} is a partition $(A, B)$ of $V(G)$.
A cut $(A, B)$ is a {\em $d$-cut} of $G$ if every vertex of $A$ has at most $d$ neighbors in $B$ and every vertex of $B$ has at most $d$ neighbors in $A$. For a fixed integer $d \geq 1$, the decision version of  {\dcut} is now defined as follows.

\defdecproblem{{\dcut}}{A simple undirected graph $G = (V, E)$.}{Is there a cut $(A, B)$ of $G$ such that every vertex $u \in A$ has at most $d$ neighbors in $B$ and every vertex $u \in B$ has at most $d$ neighbors in $A$?}

For the enumeration variants of {\dcut}, it is more meaningful to characterize $d$-cuts via \emph{edge cuts}. For two disjoint vertex sets~$A$ and~$B$ in $G$, let~$E_G(A,B)$ denote the set of edges with one endpoint in~$A$ and the other endpoint in~$B$. If~$(A,B)$ is a cut, then $E_G(A, B)$ is an {\em edge cut} and an edge cut~$E_G(A,B)$ is a~$d$-cut if every vertex of~$V(G)$ is an endpoint of at most~$d$ edges of~$E_G(A,B)$. This viewpoint of the problem allows to reason about minimality of and maximality of cuts: A $d$-cut $(A, B)$ is {\em (inclusion) minimal} if there is no $d$-cut $(A', B')$ such that $E_G(A',B') \subsetneq E_G(A, B)$.
Similarly, a $d$-cut $(A, B)$ of $G$ is {\em (inclusion) maximal} if there does not exist any $d$-cut $(A', B')$ of $G$ such that $E_G(A, B) \subsetneq E_G(A', B')$.
  This gives rise to the following three enumeration variants of {\dcut}\iflong denoted as {\enummindcut}, {\enummaxdcut}, and {\enumdcut}, respectively\fi.

\iflong 
\defproblem{{\enumdcut}}{A simple undirected graph $G = (V, E)$.}{Enumerate all the $d$-cuts of $G$.}

\defproblem{{\enummindcut}}{A simple undirected graph $G = (V, E)$.}{Enumerate all minimal $d$-cuts of $G$.}

\defproblem{{\enummaxdcut}}{A simple undirected graph $G = (V, E)$.}{Enumerate all maximal $d$-cuts of $G$.}
\fi

We underline that in our definition, two $d$-cuts $(A, B)$ and $(A', B')$ are {\em identical} if their edge cuts are the same, that is, if $E_G(A, B) = E_G(A', B')$, and  {\em distinct} otherwise.
If we put $d = 1$, then the {\dcut} is precisely the {\MC} problem.
Both {\MC} and {\dcut} are NP-complete problem even in special classes of graphs \cite{Bonsma09,ChenHLLP21,Chvatal84,KratschL16,LeL19,LeT22,LuckePR22,PatrignaniP01} and have been explored from the perspective of parameterized complexity \cite{AravindKK17,AravindKK22,AravindS21,GomesS21,KomusiewiczKL20}.


\subparagraph{Our Contributions and Motivations Behind the Choice of Parameters.}
\todo{C: This should be updated at the very end. }
We apply both types of enumeration kernelizations to {\enummindcut}, {\enumdcut}, and {\enummaxdcut}.
As {\dcut} is NP-hard \cite{Chvatal84}, it is unlikely that {\enumdcut} and other variants admit polynomial-delay enumeration algorithms.
The structural parameters considered in this paper are the size of a minimum vertex cover~(${\vc}$), the neighborhood diversity~(${\ndd}$), and the clique-partition number~(${\clp}$) of the input graph.
The motivation behind the choice of these parameters is 
that none of these problems admit enumeration kernelizations of polynomial size when parameterized by a combined parameter of {\em solution size}, {\em maximum degree}, and {\em treewidth of the input graph} (see \cite{GomesS21}).
In fact, polynomial-size kernels are also excluded for parameterization by the vertex-integrity of the input graph (see Proposition~\ref{prop:no-poly-kernel-treedepth}).
Hence, it is natural to select structurally larger parameters such as minimum vertex cover size (${\vc}$) and neighborhood diversity (${\ndd}$).
\iflong
\begin{figure}[t]
\centering
	\includegraphics[scale=0.26]{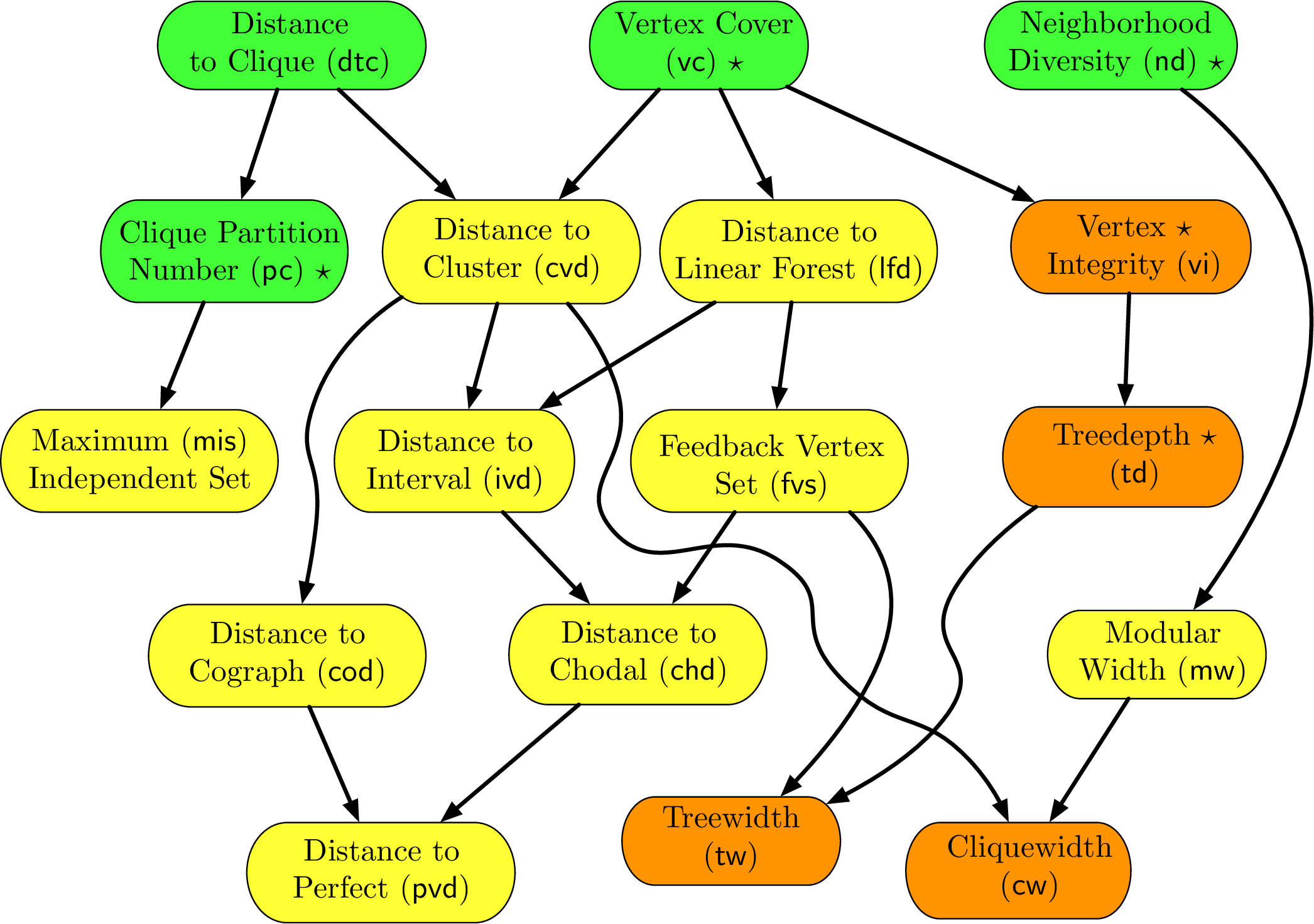}
	\caption{An overview of parameter hierarchy. 
          A green box shows the existence of (polynomial-delay) enumeration kernelizations of polynomial size. An orange box indicates non-existence of enumeration kernelizations of polynomial size unless {\nka}. A pink box indicates the NP-hardness for a constant parameter value. A yellow box indicates an open status of the existence of (polynomial-delay) enumeration kernelizations of polynomial size.
	The results of this paper are marked with a $\star$.}
\label{fig:parameter-hierarchy} 
\end{figure} \fi
Refer to \iflong Figure \ref{fig:parameter-hierarchy} for an illustration of parameters and to \fi Table~\ref{table-our-results-summary} for a short summary of our results.

Our main results for {\enummindcut}, {\enumdcut}, and {\enummaxdcut} are for the parameterization by  the vertex cover number (${\vc}$) of the input graph.
We prove these three results in Section~\ref{sec:improved-vc-kernels}.

\begin{restatable}{theorem}{ThmVCminimal}
\label{thm:min-d-cut-vc-enum}
For every fixed $d \geq 1$, {\enummindcut} parameterized by ${\vc}$ admits a fully-polynomial enumeration kernel with $\OO(d {\vc}^2)$ vertices.
\end{restatable}

\begin{restatable}{theorem}{ThmVCalldCuts}
\label{thm:all-d-cut-new}
For every fixed $d \geq 1$, {\enumdcut} admits a polynomial-delay enumeration kernel with $\OO(d {\vc}^2)$ vertices.
\end{restatable}

\begin{restatable}{theorem}{ThmVCOne}
\label{thm:vc-result-main-1}
For every fixed $d \geq 1$, 
	{\enummaxdcut}
	parameterized by~${\vc}$ admits polynomial-delay enumeration kernel with $\OO(d^3 {\vc}^{d+1})$ vertices.
\end{restatable}

\iflong
Next, we consider ${\ndd}$ the neighborhood diversity of the input graph, as parameter and prove the following result in Section~\ref{sec:neighborhood-diversity-modular-width}.
\fi

\begin{restatable}{theorem}{ThmNDComb}
\label{thm:nd-result-main-1}
For every fixed $d \geq 1$,
{\enummindcut} parameterized by ${\ndd}$ admit a fully-polynomial enumeration kernel with $\OO(d^2 {\ndd})$ vertices.
Moreover, {\enumdcut} and {\enummaxdcut} admit polynomial-delay enumeration kernels with $\OO(d^2 {\ndd})$ vertices.
\end{restatable}

\iflong
Finally, in Section \ref{sec:clique-partition}, we consider the clique partition number of the input graph, denoted ${\clp}$, as the parameter and prove the following result.
\else
\fi

\begin{restatable}{theorem}{ThmCLP}
\label{thm:clique-partition-result}
For every fixed $d \geq 1$,
{\enumdcut}, {\enummindcut}, and {\enummaxdcut} parameterized by ${\clp}$ admit bijective enumeration kernels with $\OO({\clp}^{d+2})$ vertices.	
\end{restatable}


\begin{table}[t]\caption{Summary of our results. Here, ``delay enum-kernel'' means polynomial-delay enumeration kernelization, ``full enum-kernel'' means fully-polynomial enumeration kernelization, and ``bi-enum-kernel'' means bijective enumeration kernelization. A ($\star$) means the lower bound assumes ${\rm NP \not\subseteq coNP/poly}$. The results marked as ($\star \star$) are the results proved in this paper.}

\begin{tabular}{ |p{2.2cm}||p{3.7cm}|p{3.7cm}|p{3.7cm}|  }
 \hline
 {Parameter $k$} & {\enumdcut} & {\enummindcut} & {\enummaxdcut}\\
 \hline
 vertex cover number ($\star \star$)   & $\OO(d {k}^{2})$-vertex delay enum-kernel (Theorem \ref{thm:all-d-cut-new}) & $\OO(d {k}^{2})$-vertex full enum-kernel (Theorem \ref{thm:min-d-cut-vc-enum}) &   $\OO(d^3 {k}^{d+1})$-vertex delay enum-kernel (Theorem \ref{thm:vc-result-main-1}) \\
 \hline
 neighborhood diversity ($\star \star$) &  $\OO(d^2 k)$-vertex delay enum-kernel (Theorem \ref{thm:nd-result-main-1})  & $\OO(d^2 k)$-vertex full enum-kernel (Theorem \ref{thm:nd-result-main-1}) & $\OO(d^2 k)$-vertex delay enum-kernel (Theorem \ref{thm:nd-result-main-1}) \\
 \hline
clique partition number ($\star \star$) & $\OO({k}^{{d}+2})$-vertex bijective enum-kernel (Theorem \ref{thm:clique-partition-result}) & $\OO({k}^{{d}+2})$-vertex bi-enum-kernel (Theorem \ref{thm:clique-partition-result}) &  $\OO({k}^{d+2})$-vertex bijective enum-kernel (Theorem \ref{thm:clique-partition-result}) \\
\hline
 vertex-integrity ($\star\star$)  & No polynomial size delay enum-kernel ($\star$) (Proposition \ref{prop:no-poly-kernel-treedepth})& No polynomial size delay enum-kernel ($\star$) (Proposition \ref{prop:no-poly-kernel-treedepth}) &  No polynomial size delay enum-kernel ($\star$) (Proposition \ref{prop:no-poly-kernel-treedepth}) \\
 \hline
\end{tabular}
\label{table-our-results-summary}
\end{table}

\subparagraph{Further Related Work.} 
We refer to a survey by Wasa \cite{Wasa16arXiv} for a detailed overview of enumeration algorithms.
Bentert et al. \cite{BentertFNN19} introduced a notion of {\em advice enumeration kernels}, an extension to the enumeration problem by Creignou et al. \cite{CreignouMMSV17} where the solution-lifting algorithm does not need to known the whole input but rather only some possibly smaller advice. Another application of polynomial-delay enumeration kernels was given recently for structural parameterizations of the \textsc{Long Path} problem~\cite{KMS25}.   
\iflong Counting problems are a different type of extension of decision problems that can be viewed as an intermediary problem between decision and enumeration.
Thurley~\cite{Thurley07} defined compactors, a kernelization analog for  parameterized counting problems. Compactors also have a solution-lifting algorithm that computes for each solution of the kernel a number of corresponding solutions in the input instance. This is quite similar to enumeration kernels as described above.
Thilikos~\cite{Thilikos21} surveys some results that were obtained for compactors.  
Lokshtanov et al.~\cite{LokshtanovM0Z24} recently introduced a different framework of kernelization for counting problems where the solution-lifting algorithm computes the number of solutions of the input instance directly from the number of the solutions of the kernel. For counting problems, this is desirable, since there is no need to explicitly visit all kernel solutions as one has to do for compactors. In the enumeration setting, however, visiting all solutions cannot be avoided.
\fi


\section{Preliminaries}
\label{sec:prelims}


\iflong \subsection{Sets, Numbers, and Graph Theory}
\fi
\label{sec:prelim-part-1}

We use $[r]$ to denote $\{1,\ldots,r\}$ for some $r \in \nn$ and $A \uplus B$ to denote the disjoint union of the sets $A$ and $B$.
For a set $A$, we use $2^{A}$ to denote the {\em power set} of $A$.
For a finite set $A$ and a nonnegative integer $d \geq 0$, we use ${{A}\choose{d}}, {{A}\choose{\leq d}}$, and ${{A}\choose{\geq d}}$ to denote the collection of all subsets of $A$ of size equal to $d$, at most $d$, and at least $d$ respectively.
We use standard graph-theoretic notations from the book of Diestel~\cite{Diestel-Book}.
Throughout this paper, we consider simple undirected graphs.
For a subset $X \subseteq V(G)$, $G[X]$ denotes the subgraph induced by the vertex subset $X$.
Similarly, $G - X$ denotes the graph obtained after deleting the vertex set $X$; i.e. $G - X = G[V(G) \setminus X]$; for a vertex $u$, we use $G - u$ to represent $G - \{u\}$ for the simplicity.
Similarly, for a set of edges $F \subseteq E(G)$ (respectively for $e \in E(G)$), we write $G - F$ (respectively $G - e$) to denote the graph obtained after deleting the edges of $F$ (respectively after deleting $e$) from $G$.
Equivalently, for an edge $F \subseteq E(G)$, we use $G - F$ to denote the subgraph $G(V, E \setminus F)$
For a vertex $v$, we denote by $N_G(v)$ the {\em open neighborhood} of $v$ in $G$, i.e. the set of vertices that are adjacent to $v$ in $G$ and by $N_G[v]$, the {\em closed neighborhood} of $v$ in $G$.
Formally, $N_G[v] = N_G(v) \cup \{v\}$.
Given a simple graph $G = (V, E)$, a {\em cut} is a bipartition of its vertex set $V(G)$ into two nonempty sets $A$ and $B$, denoted by $(A, B)$.
The set of all edges with one endpoint in $A$ and the other in $B$ is called the {\em edge cut} or the set of {\em crossing edges} of $(A, B)$.
We use $F = E_G(A, B)$ to denote the edge cut of $(A, B)$.
If the graph is clear from the context, we omit the subscript.
\todo[inline]{Diptapriyo: Some of these lines can be omitted.}

A set of vertices $S \subseteq V(G)$ is called {\em vertex cover} if for every $uv \in E(G)$, $u \in S$ or $v \in S$ (or both).
The size of a minimum vertex cover of a graph is known as the {\em vertex cover number} and is denoted by ${\vc}$.
%
Given a graph $G$, a {\em neighborhood decomposition} of $G$ is a partition $\UU = \{X_1,\ldots,X_k\}$ of $V(G)$ such that every set~$X_1$ is a clique or an independent set and all the vertices of~$X_i$ have the same neighborhood outside of~$X_i$.
The {\em neighborhood diversity} of $G$ is the size of a smallest neighborhood decomposition of $G$.
A partition $\CC = \{C_1,\ldots,C_k\}$ of a graph $G$ is said to be a {\em clique partition} of $G$ if for all $1 \leq i \leq k$, $C_i$ is a clique.
The {\em clique partition number} of $G$ is denoted by ${\clp}(G)$ is the minimum $k$ such that $G$ has a clique partition with $k$ cliques.
The {\em vertex-integrity number} of $G$, denoted by ${\vi}(G)$ is the minimum number~$k$ such that $G$ has a vertex set~$X$ of size at most~$k'\le k$ with the largest connected component in~$G-X$ having size~$k-k'$.

\iflong \subparagraph*{Attaching a Pendant Clique.}\fi
We make use of the following graph operation.
Let $G$ be an undirected graph such that $u \in V(G)$.
We define the graph operation of {\em attaching a pendant clique of $r$ vertices to $u$} as follows: add a clique $C$ of $r$ vertices disjoint from~$V(G)$ to $G$ such that for every $v \in C$, $N(v) \setminus C = \{u\}$.
In short form, we often refer to this graph operation as {\em attaching an $r$-vertex pendant clique $C$ to $u$}.

\iflong \subsection{Some Fundamental Properties of $d$-Cuts}
\fi
\label{sec:prelim-part-2}

A set $T \subseteq V(G)$ is {\em monochromatic} if for every $d$-cut $(A, B)$ of $G$, either $T \subseteq A$ or $T \subseteq B$.
The following lemma gives fundamental structural characterizations of a $d$-cuts of $G$. Some of them, were already observed by Gomes and Sau~\cite{GomesS21}. We give a proof of all of them for the sake of completeness.

\begin{lemma}
\label{lem:neighborhood-break-up}
Let $G$ be an undirected graph.
\begin{enumerate}[(i)]
	\item\label{it:clique-mono} If $C$ is a clique of $G$ with at least $2d+1$ vertices, then $C$ is monochromatic.
	
	\item\label{it:overlap-mono} If $X$ and $Y$ are monochromatic sets with a nonempty intersection, then $X \cup Y$ is monochromatic. 
    
    \item\label{it:d-neighbor-mono} If a vertex~$u$ has more than~$d$ neighbors in a monochromatic set~$X$, then~$X\cup \{v\}$ is monochromatic.
	
	\item\label{it:common-neighborhood-mono} Let $T$ be a subset of $V(G)$ such that there are more than $2d$ vertices in $V(G) \setminus T$ whose neighborhoods contain $T$.
	Then, $T$ is monochromatic.
	
\end{enumerate}
\end{lemma}
\begin{proof}
We prove the items in the given order.
It is sufficient to prove each of the statement for an arbitrary $d$-cut $(A, B)$ of $G$.
\begin{enumerate}[(i)]
	\item 
	Since $C$ has at least $2d+1$ vertices, we may assume without loss of generality that~$|C \cap A| \geq d+1$. Now, if~$C\cap B$ contains a vertex~$u$, then this vertex has at least~$d+1$ neighbors in~$C \cap A$, contradicting the fact that $(A, B)$ is a~$d$-cut. Hence,~$C\subseteq A$.
	\item Since $X$ is monochromatic, we may assume without loss of generality~$X\subseteq A$. Since there is a vertex~$u\in X\cap Y$ and since~$Y$ is monochromatic, we have~$Y\subseteq A$ and thus~$X\cup Y\subseteq A$.
        \item Since~$X$ is monochromatic, we may assume~$X\subseteq A$. Then,~$u$ has at least~$d+1$ neighbors in~$A$. Consequently,~$u$ is contained in~$A$ as well.  
	\item Let~$I$ denote the vertices in $V(G)\setminus T$ whose neighborhood contains $T$. Since~$|I|>2d+1$, we may assume that  there are at least $d+1$ vertices in $I\cap A$. Hence, every vertex in~$T$ has more than~$d$ neighbors in~$A$. Consequently,~$T\subseteq A$.
\end{enumerate}
\end{proof}

\todo[inline]{C: In the end, check the references to the lemma and see whether the items are still correctly numbered, replace hard-coded item number by ref-commands }
As a consequence of item (\ref{it:overlap-mono}) of Lemma \ref{lem:neighborhood-break-up}, maximal monochromatic sets are disjoint.
Thus, we obtain the following.

\begin{observation}
\label{obs:monochromatic-partition}
Let ${\mch}(S)$ be the collection of inclusion-maximal subsets of $S$ that are monochromatic.
Then, ${\mch}(S)$ forms a partition of $S$.
\end{observation}

We also state the following observation that holds true for a minimal $d$-cut of any graph.

\begin{observation}
\label{obs:minimal-obs-trivial}
If $G$ is a disconnected graph, then $\emptyset$ is the unique minimal $d$-cut of $G$.
If $G$ is connected and has a vertex $v$ that has degree at most $d$, then $(\{v\}, V(G) \setminus \{v\})$ is a minimal $d$-cut of $G$.
\end{observation}

\iflong \subsection{Parameterized Complexity, Kernelization, and Parameterized Enumeration}
\label{sec:prelims-part-3}

\subparagraph{Parameterized Complexity and Kernelization.}
A {\em parameterized problem} is denoted by $L \subseteq \Sigma^* \times \nn$ where $\Sigma$ is a finite alphabet.
An instance of a parameterized problem is denoted $(x, k) \in \Sigma^* \times \nn$ where and $k$ is considered as the {\em parameter}.
A parameterized problem $L$ is said to be {\em fixed-parameter tractable} if there is an algorithm $\cA$ that takes an input instance $(x, k) \in \Sigma^* \times \nn$, runs in $f(k)|x|^{\OO(1)}$-time and correctly decides if $(x, k) \in L$ for a computable function $f: \nn \rightarrow \nn$.
The algorithm $\cA$ is called {\em fixed-parameter algorithm} (or FPT algorithm).
An important step to design parameterized algorithms is kernelization (or parameterized preprocessing).
Formally, a parameterized problem $L \subseteq \Sigma^* \times \nn$ is said to admit a {\em kernelization} if there is a polynomial-time procedure that takes $(x, k) \in \Sigma^* \times \nn$ and outputs an instance $(x',k')$ such that {\bf (i)} $(x,k) \in L$ if and only if $(x',k') \in L$, and {\bf (ii)} $|x'| + k' \leq g(k)$ for some computable function $g: \nn \rightarrow \nn$.
This computable function $g(\cdot)$ is the {\em size} of the kernel.
If $g(\cdot)$ is a polynomial function, then $L$ is said to admit a {\em polynomial kernel}.
Compression of a parameterized problem is a generalized notion of kernelization of a parameterized problem.
Formally, a parameterized problem $L \subseteq \Sigma^* \times \nn$ is said to admits a {\em compression} into $L'$ if there is a polynomial-time preprocessing algorithm that transforms an instance $(x, k)$ of a parameterized problem $L$ into an instance $y$ of problem $L' \subseteq \Sigma^*$ such that {\bf (i)} $y \in L'$ if and only if $(x, k) \in L$, and {\bf (ii)} $|y| \leq g(k)$ for some computable function $g: \nn \rightarrow \nn$.
If $g(\cdot)$ is a polynomial function, then $L$ is said to admit a {\em polynomial compression}.
It is well-known that a decidable parameterized problem is FPT if and only if it admits a kernelization \cite{CyganFKLMPPS15}.
For further details on details on parameterized complexity and kernelization, we refer to \cite{CyganFKLMPPS15,DowneyF13}.
 \fi


\subparagraph{Parameterized Enumeration and Enumeration Kernelization.}
We use the framework for parameterized enumeration that was proposed by Creignou et al. \cite{CreignouMMSV17}.
An {\em enumeration problem} over a finite alphabet $\Sigma$ is a tuple $(L, {\sol})$ such that
\begin{enumerate}[(i)]
	\item $L\subseteq \Sigma^*$ is a decision problem, and
	\item ${\sol}: \Sigma^* \rightarrow 2^{\Sigma^*}$ is a computable function such that ${\sol}(x)$ is a nonempty finite set if and only if $x \in L$.
\end{enumerate}
Here $x \in \Sigma^*$ is an instance, and ${\sol}(x)$ is the set of solutions to instance $x$. Observe that~$L$ is decidable since~$\sol$ is a computable function.
A {\em parameterized enumeration problem} is defined as a triple $\Pi = (L, \sol, \kappa)$ such that $(L, {\sol})$ satisfy the same as defined in item (i) and (ii) above. 
In addition to that $\kappa: \Sigma^* \rightarrow \nn$ is the {\em parameter}.
We define here the parameter as  a computable function $\kappa(x)$\iflong; it is natural to assume that the parameter is given with the input or $\kappa(x)$ can be computed in polynomial time\fi.
An {\em enumeration algorithm} $\cA$ for a parameterized enumeration problem $\Pi$ is a deterministic algorithm that given $x \in \Sigma^*$, outputs ${\sol}(x)$ exactly without duplicates and terminates after a finite number of steps. 
If $\cA$ outputs exactly ${\sol}(x)$ without duplicates and eventually terminates in $f(\kappa(x))|x|^{\OO(1)}$-time, then $\cA$ is called an {\em FPT-enumeration algorithm}.
For $x \in L$ and $1 \leq i \leq |{\sol}(x)|$, the $i$-th {\em delay} of $\cA$ is the time taken between outputting the $i$-th and $(i+1)$-th solution of ${\sol}(x)$.
The $0$-th delay is the {\em precalculation} time that is the time from the start of the algorithm until the first output.
The $|{\sol}(x)|$-th delay is the {\em postcalculation} time that is the time from the last output to the termination of $\cA$.
If the enumeration algorithm $\cA$ on input $x \in \Sigma^*$, outputs ${\sol}(x)$ exactly without duplicates such that every delay is $f(\kappa(x))|x|^{\OO(1)}$, then ${\cA}$ is called an {\em FPT-delay enumeration algorithm}.

\begin{definition}
\label{defn:poly-delay-enum-kernel}
Let $\Pi = (L, {\sol}, \kappa)$ be a parameterized enumeration problem.
A {\em polynomial-delay enumeration kernel(ization)} for $\Pi$ is a pair of algorithms $\cA$ and $\cA'$ such that.
\begin{enumerate}[(i)]
	\item\label{delay-enum-prop-1} For every instance $x$ of $\Pi$, the {\em kernelization algorithm} $\cA$ computes in time polynomial in $|x| + \kappa(x)$ an instance $y$ of $\Pi$ such that $|y| + \kappa(y) \leq f(\kappa(x))$ for a computable function~$f$.
	\item\label{delay-enum-partition-prop} For every $s \in {\sol} (y)$, the {\em solution-lifting algorithm} $\cA'$ computes with delay in polynomial in $|x| + |y| + \kappa(x) + \kappa(y)$ a nonempty set of solutions $S_s \subseteq {\sol} (x)$ such that $\{S_s \mid s \in {\sol} (y)\}$ is a partition of ${\sol} (x)$.
\end{enumerate}
We call~$f$ the {\em size} of the polynomial-delay enumeration kernel(ization).
If $f$ is a polynomial function, then $\Pi$ is said to admit a polynomial-delay enumeration kernel(ization) of {\em polynomial size}. 
\end{definition}

Observe that by (ii) of the definition above, $x \in L$ if and only if $y \in L$.
In addition, Property (ii) implies that ${\sol}(x) \neq \emptyset$ if and only if ${\sol}(y) \neq \emptyset$.
An enumeration kernel is {\em bijective} if for every $s \in {\sol}(y)$, the solution-lifting algorithm produces a unique solution to $\hat s \in {\sol}(x)$ giving a bijection between ${\sol}(y)$ and ${\sol}(x)$.

The notion of {\em fully-polynomial enumeration kernel(ization)} is defined in a similar way.
The condition (i) from Definition \ref{defn:poly-delay-enum-kernel} remains the same and the only difference is in condition (ii) which is replaced by the following condition.
\begin{description}
	\item[(ii*)] For every $s \in {\sol} (y)$, the {\em solution-lifting algorithm} $\cA'$ computes in time polynomial in $|x| + |y| + \kappa(x) + \kappa(y)$ a nonempty set of solutions $S_s \subseteq {\sol} (x)$ such that $\{S_s \mid s \in {\sol} (y)\}$ is a partition of ${\sol} (x)$.
\end{description}
In the condition (ii), the solution-lifting algorithm for a fully-polynomial enumeration kernel runs in polynomial time but the solution-lifting algorithm for a polynomial-delay enumeration kernelization runs in polynomial delay.
Every bijective enumeration kernelization is a fully-polynomial enumeration kernelization and every fully-polynomial enumeration kernelization is a polynomial-delay enumeration kernelization \cite{GolovachKKL22}. 
\iflong Let $\Pi = (L, {\sol}, \kappa)$ be a parameterized enumeration problem. 
Golovach et al. \cite{GolovachKKL22}  proved that {\bf (i)} $\Pi$ admits an FPT-enumeration algorithm (an FPT-delay enumeration algorithm) if and only if $\Pi$ admits a fully-polynomial enumeration kernel (polynomial-delay enumeration kernel), and {\bf (ii)} in fact, $\Pi$ can be solved in polynomial time (or with polynomial delay) if and only if $\Pi$ admits a fully-polynomial (polynomial-delay) enumeration kernelization of constant size.
Notice that the enumeration kernelizations that can be inferred from the above does not necessarily have polynomial size.
\fi

\iflong
  \subsection{Nonexistence of Enumeration Kernels for $\mathbf{d}$-Cut.}
\fi
\label{sec:prelim-part-5}

We now show some lower bounds for enumeration kernels for the considered problems. First, we show the nonexistence of enumeration kernels of polynomial size for the vertex-integrity of the input graph. \iflong No tailored lower-bound framework is yet known for parameterized enumeration problems. However, by the definition of enumeration kernelization, the kernelization algorithm also is a kernelization algorithm for the corresponding decision problem. 
Thus, if a parameterized enumeration problem $\Pi = (L, {\sol}, \kappa)$ admits a polynomial-delay enumeration kernel of polynomial size, then its decision version $L$ admits a polynomial kernel when parameterized by $\kappa$.
Therefore, if for some a parameterized enumeration problem $\Pi = (L, {\sol}, \kappa)$ the decision problem $L$ parameterized by $\kappa$ does not admit a polynomial kernel unless {\nka}, then $\Pi$ does not admit a polynomial-delay/fully-polynomial enumeration kernel of polynomial size unless {\nka}.
We thus directly get the following.

\begin{proposition}
\label{prop:dec-kernel-enum-kernel}
Let $\Pi = (L, {\sol}, \kappa)$ be a parameterized enumeration problem.
If $\Pi$ admits a polynomial-delay enumeration kernelization of size $f(\kappa)$, then $L$ parameterized by $\kappa$ also admits a kernel of size $f({\kappa})$.
\end{proposition}

Gomes and Sau \cite{GomesS21} showed that, unless {\nka},  {\dcut} parameterized solely by the treewidth (${\sf tw}$) or by cliquewidth of the input graph does not admit a polynomial kernel.
Therefore, the following statement holds true.

\begin{lemma}[\cite{GomesS21}]
\label{prop:no-poly-kernel-treewidth}
For any fixed $d \geq 1$,
{\enumdcut}, {\enummaxdcut} and {\enummindcut} do not admit any polynomial-delay enumeration kernelizations of polynomial size when parameterized by the treewidth (or cliquewidth respectively), unless {\nka}.
\end{lemma}

We pass on a short proof illustrating that {\dcut} also does not admit a polynomial kernelization when parameterized by the vertex-integrity~(${\sf vi}$) of the graph.  The proof is based on the framework of OR-cross composition introduced in \cite{BodlaenderJK14}.
In the following, we define the framework for the sake of completeness.

\begin{definition}[Polynomial Equivalence Relation]
\label{defn:poly-eqv-relation}
An equivalence relation $R$ on $\Sigma^*$ is a {\em polynomial equivalence relation} if the following two conditions hold:
\begin{enumerate}[(i)]
	\item Given two strings $x, y \in \Sigma^*$, there is an algorithm that correctly decides whether $x$ and $y$ belong to the same equivalence class with respect to the equivalence relation $R$ in $(|x| + |y|)^{\OO(1)}$-time, and
	\item for any finite set $Q \subseteq \Sigma^*$, the equivalence relation partitions the elements of $Q$ into at most $(\max_{x \in Q} |x|)^{\OO(1)}$ many equivalence classes.
\end{enumerate}
\end{definition}

Now, we formally state the definition of OR-cross composition that was introduced by Bodlaender et al. \cite{BodlaenderJK14} as follows.

\begin{definition}[OR-cross composition]
\label{defn:OR-cross-composition}
Let $L_1 \subseteq \Sigma^*$ be a classical problem and $L_2$ be a parameterized problem.
The problem $L_1$ {\em cross composes} into $L_2$ if there is a polynomial equivalence relation $R$ on $\Sigma^*$ and an algorithm which given $2^t$ instances $x_1,\ldots,x_{2^t} \in \Sigma^*$, belonging to the same equivalence class of $R$, computes an instance $(x^*, k^*) \in \Sigma^* \times \nn$ in time polynomial in $\sum\limits_{i=1}^{2^t} |x_i|$ such that:
\begin{enumerate}[(i)]
	\item $(x^*, k^*) \in L_2$ if and only if there exists $i \in [2^t]$ such that $x_i \in L_1$, and
	\item $k^*$ is bounded by polynomial in $(t + \max\limits_{i=1}^{2^t} |x_i|)$.
\end{enumerate}
\end{definition}

With the definitions above, Bodlander et al. \cite{BodlaenderJK14} proved the following.

\begin{proposition}[\cite{BodlaenderJK14}]
\label{prop:cross-compose-no-kernel}
If $L_1 \subseteq \Sigma^*$ is an NP-hard problem and cross composes into a parameterized problem $L_2 \subseteq \Sigma^* \times \nn$, then $L_2$ parameterized by $k^*$ does not admit a polynomial compression unless {\nka}.
\end{proposition}

We provide a cross composition of {\dcut} to itself as follows.

\begin{lemma}
\label{lemma:d-cut-cross-compose}
For any fixed $d \geq 1$, {\dcut} does not admit a polynomial compression when parameterized by the vertex-integrity of the graph, unless {\nka}.
\end{lemma}

\begin{proof}
Let $G_1,\ldots,G_{2^t}$ denote $2^t$ instances of {\dcut} all having $n$ vertices and each being connected graph.
Without loss of generality, we assume that $t > n + 2d$.
We say that two graphs $G_i$ and $G_j$ are equivalent if they have the same number of vertices.
Clearly, this is a polynomial equivalence relation and all $G_1,\ldots,G_{2^t}$ appear in the same equivalence class.
We construct a graph $G$ from $G_1,\ldots,G_{2^t}$ as follows.
For every $i \in [2^t]$, we choose $u_i \in V(G_i)$ arbitrarily.
Add a clique $C$ with exactly $2d$ vertices and for every $i \in [2^t]$, make $u_i$ adjacent to all vertices of $C$.
This completes the construction of $G$ and we set $k = {\sf td}(G)$.
Note that $G$ can be constructed in time polynomial in $\sum\limits_{i=1}^{2^t} |V(G_i)|$.

First, observe that ${\sf vi}(G)$, the vertex-integrity of $G$ is at most $2d + |V(G_i)| = 2d + n$.
Note that $2d + n$ is upper-bounded by a polynomial in $t + n$ as $t > 2d$.
It remains to argue that there exists $i \in [2^t]$ such that $G_i$ has a $d$-cut if and only if $G$ has a $d$-cut. What is crucial here is that $G$ is a connected graph.

For the forward direction ($\Rightarrow$), let $G_i$ has a $d$-cut $(A_i, B_i)$ for some $i \in [2^t]$.
As $G_i$ is connected, $(A_i, B_i)$ is nonempty $d$-cut.
We extend $(A_i, B_i)$ into a $d$-cut $(A, B)$ of $G$ as follows.
Initialize $A := A_i$ and $B := B_i$.
If $u_i \in A_i$, then we add $C$ into $A$.
Subsequently, for every $j \neq i$, we add all vertices of $G_j$ into $A$.
Otherwise, $u_i \in B_i$, and we add $C$ into $B$.
Subsequently, for every $j \neq i$, we add all vertices of $G_j$ into $B$.
Observe that there is no edge between the vertices of $A_i$ and $A_j$ (respectively between the vertices of $B_i$ and $B_j$, and the vertices of $A_i$ and $B_j$), this ensures us that $(A, B)$ is a $d$-cut of $G$.

Conversely, for the backward direction ($\Leftarrow$), let $(A, B)$ be a $d$-cut of $G$.
As $G$ is connected, $(A, B)$ is a nonempty $d$-cut.
Note that for every $i \in [2^t]$, $C \cup \{u_i\}$ is a clique with $2d+1$ vertices.
Due to the item (\ref{it:clique-mono}) of Lemma \ref{lem:neighborhood-break-up}, $C \cup \{u_i\} \subseteq A$ or $C \cup \{u_i\} \subseteq B$.
\todo[inline]{Diptapriyo: actually this follows both because of item (\ref{it:clique-mono}) and (\ref{it:d-neighbor-mono}) of Lemma \ref{lem:neighborhood-break-up} Referring one of them should be good enough, I hope.}
If $C \cup \{u_i\} \subseteq A$, then for every $j \in [2^t]$, $C \cup \{u_{j}\} \subseteq A$.
{As for every $i \in [2^t]$, $C \cup \{u_i\}$ is monochronatic, due to item (\ref{it:overlap-mono}) of Lemma \ref{lem:neighborhood-break-up}, $(\cup_{i \in [2^t]} \{u_i\}) \cup C$ is monochromatic.
As $(\cup_{i \in [2^t]} \{u_i\}) \cup C$ is monochromatic, and} $(A, B)$ is a nonempty $d$-cut of $G$, there must be some $j \in [2^t]$ such that $V(G_j) \cap A, V(G_j) \cap B \neq \emptyset$.
\todo{Diptapriyo: added an extra argument line as otherwise, I had a feel that I was handwaving before.}
As $G_j$ is a subgraph of $G$, the vertex bipartition $(A \cap V(G_i), B \cap V(G_j))$ is a $d$-cut of $G_j$.
Therefore, there exists $j \in [2^t]$ such that $G_j$ has a $d$-cut.

The above construction thus satisfies Definition~\ref{defn:OR-cross-composition}, and {\dcut} is an NP-hard problem.
Hence, due to Proposition \ref{prop:cross-compose-no-kernel}, {\dcut} parameterized by the vertex-integrity of the graph does not admit a polynomial compression unless {\nka}.
\end{proof}

Note that vertex-integrity of a graph is always at least as large as its treedepth. 
Using the above lemma and Proposition \ref{prop:dec-kernel-enum-kernel}, we thus have the following. \fi
\begin{proposition}
\label{prop:no-poly-kernel-treedepth}
For every fixed $d \geq 1$, {\enumdcut}, {\enummaxdcut}, and {\enummindcut} when parameterized ${\vi}$ or by ${\sf td}$ and do not admit polynomial-delay enumeration kernels of polynomial size unless {\nka}.
\end{proposition}

For enumeration kernels, we can also exclude fully-polynomial enumeration kernels by showing that the number of solutions is too large. We do this for two parameterizations considered in this work, the vertex cover number~$\vc(G)$ and the neighborhood diversity~$\nd(G)$, and the problems of enumerating all or all maximal $d$-cuts.
\iflong 
\begin{proposition}
  \label{lem:dcut-sol-vc-nd}
  For each constant~$d\ge 1$ and each number~$k>2d$, there are arbitrarily large graphs~$G$ with~$((n-k)/k)^{k}$ maximal $d$-cuts for~$k= {\vc}(G)$ and~$k = {\ndd}(G)-1$ where~$c>0$ is some fixed constant.   
\end{proposition}

\begin{proof}
Consider the following family of graphs~$G_n$,~$n\in \mathbb{N}$ with~$(n-k)\mod k=0$. Each graph~$G_n$ contains a clique~$S$ of size~$k$. For each~$u\in S$,~$G_n$ contains~$(n-k)/k$ degree-one vertices to~$G$ that are adjacent to~$u$. Observe that~${\vc}(G)=k$ and~${\ndd}(G)=k+1$.
  Then, the maximal $d$-cuts of this graph contain for each~$u\in S$ exactly~$d$ edges between~$u$ and its degree-one neighbors. The total number of maximal~$d$-cuts is thus~$\binom{(n-k)/k)}{d}^k \ge ((n-k)/k)^k$.   
\end{proof}

Now, the above size bound grows roughly like~$n^k$. In contrast, problems with fully-polynomial enumeration kernels have at most~$f(k)\cdot n^{\OO(1)}$ solutions  for some function~$f(k)$ as discussed above. Hence, we obtain the following. 
\fi
\begin{proposition}
\label{cor:no-fully-poly-enum-kernels}
  {\enumdcut} and {\enummaxdcut} parameterized by ${\vc}(G)$ or ${\ndd}(G)$ do not admit fully-polynomial enumeration kernels of any size. 
\end{proposition}

\section{Parameterization by the Vertex Cover Number}
\label{sec:d-cut-vc}

In this section, we provide a fully-polynomial enumeration kernel for {\enummindcut} and a polynomial-delay enumeration kernel for {\enumdcut}. For both kernels, we assume that the input consists of the graph~$G$ together with a vertex cover $S$ of size at most $2{\vc}(G)$. We can make this assumption without loss of generality since such a vertex cover can be computed in linear time~\cite{Savage82}.
Let $|S| = k \leq 2{\vc}(G)$ and $I = V(G) \setminus S$.

\subsection{Enumerating Minimal $d$-Cuts}
\label{sec:new-enum-min-d-cut}

\subparagraph{Marking Scheme and Kernelization Algorithm.} 
If the input graph $G$ is disconnected, then our kernelization algorithm outputs $H = 2K_1$ as the output graph.
Otherwise, we proceed as follows.
As $S$ is a vertex cover of $G$, every vertex of $I$ has at least one neighbor in $S$ since~$G$ is connected.
We partition $I = I_1 \uplus I_2$ such that $I_1$ denotes the vertices of degree exactly one, and $I_2$ denotes the vertices of degree at least two.
Our marking scheme works as follows.
\begin{enumerate}
	\item[(i)] For every $x \in S$, mark an arbitrary vertex of $N_G(x) \cap I_1$.
	\item[(ii)] For every $x, y \in S$, choose an arbitrary set of $\min\{2d+1, |N_G(x) \cap N_G(y) \cap I_2|\}$ vertices from $I_2$ that are adjacent to both $x$ and $y$ and mark them for the pair $\{x, y\}$.
\end{enumerate}

If~$I$ contains an unmarked vertex~$z$ of degree at most~$d$, we do the following.
First, we check if~$I$ contains a marked vertex~$u$ of degree at most $d$.
If this is the case, then attach a pendant clique $C$ of $2d+2$ vertices to $u$ and mark them.
Otherwise, we choose some $x \in N_G(z) \cap S$, attach a pendant neighbor $y$ to $I$, mark~$y$, and attach a pendant clique $C$ to~$y$.
Let $Z$ be the set of all marked vertices from $I$ and $H = G[S \cup C \cup Z]$.
Output~$H$.
This completes the kernelization algorithm. Clearly, the algorithm has a polynomial running time. The main idea of the kernelization algorithm is that we remove vertices of~$I$ as long as we guarantee that their neighborhood is monochromatic. This is formalized as follows.
\begin{lemma}\label{lem:removed-mono}Assume that~$G$ is connected and let~$u$ be a vertex in~$V(G)\setminus V(H)$, then~$N_G(u)$ is monochromatic in~$G$ and in~$H$.
\end{lemma}
\begin{proof}
  We show that for every pair of vertices~$v$ and~$w$ in $N_G(u)$, the set~$\{v,w\}$ is monochromatic. 
  Since~$u$ is not in~$V(H)$, the vertices~$v$ and~$w$ have at least~$2d+1$ common neighbors in~$G$ and, by the construction of the kernel, also in~$H$. Thus, by item~(\ref{it:common-neighborhood-mono}) of Lemma~\ref{lem:neighborhood-break-up},~$\{v,w\}$ is monochromatic in~$G$ and in~$H$.  
\end{proof}
We may observe the following bound on the size of the kernel.

\begin{observation}
\label{obs:new-minimal-vc-H-size}
If $G$ is disconnected, then $H$ is isomorphic to $2K_1$.
If $G$ is connected, then $H$ is connected and has $\OO(d\cdot {\vc}^2)$~vertices.
\end{observation}

\begin{proof}
The running time and the claim for disconnected~$G$ follow directly from the description of the marking scheme. 
In case $G$ is connected, then $|Z| \leq {{|S|}\choose{2}}(2d+1) + |S| + 1$ since~$S$ has size at most~$2\cdot \vc$. The clique $C$, it it is added, has size~$2d+2$. 
Therefore, $H$ has $\OO(d {\vc}^2)$ vertices and has a vertex cover $C \cup S$ of size at most $|S| + 2d+2$.
Moreover, $G$ is connected since the marking scheme preserves the connectivity between every pair of vertices of the vertex cover $S$ that is adjacent or has some common neighbors and since every vertex of~$Z$ has a neighbor in~$S$.
\end{proof}

\subparagraph{Equivalence Classes of Minimal $d$-Cuts.} It is not hard to see that $H$ does not keep the information of all minimal $d$-cuts of $G$. For example, when some vertex~$v$ of~$I_1$ is not contained in~$H$, then the kernel does not contain a cut directly corresponding to the minimal $d$-cut~$(\{v\},V\setminus \{v\})$.
To establish a relation between the minimal $d$-cuts of $G$ and the minimal $d$-cuts of $H$, we define an equivalence relation as follows.

Let $F_1$ and $F_2$ be two edge sets of a connected graph $G$ (respectively, of $H$).
We say that $F_1$ and $F_2$ are {\em equivalent} if $F_1 = F_2$ or there are two distinct vertices $u, v \in I$ of degree at most~$d$ such that $F_1 = \{uw \in E(G) \mid w \in N_G(u)\}$ and $F_2 = \{vw \in E(G) \mid w \in N_G(v)\}$.
Observe that the above defined relation is an equivalence relation.
The same notion of equivalence relation can be defined between an edge set $F_1 \subseteq E(G)$ and an edge set $F_2 \subseteq E(H)$ as well.
The following observation holds true as per the definition of the equivalence relation.

\begin{observation}
\label{obs:minimal-vc-eqv-d-cut-new}
Let $F_1$ and $F_2$ be two edge sets of $G$ (respectively, two edge sets of $H$, or $F_1 \subseteq E(G)$ and $F_2 \subseteq E(H)$) that are equivalent to each other.
Then, $F_1$ is a $d$-cut of $G$ (respectively, of $H$), if and only if $F_2$ is a $d$-cut of $G$ (respectively, of $H$).	
\end{observation}

\subparagraph{Challenges to Avoid Duplicate Enumeration.} We need to ensure that our solution-lifting algorithm does not enumerate duplicate $d$-cuts.
To do this, we define the notion of distinguished minimal $d$-cut of $H$ as follows.
Let $F' = E_H(A', B')$ be a minimal $d$-cut of~$H$.
We say that $(A', B')$ is {\em distinguished} if $A' = C \cup \{v\}$ for some $v \in I$, and $B' = V(H) \setminus A'$ such that $C$ is a pendant clique of $2d+1$ vertices attached to $v$ and is disjoint from $G$.
The polynomial-time solution-lifting algorithm now works as follows. 
If the given minimal $d$-cut $F' = E_H(A', B')$ of $H$ is distinguished, then the solution-lifting algorithm outputs $F'$ and all minimal $d$-cuts of $G$ that are not minimal $d$-cuts of $H$.
For every other minimal $d$-cut $F' = E_H(A', B')$ of $H$ that is not distinguished, the solution-lifting algorithm only outputs $F'$.
In the next lemma, we formally prove that we can design such an enumeration algorithm when a distinguished minimal $d$-cut of $H$ is given as part of the input.

\begin{lemma}
\label{lemma:new-distinguished-minimal-d-cut-enumeration}
Given a distinguished minimal $d$-cut $F' = E_H(A', B')$ of $H$, there exists a polynomial-time algorithm that outputs all minimal $d$-cuts of $G$ that are not minimal $d$-cuts of $H$. 
\end{lemma}

\begin{proof}
Let $F' = E_H(A', B')$ be a distinguished $d$-cut of $H$.
Then, $A' = C \cup \{v\}$ such that $C$ is a pendant clique having $2d+1$ vertices that is attached to $v$ and is disjoint from~$G$.
Since $C \cup \{v\}$ has at least $2d+1$ vertices, due to item (i) of Lemma \ref{lem:neighborhood-break-up}, $C \cup \{v\}$ is monochromatic.
We consider every vertex $u \in I \setminus Z$ one by one.
If the degree of $u$ is more than $d$, then we move to the next vertex.
Otherwise, the degree of $u$ is at most $d$, and we output the cut $(\{u\},V(G) \setminus \{u\})$.

First, observe that every cut~$(A,B)$ output by the algorithm is a minimal $d$-cut: Since the degree of $u$ is at most $d$,~$(A,B)$ is a~$d$-cut. Morover, by Lemma~\ref{lem:removed-mono},~$N_G(u)$ is monochromatic.  Consequently, $(\{u\}, V(G) \setminus \{u\})$ is the only minimal $d$-cut of $G$ that contains edges incident with~$u$. Hence, $(A,B)$ is a minimal $d$-cut of $G$.

It remains to argue that every minimal $d$-cut of $G$ that is not a $d$-cut of $H$ is output by the above algorithm. Any such~$d$-cut contains at least one edge incident with some vertex $u \in I \setminus Z$.
Again, by Lemma~\ref{lem:removed-mono}, $N_G(u)$ is monochromatic. Thus, the only minimal cut containing edges incident with~$u$ is $(\{u\}, V(G) \setminus \{u\})$ is the only minimal $d$-cut of $G$. If~$u$ has degree at most~$d$, then we output this cut. Otherwise, when $u$ has degree more than $d$, then there is no~$d$-cut containing edges incident with~$u$.
Therefore, our algorithm outputs every $d$-cut of $G$ that is not a $d$-cut of $H$.
\end{proof}

Using the above lemma, we are ready to prove our result.

{\ThmVCminimal*}

\begin{proof}
Our enumeration kernel has two parts.
The first part is the kernelization algorithm and the second part is the solution-lifting algorithm.
Let $G$ be the input graph and $S$ be a vertex cover of $G$ such that $|S| = k \leq 2{\vc}(G)$.
The kernelization algorithm invokes the \textit{marking scheme} and outputs the graph $H$.
It follows from Observation \ref{obs:new-minimal-vc-H-size} that $H$ has $\OO(d|S|^2) = \OO(d {\vc}^2)$ vertices.

Our solution-lifting algorithm works as follows.
If $H$ is a disconnected graph, then due to Observation \ref{obs:new-minimal-vc-H-size}, $H$ is isomorphic to $2K_1$ and $G$ is disconnected as well.
Clearly, $\emptyset$ is the unique minimal $d$-cut of both $G$ and $H$.
Hence, given a minimal $d$-cut $\emptyset$ of $H$, the solution-lifting algorithm outputs $\emptyset$.

If $G$ is connected, then due to Observation \ref{obs:new-minimal-vc-H-size}, $H$ is also connected.
Therefore, every minimal $d$-cut of $H$ is nonempty.
First, we check whether the given minimal $d$-cut $(A', B')$ of $H$ is distinguished or not.
We check whether $A' = C \cup \{v\}$ and $B' = V(H) \setminus A'$ such that $C$ is a pendant clique of $2d+1$ vertices attached to $v$ and $C$ is disjoint from the graph $G$ or not.
If $A' = C \cup \{v\}$ such that $C$ is a pendant clique of $2d+1$ vertices attached to $v$ and $C$ is disjoint from the graph $G$, then first output $(A', B')$ and then invoke Lemma \ref{lemma:new-distinguished-minimal-d-cut-enumeration} to output all minimal $d$-cuts of $G$ that are not minimal $d$-cuts of $H$ in polynomial time.
If $(A', B')$ is not distinguished, then output $(A', B')$.
It is not hard to see that no minimal $d$-cut of $G$ is outputted for two distinct minimal $d$-cuts of $H$.
This completes the proof of a fully-polynomial enumeration kernel with $\OO(d {\vc}^2)$ vertices for {\enummindcut} parameterized by the vertex cover number.
\end{proof}

\subsection{Enumerating all $d$-Cuts}
\label{sec:new-enum-all-d-cuts}

In this section, we describe a polynomial-delay enumeration kernelization for {\enumdcut}.
While the kernelization is almost the same, the solution-lifting algorithm will be very different from the one in {\enummindcut}.
As above, we assume that the input instance consists of $G$ and a vertex cover $S$ such that $|S| \leq 2{\vc}(G)$ and denote $I = V(G) \setminus S$.
We partition $I = I_0 \uplus I_1 \uplus I_2$ such that $I_0$ denotes the set of all isolated vertices, $I_1$ denotes the vertices that are pendants, and $I_2$ denotes the vertices degree at least two in $G$.

\subparagraph{Marking Scheme and Kernelization Algorithm.} 
Our marking scheme works as follows.
\begin{enumerate}
	\item Mark an arbitrary vertex from $I_0$.
	\item For every $x \in S$, mark an arbitrary vertex from $N_G(x) \cap I_1$.
	\item For every $\{x, y\} \in {{S}\choose{2}}$, choose an arbitrary set of $\min\{2d+1, |N_G(x) \cap N_G(y) \cap I_2|\}$ vertices from $I_2$ that are adjacent to both $x$ and $y$, and mark those vertices for the pair $\{x, y\} \in {{S}\choose{2}}$.
\end{enumerate}

Let $Z$ be the set of all marked vertices from $I$ and $H = G[S \cup Z]$ be the output graph which is the kernel.
This completes the description of the kernelization algorithm. As above, we see that the removed vertices have a monochromatic neighborhood.
\begin{lemma}\label{lem:removed-mono-2}Let~$u$ be a vertex in~$V(G)\setminus V(H)$, then~$N_G(u)$ is monochromatic in~$G$ and in~$H$.
\end{lemma}

We have the following size bound.

\begin{observation}
\label{obs:kernel-size-vc-all-d-cut-enum-new}
Let $H$ be the graph obtained from $G$ after invoking the above marking scheme.
Then, $H$ has $\OO(d\cdot |S|^2)$ vertices and at most one isolated vertex.
\end{observation}

\begin{proof}
It is clear from the marking scheme that at most one isolated vertex is marked from $G$.
Hence, $H$ has at most one isolated vertex.
Since for every unordered pair of vertices $(x, y)$ from $S$, at most $2d+1$ common neighbors of $x$ and $y$ from $I$ are marked.
Therefore, $|Z| \leq (2d+1){{|S|}\choose{2}}$.
Since $H = G[S \cup Z]$, it follows that $H$ has $\OO(d\cdot |S|^2)$ vertices.
\end{proof}

\subparagraph{Equivalence Classes of $d$-Cuts.} Let $F'$ be the edge set of a $d$-cut $(A', B')$ of $H$.
Let $S_A' = S \cap A'$, $S_B' = S \cap B'$, $I_A' = Z \cap A'$, and $I_B' = Z \cap B'$.
Therefore, for every $u \in I \setminus Z$, it holds that either $N_G(u) \subseteq A'$ or $N_G(u) \subseteq B'$.
Let $J_A = \{u \in I \setminus Z \mid N_G(u) \subseteq A', N_G(u) \neq \emptyset\}$ and $J_B = \{u \in I \setminus Z \mid N_G(u) \subseteq B', N_G(u) \neq \emptyset\}$.
Observe that $(A' \cup J_A, B' \cup J_B)$ is a $d$-cut of $G$.
In particular, the edge cut of $(A', B')$ and $(A' \cup J_A, B' \cup J_B)$ are the same.
Hence, if the edge cut of some $d$-cut $(A, B)$ is a proper superset of the edge cut of $(A', B')$, then there exist two sets  $Q \subseteq J_A$ and $P \subseteq J_B$ such that $(A' \cup P \cup (J_A \setminus Q), B' \cup Q \cup (J_B \setminus P))$ is a $d$-cut and at least one of $P$ and $Q$ is nonempty. Informally,~$J_A$ and~$J_B$ are those vertices outside of the kernel that will create additional cut edges by being on a different side of the cut than their monochromatic neighborhood.

We now want to characterize those sets $Q \subseteq J_A$ and $P \subseteq J_B$ that give a~$d$-cut. To this end, for every $u \in S_A'$, let $h(u)$ be the number of neighbors of $u$ that are in $B'$.
Analogously, for every $u \in S_B'$, let $h(u)$ be the number of neighbors of $u$ that are in $A'$.
Now if $(A' \cup P \cup (J_A \setminus Q) , B' \cup Q) \cup (J_B \setminus P))$ is a $d$-cut of $G$, then it must be that  adding $P$ to $A'$ and adding $Q$ to $B'$ do not violate the $d$-cut property for the cut $(A, B)$.
We need to enumerate all pairs $(P, Q)$ that satisfies the above properties.
To formalize this, we say that a pair of sets $(P, Q) \in 2^{J_A} \times 2^{J_B}$ is {\em legal} with respect to a $d$-cut $(A', B')$ of $H$ if the following conditions are satisfied:
\begin{enumerate}[(i)]
	\item every vertex of $P$ has at most $d$ neighbors in $B'$,
	\item every vertex of $Q$ has at most $d$ neighbors in $A'$,
	\item every $u \in S_A'$ has at most $d - h(u)$ neighbors in $Q$, and
	\item every $u \in S_B'$ has at most $d - h(u)$ neighbors in $P$.
\end{enumerate}

Finally, we call a nonempty cut $(A, B)$ of $G$ a {\em legal extension} of a nonempty $d$-cut $(A', B')$ of $H$ if $A = A' \cup P \cup (J_A \setminus Q)$ and $B = B' \cup Q \cup (J_B \setminus P)$ such that $(P, Q)$ is a legal pair with respect to $(A', B')$. Observe that~$E(A',B')\subseteq E(A,B)$  and that every edge of $E(A,B)\setminus E(A',B')$ is incident with at least one vertex that is not in~$A'\cup B'$.
Based on the definition of legal extension, the following lemma holds true.

\begin{lemma}
\label{lemma:legal-extension-property}
Let $(A', B')$ be a $d$-cut of $H$ and $(A, B)$ be a cut of $G$ that is a legal extension of $(A', B')$.
Then, $(A, B)$ is a $d$-cut of $G$.
Conversely, for every $d$-cut $(A, B)$ of $G$, there exists a unique $d$-cut $(A', B')$ of $H$ such that $(A, B)$ is a legal extension of $(A', B')$.
\end{lemma}
\begin{proof}
We give the proof of the items in the given order.
For the first item, let $(A', B')$ be a $d$-cut of $H$ and $(A, B)$ be a cut of $G$ that is a legal extension of $(A', B')$.
Let $S_A' = S \cap A'$, $S_B' = S \cap B'$, $I_A' = Z \cap A'$, and $I_B' = Z \cap B'$.
Since $(A', B')$ is a $d$-cut of $H$, every $u \in A'$ (respectively, every $u \in B'$) has at most $d$ neighbors in~$B'$ (respectively, in $A'$).

Since $(A, B)$ is a legal extension of $(A', B')$, there exist $P \subseteq J_B$ and $Q \subseteq J_A$ such that every $u \in S_A'$ has at most $d - h(u)$ neighbors in $Q$ and every $u \in S_B'$ has at most $d - h(u)$ neighbors in $P$.
Moreover, $A = A' \cup P \cup (J_A \setminus Q)$ and $B = B' \cup Q \cup (J_B \setminus P)$.
Observe that $u \in S_A'$ has $h(u)$ neighbors in $B'$ and at most $d - h(u)$ neighbors in $Q$ due to condition (iii).
A vertex $u \in S_A'$ does not have any neighbor in $J_B \setminus P$ as the vertices of $J_B$ have neighbors only in $B'$.
Therefore, $u \in A'$ has at most $d$ neighbors in~$B$.
Similarly, consider any $u \in S_B'$.
As $(A, B)$ is legal extension, $u$ must have at most $d - h(u)$ neighbors in $P$ by condition (iv) and exactly $h(u)$ neighbors in $A'$.
Since the vertices of $J_A \setminus Q$ have neighbors only in $A'$, therefore, $u \in S_B'$ has at most $d$ neighbors in~$A$.
Finally, by  condition~(i) in the definition of legal extension, every vertex in $P$ has at most $d$ neighbors in $B'$ (hence in $B$) and every vertex of $J_A \setminus Q$ has only neighbors in $A'$, therefore, every $u \in A$ has at most $d$ neighbors in~$B$.
Similarly, by condition~(ii) in the definition of legal extension, every vertex in $Q$ has at most $d$ neighbors in $A'$ (hence in $A$) and every vertex of $J_B \setminus P$ has only neighbors in $B'$, therefore, every $u \in B$ has at most $d$ neighbors in $A$.
Since $A \cap B = \emptyset$ by construction and $A \cup B = V(G)$, hence, it follows that $(A, B)$ is a $d$-cut of $G$.

For the second item, let $(A, B)$ be a $d$-cut of $G$. If the edge cut of~$(A,B)$ is empty, then~$H$ also contains an empty cut which is equivalent to~$(A,B)$. Hence, assume that this is not the case.   
Consider $A' = A \cap V(H)$ and $B' = B \cap V(H)$.
Since $H$ is an induced subgraph of $G$ and $(A, B)$ is a $d$-cut, we have that $(A', B')$ is a $d$-cut of $H$. 
Clearly, $(A, B)$ is a legal extension of $(A', B')$. Moreover, it cannot be an extension of another $d$-cut of~$H$ as the only additional edges in the edge cut of $(A, B)$ compared to $(A',B')$ are incident with vertices of~$V(G)\setminus V(H)$.
\end{proof}

Now, we define the notion of ``equivalence'' between two $d$-cuts of $G$ and one $d$-cut of $G$ with one $d$-cut of $H$.
Two distinct nonempty $d$-cuts $(A,B)$ and $(\widehat A, \widehat B)$ of $G$ are {\em equivalent} if both $(A, B)$ and $(\widehat A, \widehat B)$ are legal extensions of a some $d$-cut $(A', B')$ of $H$.
Additionally, a $d$-cut $(A, B)$ of $G$ is {\em equivalent} to a $d$-cut $(A', B')$ of $H$ if either both $E_G(A, B) = \emptyset$ and $E_H(A', B') = \emptyset$ or $(A, B)$ is a legal extension of $(A', B')$.
\todo[inline]{Diptapriyo: should we define this with respect to only $d$-cuts? or with respect to any edge set? C: I would think: only $d$-cuts.}

\subparagraph{Challenges and Designing the Solution-Lifting Algorithm.}
Now, we discuss how we design the solution-lifting algorithm and how we address the challenges in related to that.
Our solution-lifting needs to ensure the following aspects.
\begin{enumerate}
	\item Given a $d$-cut $(A', B')$ of $H$, our solution-lifting algorithm must output all legal extensions of $(A', B')$, with polynomial-delay.
	
	\item For every $d$-cut $(A, B)$ of $G$, there exists a unique $d$-cut $(A', B')$ of $H$ such that $(A, B)$ is outputted by the solution-lifting algorithm when $(A', B')$ is given.
\end{enumerate}

The first challenge is to ensure that no $d$-cut of $G$ is omitted and the second challenge is to ensure that no $d$-cut of $G$ is output for two distinct $d$-cuts of $H$.
Now, our solution-lifting algorithm works as follows to address the above-mentioned challenges.

\begin{algorithm}[t]
   \DontPrintSemicolon
\SetKwFunction{enumeqvt}{{\sc EnumEqvt}}
\KwData{$d$-Cut~$(A^*,B^*)$ in~$H$,\; $J_A \leftarrow \{u \in I \setminus Z \mid N_G(u) \subseteq A^*, N_G(u) \neq \emptyset\}$,\;
	$J_B \leftarrow \{u \in I \setminus Z \mid N_G(v) \subseteq B^*, N_G(u) \neq \emptyset\}$\; }
\SetKwProg{Fn}{}{}{}

\Fn{\enumeqvt{$P, Q, D$}}{
	
	
        \SetKw{KwOutput}{output}
	\KwOutput $E(A^* \cup P \cup (J_A \setminus Q), B^* \cup Q \cup (J_B \setminus P))$ \tcp*{current edge cut}
	\ForEach{$v \in V(G)\setminus (V(H)\cup P \cup Q \cup D)$}{
			\If{$v \in J_B$ and $(P \cup \{v\}, Q)$ is legal with respect to  $(A^*, B^*)$}{
				{\enumeqvt}$(P \cup \{v\}, Q, D)$\;
			}
			\If{$v \in J_A$ and $(P, Q \cup \{v\})$ is legal with respect to   $(A^*, B^*)$}{
				{\enumeqvt}$(P, Q \cup \{v\}, D)$\;
			} 
			$D \leftarrow D \cup \{v\}$ \tcp*{Now consider cuts where~$v$ remains with neighbors }
		}
              }
            
\caption{Solution-Lifting Algorithm for {\enumdcut}}
\label{alg:two-new-all-d-cut}
\end{algorithm}

\begin{lemma}
\label{lemma:all-d-cut-enum-eqv-d-cut}
Let $(A^*, B^*)$ be a $d$-cut of $H$.
Then, there is an enumeration algorithm that enumerates all $d$-cuts of $G$ that are legal extensions of $(A^*, B^*)$ without duplicates with polynomial-delay.
\end{lemma}

\begin{proof}
Our enumeration algorithm {\enumeqvt} is recursive backtracking algorithm which maintains three sets $P$,~$Q$, and~$D$. In a call {\enumeqvt$(P,Q,D)$}, we search for all legal extensions~$(P',Q')$ of $(A^*, B^*)$ such that~$P\subseteq P'$ and~$Q\subseteq Q'$ are vertices from~$J_B$ and~$J_A$ and~$P'$ and~$Q'$ avoid~$D$.
We initialize $D$ to be the set of all vertices from $I \setminus V(H)$ that have degree more than $d$.
In addition, we initialize $P = Q = \emptyset$.
In other words, the initial call is {\enumeqvt}$(\emptyset, \emptyset, D)$.

Observe that the initial definition of~$D$ is correct in the sense that every legal extension must avoid~$D$: Every vertex~$u$ of~$D$ has a monochromatic neighborhood of size at least~$d+1$.
Hence, $u$ cannot be added to $A^*$ (respectively, to $B^*$) if $N_G(u) \subseteq B^*$ (respectively, if $N_G(u) \subseteq A^*$).

At the first step of {\enumeqvt}($P, Q, D$), we output the edge cut of the current $d$-cut $(A^* \cup P \cup (J_A \setminus Q), B^* \cup Q \cup (J_B \setminus P))$.
Next, we consider one-by-one the vertices~$v \notin D \cup P \cup Q$. For each such~$v$, we check whether it can be added to~$P$ or~$Q$, that is, we check whether  
\begin{enumerate}[(i)]
	\item $v \in J_B$ and $(P \cup \{v\}, Q)$ is legal with respect to $(A^*, B^*)$, or
	\item $(P, Q \cup \{v\})$ is legal with respect to $(A^*, B^*)$ and $v \in J_A$.
\end{enumerate}
If (i) is satisfied, then $(A^* \cup P \cup \{v\} \cup (J_A \setminus Q), B^* \cup Q \cup (J_B \setminus P))$ is a legal extension of $(A^*, B^*)$.
Then, the enumeration algorithm recursively calls itself with~$(P \cup \{v\}, Q, D)$.

If (ii) is satisfied, then $(A^* \cup P \cup (J_A \setminus Q), B^* \cup Q \cup \{v\} \cup (J_B \setminus P))$ is a legal extension of $(A^*, B^*)$.
Then, the enumeration algorithm recursively calls itself with $(P, Q \cup \{v\}, D)$.
If neither of the above situations occur but a vertex $v \notin D \cup P \cup Q$ exists, then the algorithm recursively calls itself with $(P, Q, D \cup \{v\})$.
When there does not exist any vertex outside $P \cup Q \cup D$ but in $J_A \cup J_B$ that are of degree at most $d$, then the algorithm terminates.

{\bf Delay Bound:} First, observe that at each node of the tree defined by the recursive calls, the total time spent at that node (without the time for the recursive calls) is polynomial in~$n$. Moreover, for each recursive call, one~$d$-cut is output. Thus, to obtain a bound on the delay it is sufficient to bound the number of nodes visited between to outputs.
Now, after a recursive call, the algorithm visits at most~$n$ nodes before either spawning a new recursive call or terminating, since the tree has depth at most~$n$.  Therefore, the algorithm has polynomial delay.

{\bf Enumeration without duplicates:} We have to ensure that no $d$-cut of $G$ that is a legal extension of $(A^*, B^*)$ is outputted twice.
Observe that every legal extension $(A, B)$ of $(A^*, B^*)$ satisfies that $A = A^* \cup P \cup (J_A \setminus Q)$ and $B = B^* \cup Q \cup (J_B \setminus P)$.
At every node of the tree, a legal extension $(A, B)$ of $(A^*, B^*)$ is output for a different pair~$(P, Q)$.
Hence, no legal extension of $(A^*, B^*)$ is output twice.
\end{proof}

{\ThmVCalldCuts*}


\begin{proof}
Our enumeration kernelization has two parts.
The first part is kernelization algorithm and the second part is solution-lifting algorithm.
We assume without loss of generality that $(G, S, k)$ be an input instance such that $|S| = k$ and $S$ is a vertex cover with at most $2{\vc}(G)$ vertices.
The kernelization algorithm invokes the {\em marking scheme} and outputs the instance $(H, S, k)$.
Due to Observation \ref{obs:kernel-size-vc-all-d-cut-enum-new}, $H$ has $\OO(k^2)$ vertices.

The solution-lifting algorithm works as follows.
Let $(A^*, B^*)$ be a $d$-cut of $H$.
We invoke Lemma \ref{lemma:all-d-cut-enum-eqv-d-cut} to output all $d$-cuts of $G$ that are legal extensions of $(A^*, B^*)$ without duplicates with delay in polynomial.
Additionally, due to Lemma \ref{lemma:legal-extension-property}, for every $d$-cut $(A, B)$ of $G$, there is a unique $d$-cut $(A', B')$ of $H$ such that $(A, B)$ is a legal extension of $(A', B')$.
Therefore, the property (ii) of Definition \ref{defn:poly-delay-enum-kernel} is satisfied.
This completes the proof that {\enumdcut} admits a polynomial-delay enumeration kernelization with $\OO(d {\vc}^2)$ vertices.
\end{proof}


\subsection{Enumerating Maximal $d$-Cuts}
\label{sec:improved-vc-kernels}

\iflong
In this section, we show an enumeration kernel for {\enummaxdcut}.
We again assume that the input is $(G, S, k)$  where $S$ is a vertex cover of $G$ such that $|S| = k \leq 2{\vc}(G)$ and denote $I = V(G) \setminus S$.
Let $I_{d}$ be the vertices of $I$ that has degree at most $d$ in $G$, and $I_{d+1}$ be the vertices of $I$ that has degree at least $d+1$ in $G$.
For every nonempty $T \in {{S}\choose{\leq d}}$, we define $\widehat{I_T}$, the set of all vertices from $I \cap V(G)$ each of which have neighborhood exactly $T$.
Formally, $\widehat{I_T} = \{u \in I_d \cap V(G) \mid N_G(u) = T\}$.

\subparagraph*{Overview of our Enumeration Kernelization:} In Section \ref{sec:vc-marking-scheme}, we describe the kernelization algorithm that involves a marking scheme to construct $\widetilde{H}$, a subgraph of $G$.
Then, we add some important structures to $\widetilde{H}$ by attaching pendant cliques to construct the output graph $H$ of the kernelization.
Next, in Section \ref{sec:vc-props-eqv-d-cuts}, we first prove that every $d$-cut of $H$ is a $d$-cut of $\widetilde{H}$ and vice versa.
Next, we establish some relations between the $d$-cuts of~$\widetilde{H}$ and the~$d$-cuts of $G$ by defining equivalence classes among them.
After that, in Section \ref{sec:vc-duplicate-enum}, we illustrate a challenge that can cause enumeration of a $d$-cut of $G$ that is equivalent to two distinct $d$-cuts of $H$.
Then, we provide a notion of `suitability' that circumvents this challenge.
Finally, in Section~\ref{sec:vc-final-algorithms}, we illustrate how we can use the previous lemmas to design a solution-lifting algorithm that eventually proves our main result Theorem~\ref{thm:vc-result-main-1} of this section.

\subsubsection{Marking Scheme and Kernelization Algorithm}
\label{sec:vc-marking-scheme}

In this section, we describe the kernelization algorithm.
We use the following marking procedure ${\sf MarkVC}(G, S, k)$.

\begin{figure}
\centering
	\includegraphics[scale=0.35]{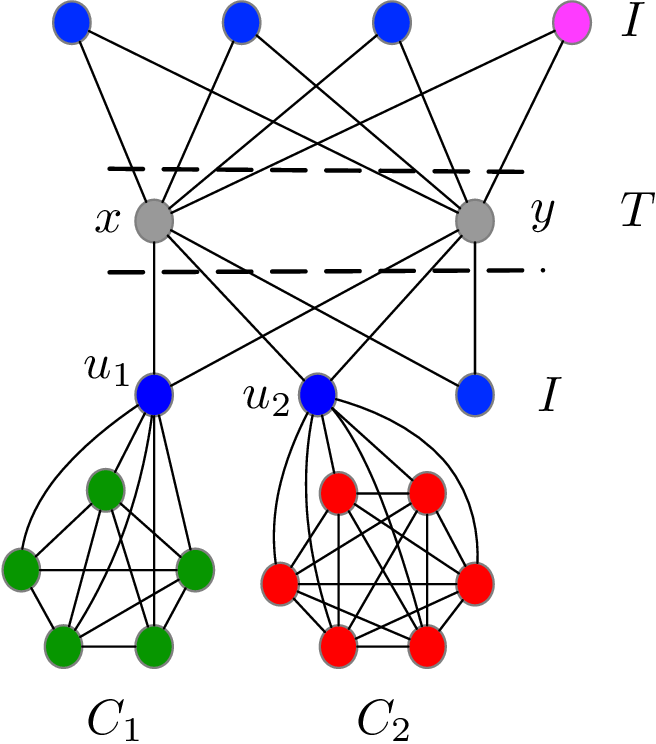}
	\caption{An illustration of attaching pendant cliques for $d = 2$ with the set $T = \{x, y\} \subseteq S$. The green, blue, and pink colored vertices are the vertices of $\widehat{I_T}$. The marked vertices from $\widehat{I_T}$ are colored blue and unmarked vertices are colored pink. The pendant clique $C_1$ with 5 vertices is attached to $u_1$ the pendant clique $C_2$ with 6 vertices is attached to $u_2$. The graph $H$ consists of all vertices except the pink vertex.}
\label{fig:vc-H-construction}	
\end{figure}

\subparagraph{\underline{Procedure ${\sf MarkVC}(G, S, k)$.}}
\begin{enumerate}
	\item Initialize $H^* := G$ and $S_0 := \emptyset$.
	\item Mark one arbitrary isolated vertex of $H^*$ if at least one exists.
	
	\item For every nonempty set $T \in {{S}\choose{\leq d}}$,
	 mark $\min\{2d+2, |\widehat{I_T}|\}$ vertices from $\widehat{I_T}$.
	If there is a vertex $w \in \widehat{I_T}$ that is unmarked, then we consider $d$ arbitrary marked vertices $u_1,\ldots,u_{d} \in \widehat{I_T}$.
	For every $i \in [d]$, attach a pendant clique $C_i$ of $2d+i$ vertices to $u_i$ and set $S_0 := S_0 \cup C_i$.
	We refer to Figure \ref{fig:vc-H-construction} for an illustration with $d = 2$ and $|T| = d$.
	
	\item For every nonempty set $T \in {{S}\choose{d+1}}$ mark $\min\{2d+1,|\bigcap\limits_{x \in T} N_G(x) \cap I_{d+1}|\}$ arbitrary vertices from $\bigcap\limits_{x \in T} N_G(x) \cap I_{d+1}$.
\end{enumerate}

We use $Z$ to denote the set of all {\em marked vertices} from $I$ and let $H = H^*[S \cup Z \cup S_0]$.
It is straightforward to see that $H$ can be constructed in polynomial time.
Every pendant clique~$C$ of at least $2d+1$ vertices and at most $3d$ vertices that has been attached to a few selected vertices of $Z$ are connected components of $H[S_0]$.
Hence, we often refer to these connected components of $H[S_0]$ as ``pendant cliques in $S_0$''.
The following observation illustrates an upper bound on the number of vertices in $H$.

\begin{observation}
\label{obs:size-H-bound}
The number of vertices in $S_0$ is $\OO(d^3 |S|^d)$ and every connected component of $H[S_0]$ is a clique with at least $2d+1$ and at most $3d$ vertices.
Moreover, $S \cup S_0$ is a vertex cover of $H$ and the total number of vertices in $H$ is at most $18d^3 |S|^{d+1}$.
\end{observation} 

\begin{proof}
Observe that the number of marked vertices from $I$ is at most $1 + (2d+2)\sum\limits_{i=1}^{d+1} {{|S|}\choose{i}}$.
As~$Z$ is the set of all marked vertices from $I$, there are $\OO(d^2 |S|^{d+1})$ vertices in $Z$.
In particular $|Z| \leq 6d^2 |S|^{d+1}$.
Moreover, for every $T \in {{S}\choose{d}}$, the construction ensures that $H$ has at most $2d + 2$ vertices from $\widehat{I_T}$ .
Out of these $2d+2$ vertices from $\widehat{I_T}$ there are at most $d$ vertices, each of which has a pendant clique attached to.
Each of these attached pendant cliques have at most $3d$ vertices and they are in $S_0$.
Hence, every connected component of $H[S_0]$ has at least $2d+1$ vertices and at most $3d$ vertices.
Therefore, the total number of vertices in $H$ is at most $18d^3 |S|^{d+1}$ and the number of vertices in $S_0$ is $\OO(d^3 |S|^d)$.
Also, as $S$ is a vertex cover of $G$ and by construction, $Z \subseteq I$, it follows that $S \cup S_0$ is a vertex cover of $H$.
\end{proof}


\subsubsection{Structural Properties and Equivalence Classes of $d$-Cuts}
\label{sec:vc-props-eqv-d-cuts}

Our marking scheme is sufficient to prove that $G$ has a $d$-cut if and only if $H$ has a $d$-cut.
To design an enumeration kernelization, it is not sufficient to prove that the output instance is equivalent to the input instance.
We need to establish that the information of all $d$-cuts of $G$ are available in the output graph $H$ which is the kernel.
Consider the graph $\widetilde{H} = H - S_0$ which is an induced subgraph of $G$.
This section is devoted to prove the following three crucial items.
\begin{enumerate}[(i)]
	\item First we prove that the collection of all $d$-cut of $H$ and the collection of all $d$-cuts of $\widetilde{H}$ are precisely the same.
	\item Then, we provide some specific relations and equivalence classes among the $d$-cuts of $\widetilde{H}$ and the $d$-cuts of $G$.
	These specific relations are crucial to establish that $\widetilde{H}$ (and hence $H$) has the information about all the $d$-cuts of $G$.
\end{enumerate}

Now, we move on to proving the first item.
Our next observation illustrates that for every pendant clique~$C$ that is attached to some vertex of $\widetilde{H}$, the set of edges incident to $C$ are disjoint from any $d$-cut of $H$.
This observation will be sufficient to prove that the first item above (the set of all $d$-cuts of $\widetilde{H}$ and the $d$-cuts of $H$ are precisely the same).
We use Lemma \ref{lem:neighborhood-break-up} and Observation \ref{obs:size-H-bound} to prove the following observation.

\begin{observation}
\label{obs:vc-S0-empty-H-intersection}
Let $F$ be a $d$-cut of $H$ and $uv \in E(H)$ such that $u, v \in S_0$ or $u \in S_0, v \in Z$, then $uv \notin F$.
\end{observation}

\begin{proof}
Let the premise of the statement be true and $uv \in E(H)$ such that either $u, v \in S_0$ or $u \in S_0, v \in Z$.
It follows from Observation \ref{obs:size-H-bound} that every connected component of $H[S_0]$ is a clique with at least $2d+1$ vertices and at most $3d$ vertices.
Furthermore, if $u, v \in S_0$ then $u, v \in C$ for a connected component $C$ of $H[S_0]$ that is a clique at least $2d+1$ vertices.
Additionally, every vertex of $S_0$ is adjacent to some vertex of $Z$ such that $Z \subseteq I \cap V(G)$.
In case $u \in S_0$ and $v \in Z$, then $u, v \in \widehat C$ that is also a clique with at least $2d+2$ vertices.
Let $(A, B)$ be the corresponding $d$-cut of $H$ such that $F = E_H(A, B)$. 
Since $(A, B)$ is a $d$-cut of $H$ and $uv$ is an edge of a clique $\widehat C$ containing at least $2d+1$ vertices, it follows from Lemma \ref{lem:neighborhood-break-up} that either $\widehat C \subseteq A$ or $\widehat C \subseteq B$.
Therefore, $uv \notin F$.
\end{proof}

The above observation implies that if an edge of $H$ is incident to the vertices of $S_0$, then that edge cannot be present in any $d$-cut of $H$.
We use Observation \ref{obs:vc-S0-empty-H-intersection} to prove the following statement, saying that the collection of all $d$-cuts of $H$ and the collection of all $d$-cuts of $\widetilde{H}$ are precisely the same.


\begin{lemma}
\label{lemma:vc-d-cuts-without-S0}
A set of edges $F$ is a (maximal) $d$-cut of~$\widetilde{H}$ if and only if $F$ is a (maximal) $d$-cut of~$H$.
\end{lemma}

\begin{proof}
We first give the backward direction ($\Leftarrow$) of the proof.
Let $(A, B)$ denote the bipartition such that $F = E_H(A, B)$ is a $d$-cut of $H$.
It follows from Observation \ref{obs:vc-S0-empty-H-intersection} that if an edge $uv \in E(H)$ is incident to the vertices of $S_0$, then $uv \notin F$.
Hence, the edges $E_H(A \setminus S_0, B \setminus S_0)$ and the edges $E_H(A, B)$ are precisely the same, implying that $F$ is a $d$-cut of $\widetilde{H}$.

For the forward direction ($\Rightarrow$), let $(A, B)$ be the bipartition corresponding to the $d$-cut $F = E_{\widetilde{H}}(A, B)$ of $\widetilde{H}$.
Note that every connected component of $H[S_0]$ is a clique with at least $2d+1$ vertices.
Furthermore, if $C$ is a connected component of $H[S_0]$, then every vertex of $C$ is adjacent to a unique vertex $u \in I_d$.
Therefore, if $u \in A$, then we add the vertices of $C$ to $A$.
Otherwise, if $u \in B$, then we add the vertices of $C$ to $B$.
This completes our construction of a $d$-cut of $H$ and the set of edges remain the same as $E_{\widetilde{H}}(A, B)$.
This completes the proof that $F$ is also a $d$-cut of $H$.
\end{proof}


\subparagraph*{Information of all $d$-cuts of $G$ are available in $\widetilde{H}$.}
Now, we move on to illustrating why the graph $H$ contains the information of all the $d$-cuts of $G$.
As every $d$-cut of $H$ is also a $d$-cut of $\widetilde{H}$ which is a subgraph of $G$, we move on to establish relations among the $d$-cuts of $G$ and the $d$-cuts of $\widetilde{H}$.

Given a nonempty set $T \in {{S}\choose{\leq d}}$, we define $L_T = \{xy \in E(G) \mid x \in T, y \in \widehat{I_T}\}$, the set of all edges that have one endpoint in $T$ and other endpoint in $\widehat{I_T}$.
Additionally, let $I_T = \widehat{I_T} \cap Z$, the vertices from $\widehat{I_T}$ that have been marked by {\sf MarkVC}($G, S, k$).
We call a nonempty set $T \in {{S}\choose{\leq d}}$ {\em good} in $G$ (respectively in $\widetilde{H}$) if $|\widehat{I_T} \cap V(G)| \leq 2d+1$ (respectively $|V(\widetilde{H}) \cap \widehat{I_T}| \leq 2d+1$).
A nonempty set $T \in {{S}\choose{\leq d}}$ is {\em bad} in $G$ (respectively in $\widetilde{H}$) if $T$ is not a good subset of $S$ in $G$ (respectively in $\widetilde{H}$).
Observe from the {\sf MarkVC}($G, S, k$) that if a nonempty set $T \in {{S}\choose{\leq d}}$ is good in $G$, then every $u \in \widehat{I_T}$ is marked in $Z$ and hence are also in $I_T$.
Therefore, every vertex of $\widehat{I_T}$ is also in $\widetilde{H}$ and $|\widehat{I_T} \cap V(\widetilde{H})| \leq 2d+1$.
Hence, $T \subseteq S$ is also a good subset in $\widetilde{H}$.
Similarly, if a nonempty set $T \in {{S}\choose{\leq d}}$ is bad in $G$, then exactly $2d+2$ vertices from $\widehat{I_T}$ are marked and are present in $\widetilde{H}$.
Then, $T$ is also a bad subset of $S$ in $\widetilde{H}$.

Let $\YY \subseteq {{S}\choose{\leq d}}$ denote the collection of all bad subsets of $S$ in $G$ (respectively in $\widetilde{H}$).
We define $L(\YY) = \cup_{T \in \YY} L_T$ and $I(\YY)$ the set of all vertices of $I_d$ each of which have neighborhood exactly $T$ such that $T \in \YY$.
Formally, $I(\YY) = \bigcup_{T \in \YY} \widehat{I_T}$.
Our next lemma establishes that the marking scheme preserves all the $d$-cuts of $G - I(\YY)$.


\begin{lemma}
\label{lemma:vc-solutions-higher-degree}
A set of edges $F \subseteq E(G)$ is a $d$-cut of $G - I(\YY)$ if and only if $F$ is a $d$-cut of $\widetilde{H} - I(\YY)$. 
\end{lemma}

\begin{proof}
For simplicity of presentation, we let $\vec{H} = \widetilde{H} - I(\YY)$ and $\vec{G} = G - I(\YY)$.
If every vertex of $\vec{G}$ is also a vertex of $\vec{H}$, then $\vec{G}$ and $\vec{H}$ are precisely the same graph and the statement trivially holds true.
So, we assume that there is $v \in \vec{G}$ that is not in $\vec{H}$.
Observe that if a vertex $v \in \widehat{I_T}$ for some $T \in {{S}\choose{\leq d}} \setminus \YY$, then $v$ is marked by the procedure ${\sf MarkVC}(G, S, k)$ and present in $\vec{H}$.
Therefore, any vertex $v$ that is in $\vec{G}$ but not in $\vec{H}$ must be in $I_{d+1}$.
We consider an arbitrary $v \in I_{d+1}$ that is not marked by the procedure ${\sf MarkVC}(G, S, k)$.

First, we give the backward direction ($\Leftarrow$) of the proof.
Suppose that $(A', B')$ be the bipartition corresponding to a $d$-cut $F$ of $\vec{H}$.
We prove the existence of a $d$-cut $(A, B)$ of $\vec{G}$ such that $F = E_{\vec{G}}(A, B)$.
At the very beginning, we initialize $A := A'$ and $B := B'$.
Consider a vertex $v$ that is in $\vec{G}$ but not in $\vec{H}$.
Then, such a vertex $v$ satisfies that $v \in I_{d+1}$ that is in $\vec{G}$ but not in $\vec{H}$.
	Then, for every $T \subseteq N_G(v)$ with $|T| = d+1$ the procedure ${\sf MarkVC}(G, S, k)$ did not mark $v$ from the vertices of $\bigcap\limits_{x \in T} (N_G(x) \cap I_{d+1})$.
	If $v$ has exactly $d+1$ neighbors in $S$, then $T = N_G(v)$ is a unique set with $d+1$ vertices.
	Due to Lemma \ref{lem:neighborhood-break-up}, either $T \subseteq A'$ or $T \subseteq B'$.
	If $T \subseteq A'$, then set $A := A \cup \{v\}$.
	Otherwise, $T \subseteq B'$, we set $B := B \cup \{v\}$.
	
	If $v$ has more than $d+1$ neighbors in $S$, then there are two distinct subsets $T_1, T_2 \subseteq N_G(v)$ with $d+1$ vertices each.
	The procedure ${\sf MarkVC}(G, S, k)$ did not mark $v$ from the vertices of $\bigcap\limits_{x \in T_1} (N_G(x) \cap I_{d+1})$ and $\bigcap\limits_{x \in T_2} (N_G(x) \cap I_{d+1})$.
	Then there are $2d+1$ marked vertices in $I_{d+1}$ the neighborhood of each of which contains $T_1$ as well as $T_2$.
	
	It follows from Lemma \ref{lem:neighborhood-break-up}, either $T_1 \subseteq A'$ or $T_1 \subseteq B'$.
	Moreover, it also holds from Lemma \ref{lem:neighborhood-break-up} that either $T_2 \subseteq A'$ or $T_2 \subseteq B'$.
	We have to argue that $T_1 \cup T_2 \subseteq A'$ or $T_1 \cup T_2 \subseteq B'$.
	For the sake of contradiction, let $T_1 \subseteq A'$ and $T_2 \subseteq B'$ (the case of $T_1 \subseteq B'$ and $T_2 \subseteq A'$ is symmetric).
	We construct a set $T_3$ by taking $d$ vertices from $T_1$ and one vertex of $T_2 \setminus T_1$.
	By our choice $T_3 \cap A', T_3 \cap B' \neq \emptyset$.
	Observe that $T_3 \subseteq N_G(v)$ and has $d+1$ vertices.
	Then this assumption on $T_3$ implies that $T_3 \cap A', T_3 \cap B' \neq \emptyset$.
	
	Note that $v \in \bigcap\limits_{x \in T_3} (N_G(x) \cap I_{d+1})$ and the procedure ${\sf MarkVC}(G, S, k)$ did not mark $v$. 
	It means that there are at least $2d+1$ vertices in $I$ the neighborhood of which contains $T_3$.
	But then from Lemma \ref{lem:neighborhood-break-up}, it follows that $T_3 \subseteq A'$ or $T_3 \subseteq B'$.
	This contradicts the choice of $T_3$ that intersects both $A'$ and $B'$.
	Therefore, all subsets of ${{N_G(v)}}$ with $d+1$ vertices either are contained in $A'$ or is contained in $B'$ implying that $N_G(v) \subseteq A'$ or $N_G(v) \subseteq B'$.
	
	If $N_G(v) \subseteq A'$, then set $A := A \cup \{v\}$, otherwise $N_G(v) \subseteq B'$, and we set $B := B \cup \{v\}$.
	We repeat this for every vertex $v$ that is in $\vec{G}$ but not in $\vec{H}$.
	This construction of $A$ and $B$ preserves that $E_{\vec{G}}(A, B) = E_{\vec{H}}(A', B') = F$.
	Therefore, $(A, B)$ is a $d$-cut of $\vec{G}$ such that $F = E_{\vec{G}}(A, B)$.
	This completes the correctness proof of the backward direction.
	
	\medskip
	
	Let us now give the forward direction ($\Rightarrow$) of the proof.
	Suppose that $(A, B)$ be the corresponding vertex bipartition of a $d$-cut $F$ of $\vec{G}$ such that $F = E_{\vec{G}}(A, B)$.
	We argue that $F$ is a $d$-cut of $\vec{H}$ as well.
	For the sake of contradiction, we assume that there is $uv \in E(\vec{G}) \setminus E(\vec{H})$ such that $uv \in F$.
	As $\vec{H}$ is an induced subgraph of $\vec{G}$ and the vertices of $\vec{G} \setminus \vec{H}$ forms an independent set, it must be that either $u \notin V(\vec{H})$ or $v \notin V(\vec{H})$.
	Without loss of generality, we assume that $u \notin V(\vec{H})$.
	Then, $u$ is one of the unmarked vertices from $I_{d+1}$, and we assume that $u \in A$.
	Consider an arbitrary $T \in {{N_G(u)}\choose{d+1}}$ such that $v \in T$.
	There are at least $2d+1$ marked vertices in $\bigcap\limits_{x \in T} (N_G(x) \cap I_{d+1})$ and $u$ is excluded from marking by the procedure ${\sf MarkVC}(G, S, k)$.
	In particular, there are at least $2d+2$ vertices in $G$ the neighborhood of which contains $T$.
	Then, due to Lemma \ref{lem:neighborhood-break-up}, it follows that either $T \subseteq A$ or $T \subseteq B$.
	But, by our choice $u \in A$.
	Since $v \in T$ and $uv \in F$, it must be that $v \in B$.
	But then $T \subseteq B$ implies that $u \in A$ has at least $d+1$ neighbors in $B$.
	This contradicts our initial assumption that $F = E_{\vec{G}}(A, B)$ is a $d$-cut of $\vec{G}$.
	Therefore, every edge $uv \in F$ is also an edge of $\vec{H}$.
Since $\vec{H}$ is a subgraph of $\vec{G}$, this completes our proof that $F$ is a $d$-cut of $\vec{H}$.
\end{proof}


Given an arbitrary subset $T \in \YY$ and $u \in \widehat{I_T}$, we define $J_T^u = \{ux \mid x \in T\}$; i.e. the set of edges with one endpoint $u$ and the other endpoint in $T$. 
A bad subset $T \in \YY$ is {\em broken} by an edge set $F$ if there exists $u \in \widehat{I_T}$ such that $J^u_T \not\subset F$.
A bad subset $T \in \YY$ is {\em occupied} by an edge set $F \subseteq E(G)$ if $T$ is not broken by $F$.
Equivalently, if a bad subset $T \in \YY$ is occupied by an edge set $F$, then for every $u \in \widehat{I_T}$, either $J^u_T \subseteq F$, or $J^u_T$ is disjoint from $F$.
We prove the following observation about the bad subsets $T \in \YY$ can be proved by using Lemma \ref{lem:neighborhood-break-up}.


\begin{observation}
\label{obs:broken-case-understanding}
Let an edge set $F$ be a $d$-cut of $G$ (respectively, a $d$-cut of $\widetilde{H}$).
Then, every bad subset $T \in \YY$ is occupied by $F$.
\end{observation}

\begin{proof}
We give the proof for $G$.
As bad subset $T \in \YY$ in $G$ is a bad subset of $S$ in $\widetilde{H}$, the proof arguments for $\widetilde{H}$ are exactly the same.

Let $(A, B)$ be the vertex bipartition of a $d$-cut $F$ of $G$ such that $F = E_G(A, B)$ and $T \in \YY$ be a bad subset of $S$.
If $T \in \YY$ is a bad subset of $S$ in $G$, there are more than $2d+1$ vertices in $\widehat{I_T}$.
It follows from Lemma \ref{lem:neighborhood-break-up} that $T \subseteq A$ or $T \subseteq B$.
Let $T \subseteq A$ (arguments for $T \subseteq B$ is symmetric).
Then, for any $u \in \widehat{I_T}$, if $u \in A$, then $J^u_T \cap F = \emptyset$.
If $u \in B$, then $J^u_T \subseteq F$.
This completes the proof that $T$ is occupied by $F$.
\end{proof}

\subparagraph{Equivalence Classes of Edge Sets:} 
With the above observation in hand, we define the notion of `equivalence' among the edge sets of $G$ and the edge sets of $\widetilde{H}$ as follows.

\begin{definition}
\label{defn:vc-eqv-d-cuts}
Two sets of edges $F_1$ and $F_2$ are {\em equivalent} if
\begin{enumerate}[(i)]
	\item $F_1 \setminus L(\YY) = F_2 \setminus L(\YY)$,
	\item for every $T \in \YY$, $|F_1 \cap L_T| = |F_2 \cap L_T|$, and
	\item every bad subset $T \in \YY$ is occupied by $F_1$ if and only if $T$ is occupied by $F_2$.
\end{enumerate}
\end{definition}

It is not hard to see that the above definition gives an {\em equivalence relation}.
This definition not only holds for two edge sets of $G$, but also holds for two edge sets of $\widetilde{H}$ as well as for one edge set of $G$ and one edge set of $\widetilde{H}$.
From now onwards, for every edge set $F$ of $G$ (respectively, of $\widetilde{H}$), we consider the partition $F = F_b \uplus F_g$ such that $F_b = \bigcup_{T \in \YY} (F \cap L_T)$ and $F_g = F \setminus F_b$.
Observe that $F_b$ is the intersection of $F$ with $\cup_{T \in \YY} L_T$ and $F_g = F \setminus L(\YY)$.
Our next lemma is a crucial structural property saying that any edge set equivalent to a $d$-cut is also a $d$-cut.

\begin{lemma}
\label{lemma:equivalence-connection-H-G}
Let $F' \subseteq E(\widetilde{H})$ be a (maximal) $d$-cut  $d$-cut, respectively) of $\widetilde{H}$.
Then, every $F \subseteq E(G)$ that is equivalent to $F'$ is a (maximal) $d$-cut of $G$.
\end{lemma}

\begin{proof}
Let $F' = E_{\widetilde{H}}(A', B')$ be a $d$-cut $F'$ of $\widetilde{H}$ and $F \subseteq E(G)$ is equivalent to $F'$.
If $F' = \emptyset$, then $G$ and $\widetilde{H}$ are disconnected graphs and $F = \emptyset$ is the only edge set that is equivalent to $F'$.
Then, $F$ is a $d$-cut of $G$.
We consider the nontrivial case when $F' \neq \emptyset$.
It suffices to prove that 
\begin{itemize}
	\item every $u \in V(G)$ has at most $d$ edges incident in $F$, and
	\item there is a vertex bipartition $(A, B)$ of $G$ such that $F = E_G(A, B)$ or the deletion of $F$ increases the number of connected components in $G$.
\end{itemize}

We consider the partition $F' = F_g' \uplus F_b'$ such that $F_b' = \cup_{T \in \YY} (F' \cap L_T)$.
Then $F_g' = F' \setminus L(\YY)$.
There are three cases.

\begin{description}
	\item[Case (i):] $F_g' \neq \emptyset$ but $F_b' = \emptyset$. 
	Observe that $F_g'$ is a $d$-cut of $\widetilde{H} - I(\YY)$.
	Due to Lemma \ref{lemma:vc-solutions-higher-degree}, $F_g'$ is a $d$-cut of $G - I(\YY)$.
	As $F$ is equivalent to $F'$, it must be that $F \setminus L(\YY) = F' \setminus L(\YY) = F_g'$.
	Since $F_b' = \emptyset$, for every $T \in \YY$, $F' \cap L_T = \emptyset$.
	Then, $F \cap L_T = \emptyset$ for every $T \in \YY$ implying that $F = F'$.
	It remains to establish that $F$ is a $d$-cut of $G$.
	We extend $(A', B')$ into a vertex bipartition of $(A, B)$ such that $F = F' = E_G(A, B)$.
	We initialize $A = A'$ and $B = B'$.
	Since $F'$ is a $d$-cut of $\widetilde{H}$, due to Observation \ref{obs:broken-case-understanding}, every $T \in \YY$ is occupied by $F'$.
	Furthermore, $F' \cap L_T = \emptyset$ as $F_b' = \emptyset$.
	Then, for every $u \in I_T$, $J^u_T \cap F' = \emptyset$.
	In such a case, it must be that either $T \cup I_T \subseteq A'$ or $T \cup I_T \subseteq B'$.
	If $T \cup I_T \subseteq A'$, then $A := A \cup \widehat{I_T}$.
	Otherwise, $T \cup I_T \subseteq B'$, we set $B := B \cup \widehat{I_T}$.
	We repeat this procedure for every $T \in \YY$.
	Observe that this procedure does not add any edge apart from the ones in $F' (= F)$ and hence $F = E_G(A, B)$.
	Therefore, $F$ is a $d$-cut of $G$.
	
	\item[Case (ii):] $F_g' = \emptyset$ but $F_b' \neq \emptyset$.
	Since $F$ is equivalent to $F'$, it must be that $F \setminus L(\YY) = F_g' = \emptyset$.
	Therefore, for every $u \in I_{d+1}$, no edge of $F$ is incident to $u$.
	For every $T \in \YY$, $|F \cap L_T| = |F' \cap L_T|$.
	As $F_b' \neq \emptyset$, we focus only on the bad subsets $T \in \YY$ such that $F' \cap L_T \neq \emptyset$.
	Due to Observation \ref{obs:broken-case-understanding}, every $T \in \YY$ is occupied by $F'$.
	As $F$ is equivalent to $F'$, every $T \in \YY$ is occupied by $F$.
	If $F' \cap L_T \neq \emptyset$ for some $T \in \YY$, then there is $u \in I_T$ such that $J^u_T \subseteq F'$.
	For every $u \in T$, since $|F \cap L_T| = |F' \cap L_T|$, if there are exactly $\alpha_T$ vertices in $w \in I_T$ satisfying that $J^w_T \subseteq F'$, then there are exactly $\alpha_T$ vertices $w \in \widehat{I_T}$ satisfying that $J^w_T \subseteq F$.
	Therefore for every $u \in S$, the number of edges in $F$ incident to $u$ remains the same as the number of edges of $F'$ incident to $u$.
	Since $F'$ is a $d$-cut of $G$ at most $d$ edges of $F'$ are incident to $u$, it follows that at most $d$ edges of $F$ are incident to $u$.
	Since there is $w \in \widehat{I_T}$ such that $J^w_T \subseteq F$, for some $T \in \YY$, therefore, all edges incident to $w$ are in $F$.
	Hence, the deletion of $F$ from $G$ increases the number of connected components.
	Therefore, $F$ is a $d$-cut of $G$.
		
	\item[Case (iii):] $F_g', F_b' \neq \emptyset$.
	Since $F$ is equivalent to $F'$, it must be that $F \setminus L(\YY) = F' \setminus L(\YY) = F_g'$.
	Therefore, for every $u \in I_{d+1}$, the edges of $F'$ incident to $u$ and the edges of $F$ incident to $u$ are precisely the same.
	As $F'$ is a $d$-cut, at most $d$ edges of $F$ are incident to every $u \in I_{d+1}$.
	Similarly, for every $u \in I_d$, there are at most $d$ edges in $G$ that are incident to $u$.
	Therefore, for every $u \in I_d$, at most $d$ edges of $F$ are incident to $u$.
	
	We focus on the vertices $u \in S$.
	Since $F$ is equivalent to $F'$, the edges of $F \setminus L(\YY)$ that are incident to $u$ remains the same as the set of edges of $F' \setminus L(\YY)$ that are incident to $u$.
	As $F'$ is a $d$-cut of $\widetilde{H}$, due to Observation \ref{obs:broken-case-understanding}, every $T \in \YY$ is occupied by $F'$.
	As $F$ is equivalent to $F'$, every $T \in \YY$ is also occupied by $F$.
	For every $T \in \YY$, if $F' \cap L_T \neq \emptyset$, we focus on $\alpha_T$, the number of vertices $w \in I_T$ such that $J^w_T \subseteq F'$.
	As $|F \cap L_T| = |F' \cap L_T|$, there are $\alpha_T$ vertices $w \in \widehat{I_T}$ satisfies that $J^w_T \subseteq F$.
	Therefore, $\alpha_T$ edges of $F' \cap L_T$ are incident to $u$ that is the same as the number of edges in $F \cap L_T$ incident to $u$.
	
	Based on the above mentioned arguments, the number of edges of $F'$ incident to $u$ remains the same as the number of edges of $F$ incident to $u$.
	Hence, at most $d$ edges of $F$ are incident to $u$.
	Finally, observe that there is $T \in \YY$ such that $J^w_T \subseteq F$.
	Hence, all edges incident to $w$ are in $F$.
	This ensures us that the deletion of $F$ increases the number of connected components in $G$.
	Therefore, $F$ is a $d$-cut of $G$.
\end{description}
Since the above cases are mutually exhaustive, this completes the proof.
\end{proof}

From the above lemma, it is not hard to see that if $F$ is a $d$-cut of $G$, and $F'$ is equivalent to $F$, then $F'$ is a $d$-cut of $G$.

\subsubsection{A Challenge to Avoid Duplicate Enumeration}
\label{sec:vc-duplicate-enum}

It is clear from the lemmas in the previous sections that given a $d$-cut $F'$ of $H$ that is also a $d$-cut of $\widetilde{H}$, our solution-lifting algorithm should output a collection of $d$-cuts of $G$ that are equivalent to $F'$.
But given two distinct $d$-cuts $F_1$ and $F_2$ of $H$, it is possible there exists a $d$-cut $F$ of $G$ that contains an edge not in $H$ and is equivalent to both $F_1$ and $F_2$.
We have to ensure that our solution-lifting algorithm outputs $F$ when exactly one $d$-cut of $H$ is given.
To circumvent this challenge/difficulty, we define the notion of `suitable' bad subsets from $\YY$ using the information of the pendant cliques of $S_0$ that are attached to $\widetilde{H}$ as follows.

Observe that $F$ is equivalent to both $F_1$ and $F_2$, but differs from both due to the edges present in $F \cap L_T$ for some $T \in \YY$. 
Based on  procedure {\sf MarkVC}($G, S, k$), if there is a vertex in $\widehat{I_T} \setminus I_T$ (that is unmarked), then we choose $d$ arbitrarily marked vertices $u_1,\ldots,u_d \in I_T$.
By construction of $H$, for every $i \in [d]$, we have attached a pendant clique $C_T^i$ with $2d+i$ vertices into $u_i$.
In particular, $C_T^i$ is a connected component of $H[S_0]$.
We use the sizes of these attached pendant cliques to define the notion of suitable for every bad subset $T \in \YY$.

\begin{definition}
\label{defn:vc-r-suitable-subsets}	
Let $F'$ be a $d$-cut of $H$.
A bad subset $T \in \YY$ is {\em $r$-suitable} for some $r > 0$ with respect to $F'$ if 
\begin{enumerate}[(i)]
	\item $F' \cap L_T \neq \emptyset$, 
	\item there is a set $X$ of exactly $r$ vertices $v_1,\ldots,v_r \in I_T$ such that for every $i \in [r]$, $J^{v_i}_T \subseteq F'$, and the pendant clique of $S_0$ attached to $v_i$ has exactly $2d+i$ vertices, and
	\item for every $u \in I_T \setminus X$, $J^u_T \cap F' = \emptyset$.
\end{enumerate}
\end{definition}

%

\begin{figure}[t]
\centering
	\includegraphics[scale=0.3]{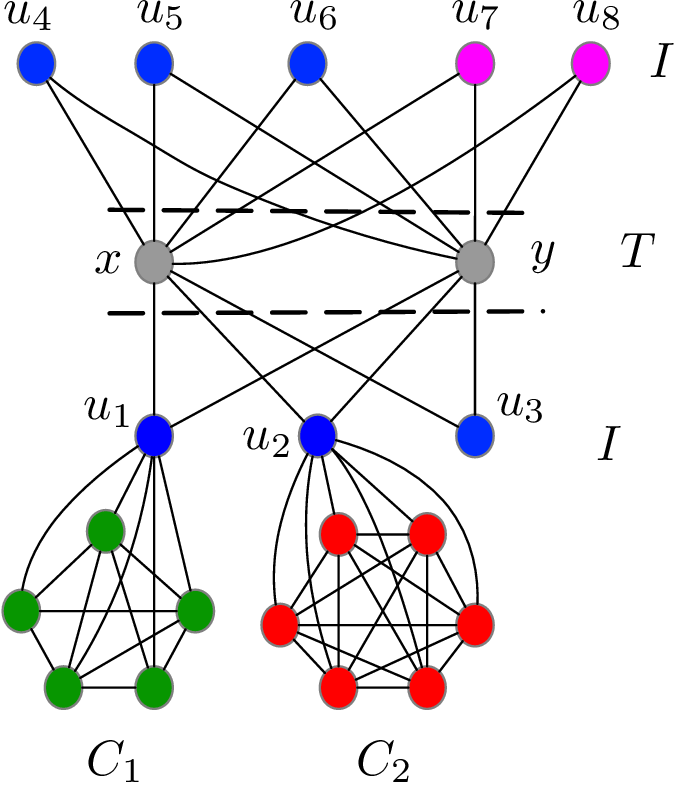}
	\caption{An illustration of $r$-suitable subsets for $d=2$. The graph $G$ contains all the vertices except the two pendant cliques $C_1, C_2$. The graph $H$ contains all but the pink vertices $u_7$ and $u_8$. 
	The set $T = \{x, y\}$ is 1-suitable with respect to the $d$-cut $F_1 = E_H(A_1, B_1)$ of $H$ such that $B_1 = C_1 \cup \{u_1\}$ and $A_1 = V(H) \setminus B'$.
	In addition, $T$ is 2-suitable with respect to the $d$-cut $F_2 = E_H(A_2, B_2)$ of $H$ such that $B_2 = C_1 \cup C_2 \cup \{u_1, u_2\}$ and $A_2 = V(H) \setminus B_2$.
	Finally, $T$ is unsuitable with respect to the $d$-cut $F_3 = E_H(A_3, B_3)$ such that $B_3 = C_2 \cup \{u_2, u_3\}$ and $A_3 = V(H) \setminus B_3$.
	We also provide an illustration of Lemma \ref{lemma:vc-equivalence-of-d-cuts}.
	Consider the $d$-cut $F = E_G(A, B)$ of $G$ such that $B = \{u_2, u_7\}$ and $A = V(G) \setminus B$.
	Then, $F_2$ is a $d$-cut of $G$  equivalent to $F$ and $T$ is 2-suitable with respect to $F_2$.}
\label{fig:r-suitable-sets}
\end{figure}


If a bad subset $T \in \YY$ is not $r$-suitable with respect to the $d$-cut $F'$ of $H$ for any $r > 0$, then $T$ is {\em unsuitable} with respect to $F'$.
We refer to Figure \ref{fig:r-suitable-sets} for an illustration of 1-suitable and 2-suitable sets with respect to a $d$-cuts for $d=2$.
The same figure also illustrates the existence of a $d$-cut with respect to which the bad subset is unsuitable.
Using the above definition of $r$-suitable bad subsets, we prove the following lemma that is central to the correctness of our solution-lifting algorithm.
In particular, we prove that for every (maximal) $d$-cut, there exists a (maximal) $d$-cut in $H$ satisfying some crucial properties.


\begin{lemma}
\label{lemma:vc-equivalence-of-d-cuts}
Let $F \subseteq E(G)$.
Then, the following statements are true.
\begin{enumerate}[(i)]
	\item If $F$ is a (maximal) $d$-cut of $G$ and $F \subseteq E(\widetilde{H})$, then $F$ is a (maximal) $d$-cut of $H$.
	\item If $F$ is a (maximal) $d$-cut  of $G$ and $F \not\subset E(\widetilde{H})$, then $H$ has a unique (maximal) $d$-cut  $F'$ that satisfies the following two properties.
	\begin{itemize}
		\item $F'$ is equivalent to $F$, and
		\item for every $T \in {\YY}$, if $F \cap L_T \not\subset E(\widetilde{H})$, then there is $r > 0$ such that $T$ is $r$-suitable with respect to $F'$.
	\end{itemize} 
\end{enumerate}
We refer to Figure \ref{fig:r-suitable-sets} for an illustration of this lemma.
\end{lemma}

\begin{proof}
Given $F \subseteq E(G)$, we prove the statements in the given order.
\begin{enumerate}[(i)]
	\item For the first item, the premise says that $F$ is a (maximal) $d$-cut of $G$ and $F \subseteq E(\widetilde{H})$.
	Since $\widetilde{H}$ is a subgraph of $G$, it follows that $F$ is a (maximal) $d$-cut of $\widetilde{H}$.
	Due to Lemma~\ref{lemma:vc-d-cuts-without-S0}, it follows that $F$ is a (maximal) $d$-cut of $H$.
	

	\item Let $F \not\subset E(\widetilde{H})$ and $F = E_G(A, B)$.
	By definition, $F = F_b \uplus F_g$ such that $F_b = \bigcup\limits_{T \in \YY} (F \cap L_T)$.
	If $F_b = \emptyset$, then $F_g$ is precisely a $d$-cut of $G - I(\YY)$.
	Then due to Lemma~\ref{lemma:vc-solutions-higher-degree}, $F_g$ is also a $d$-cut of $\widetilde{H} - I(\YY)$.
	In such a case, $F \subseteq E(\widetilde{H})$ contradicting the assumption that $F \not\subset E(\widetilde{H})$.
	So, it must be that $F_b \neq \emptyset$.
	We justify the existence of a $d$-cut $F'$ of $H$ that satisfies the properties required to prove the lemma.
	
	Initialize $F' := F$.
	Since $F_b \neq \emptyset$, there is a bad subset $T \in \YY$ such that $F \cap L_T \not\subset E(\widetilde{H})$.
	We consider every bad subset $T \in \YY$ one by one such that $F \cap L_T \not\subset E(\widetilde{H})$.
	Due to Observation \ref{obs:broken-case-understanding}, every bad subset $T \in \YY$ is occupied by $F$.
	Hence, for every $v \in \widehat{I_T}$, either $J^v_T \subseteq F$ or $J^v_T \cap F = \emptyset$.
	Since $F \cap L_T \neq \emptyset$ and consists of an edge that is not in $\widetilde{H}$, there exists $v \in \widehat{I_T} \setminus I_T$ (unmarked vertex not in $\widetilde{H}$) such that $J^v_T \subseteq F$.
	Then, there are exactly $2d+2$ vertices in $I_T$.
	Moreover, there are $d$ vertices $v_1,\ldots,v_d \in I_T$ (that are marked) such that a pendant clique $C_i$ of $2d+i$ vertices is attached to $v_i$ and $C_i$ is a connected component of $H[S_0]$.
	
	We focus on $r$, the number of vertices $w \in \widehat{I_T}$ such that $J^w_T \subseteq F$.
	Since $F \cap L_T \neq \emptyset$, clearly $r > 0$.
	We choose the vertices $X' = \{v_1,\ldots,v_r\} \subseteq I_T$ such that the pendant clique of $S_0$ attached to $v_i$ has $2d+i$ vertices.
	Then, we set $F' := (F' \setminus L_T) \cup (\bigcup\limits_{i \in [r]} J^{v_i}_T)$.
	We repeat this replacement procedure for every $T \in \YY$ such that $F \cap L_T$ contains an edge that is not in $\widetilde{H}$.
	Observe that for two distinct subsets $T_1, T_2 \in \YY$, this replacement procedure involves the vertex subsets from $I_{T_1}$ and $I_{T_2}$ that are pairwise disjoint from each other.
	Therefore, every pair of edge set replacement involves pair of edge sets that are disjoint from each other.
	
	This replacement procedure keeps $F_g$ unchanged and
	ensures us that for every $T \in \YY$, if $F \cap L_T$ has an edge that is not in $\widetilde{H}$, then $|F \cap L_T| = |F' \cap L_T|$.
	Hence, $F'$ is equivalent to $F$.
	As $F_g \subseteq E(G - I(\YY))$, $F_g$ is a $d$-cut of $E(G - I(\YY))$.
	Due to Lemma \ref{lemma:vc-solutions-higher-degree}, $F_g$ is a $d$-cut of $\widetilde{H} - I(\YY)$.
	Every $u \in I_d$, has at most  $d$ edges incident in $\widetilde{H}$.
	Hence, for every $u \in I_d$, there are at most $d$ edges that are incident to $F'$.
	
	Finally, observe that by construction, if $F \cap L_T \neq \emptyset$, then there is $w \in I_T$ such that $J^w_T \subseteq F'$.
	This means that deletion of $F'$ from $\widetilde{H}$ increases the number of connected components in $\widetilde{H}$.
	Hence, $F'$ is a $d$-cut of $\widetilde{H}$. 
	Due to Lemma \ref{lemma:vc-d-cuts-without-S0}, $F'$ is a $d$-cut of $H$.
	
	It remains to establish that for every $T \in \YY$, if $F \cap L_T \not\subset E(\widetilde{H})$, then $T$ is $r$-suitable with respect to $F'$.
	By construction of $F'$, for every $T \in \YY$, if $F \cap L_T \not\subset E(\widetilde{H})$, then a unique set $X' \subseteq I_T$ is chosen such that $E(X', T) = F' \cap L_T$ and the sizes of the pendant cliques attached to the vertices of $X'$ are $\{2d+1,\ldots,2d+r\}$.
	Therefore, $T$ is $r$-suitable with respect to $F'$ and the construction of $F'$ is unique.
	This completes the proof for $d$-cuts.
	Proof arguments for maximal $d$-cuts are similar.
\end{enumerate}
This completes the proof of the lemma.
\end{proof}

\subsubsection{Designing the Enumeration Kernelization Algorithms}
\label{sec:vc-final-algorithms}

Finally, we are ready to illustrate how we can use structural properties of equivalent $d$-cuts, and $r$-suitable subsets $T \in \YY$ to design our enumeration kernelizations.
We describe our the solution-lifting algorithm that is one of the central parts of our result.
Given a $d$-cut of $H$, our next two lemma statements illustrate the algorithms that enumerate a collection of $d$-cuts of $G$ in polynomial-delay such that property (ii) of Definition \ref{defn:poly-delay-enum-kernel} is satisfied.
The correctness proofs of the next lemma crucially relies on the correctness of the Lemma~\ref{lemma:vc-equivalence-of-d-cuts}.

\begin{lemma}
\label{lemma:vc-soln-lifting-algo-2}
Suppose that $(H, S \cup S_0, |S \cup S_0|)$ be the output instance obtained after invoking the marking procedure on the input instance $(G, S, k)$.
Furthermore, let $\FF(G)$ be the collection of all maximal $d$-cuts of $G$ and $\FF(H)$ be the collection of all maximal $d$-cuts of $H$.
Then given a maximal $d$-cut $F^*$ of $H$, then there is an algorithm that enumerates a collection $\CC(F^*)$ of maximal $d$-cuts of $G$ in $k^{\OO(d)}$-delay such that
\begin{enumerate}[(i)]
	\item every $F \in \CC(F^*)$ is equivalent to $F^*$, and
	\item if $F \neq F^*$ for some $F \in \CC(F^*)$, then $F$ has an edge that is in $G$ but not in $H$.
\end{enumerate}
Moreover, $\{\CC(F^*) \mid F^* \in \FF(H)\}$ is a partition of $\FF(G)$.
\end{lemma}

\begin{proof}
Let $F^* \in \FF(H)$ be a maximal $d$-cut of $H$.
If $F^* = \emptyset$, then we just output $F = F^*$ as the only $d$-cut of $G$.
Clearly, $\emptyset$ is the only possible equivalent $d$-cut of $\emptyset$.

Now, we consider when $F^* \neq \emptyset$.
It follows from Lemma \ref{lemma:vc-d-cuts-without-S0} that $F^*$ is a $d$-cut of $\widetilde{H}$.
Suppose $(A^*, B^*)$ denotes the vertex bipartition corresponding to the $d$-cut $F^*$ of $\widetilde{H}$.
Recall that $\YY \subseteq {{S}\choose{\leq d}}$ is the collection of all bad subsets of $S$, and $L(\YY) = \bigcup\limits_{T \in \YY} L_T$.
We consider the partition $F^* = F^*_b \uplus F^*_g$ such that $F_g^* = F^* \setminus L(\YY)$ and $F^*_b = \cup_{T \in \YY} (F^* \cap L_T)$.
Observe that $F^*_g = F^* \setminus L(\YY)$.

	If $F^*_b = \emptyset$, then $F^*$ is a $d$-cut of $\widetilde{H} - I(\YY)$.
	Then, due to Lemma \ref{lemma:vc-solutions-higher-degree}, $F^*$ is a $d$-cut of $G - I(\YY)$.
	Since $F^*$ is equivalent to itself and $F^*$ is a $d$-cut of $\widetilde{H}$, it follows from Lemma \ref{lemma:equivalence-connection-H-G} that $F^*$ is a $d$-cut of $G$.
	Then our algorithm just outputs $F^*$ for the instance $(G, S, k)$.
	
	Otherwise it must be that $F_b \neq \emptyset$.
	We use the following procedure to compute a collection of maximal $d$-cuts of $G$ with $k^{\OO(d)}$-delay.
	The algorithm is recursive and takes parameter as $(J, \DD)$ as follows.
	We refer to Algorithm \ref{alg:two} for a detailed pseudocode.
	
	\medskip
	
	\noindent
	{\bf Algorithm Description:}
	Initialize $J = F_g$ and $\DD \subseteq \YY$ such that if $F^* \cap L_T \neq \emptyset$, then $T \in \DD$.
	Note that by checking every $T \in \YY$ one by one, we can construct $\DD$.
	There are $k^d$ distinct subsets in $\YY$.
	
	\begin{enumerate}
		\item If $\DD = \emptyset$, then output $J$.
		\item Else, take an arbitrary $T \in \DD$ and let $F_T = F^* \cap L_T$.
		Note that $T$ is a bad subset of $S$ in $G$.
		Observe that every $uv \in F_T$ has one endpoint $u \in T$ and the other endpoint $v \in {I_T}$.
		Suppose $P \subseteq I_T$ such that $E(P, T) = F_T$ and $|P| = \beta$.
		Note that $\beta > 0$.
		Either $T$ is $\beta$-suitable with respect to $F^*$ or $T$ is not suitable with respect to $F^*$.
		We can check this in polynomial time as follows.
		For every $w \in P$, let $b_w$ denotes the number of vertices in the pendant clique $C_w$ that is attached to $w$.
		Formally, let $Q_P = \{b_w \mid w \in P\}$.
		By definition, if $Q_P = \{2d+1,\ldots,2d+\beta\}$, then $T$ is $\beta$-suitable with respect to $F^*$.
		Otherwise, $T$ is not $\beta$-suitable with respect to $F^*$ for any $\beta > 0$.
		It implies that $T$ is unsuitable with respect to $F^*$.
		
		If $T$ is unsuitable with respect to $F^*$, then we recursively call the algorithm with parameter $J \cup F_T$ and $\DD \setminus \{T\}$.
		Otherwise, perform the following steps.
	
		\item We are now in the case that $T$ is $\beta$-suitable with respect to $F^*$.
		Consider $P \subseteq I_T$ of $\beta$ vertices such that $E(P, T) = F^* \cap L_T$.
		For every set $X'$ of $\beta$ vertices from $\widehat{I_T}$ such that $X' = P$ or $X'$ contains a vertex not in $H$, consider the edge set $E(X', T) = \{uv \mid u \in X', v \in T\}$. 
		For every such set $X'$ of $\beta$ vertices from $\widehat{I_T}$, recursively call the algorithm with parameters $J \cup E(X', T)$ and $\DD \setminus \{T\}$.
	\end{enumerate}
	
	This completes the description (see Algorithm \ref{alg:two} for a pseudocode) of backtracking enumeration algorithm.
	
	Given a $d$-cut $F^* = E(A^*, B^*)$ of $H$, we claim that this enumeration algorithm (or Algorithm \ref{alg:two}) runs in $k^{\OO(d)}$-delay and outputs a collection of $d$-cuts of $G$ that are equivalent to $F^*$.
	Note that the enumeration algorithm keeps $F_g = F^* \setminus L(\YY)$ unchanged in each of the outputted edge sets.
	Hence, for every outputted edge set $F$, it holds that $F \setminus L(\YY) = F^* \setminus L(\YY)$.
%
	Let $T \in \YY$ be a bad set in $G$ such that $F^* \cap L_T \neq \emptyset$.
	Since $T$ is a bad subset of $S$ in $G$ and $F^*$ is a $d$-cut of $\widetilde{H}$, it follows from Observation \ref{obs:broken-case-understanding} that $T$ is occupied by $F^*$.
	If $P \subseteq I_T$ be the set of $\beta > 0$ vertices such that  $E(P, T) = F^* \cap L_T$, then 
	Algorithm \ref{alg:two}, either only enumerates $P$, or it enumerates $P$ and every collection of subsets of $\widehat{I_T}$ having $\beta$ vertices that contains at least one vertex that is not in $H$.
	Since $\beta = |P|$, for any new set $F$ in the collection of edge sets, it follows that for any bad subset $T \in \YY$ that $|F \cap L_T| = |F^* \cap L_T|$.
	Moreover, any bad subset was occupied by $F^*$ and remains occupied by $F$.
	Since the set of edges from $F_g$ are unchanged, $F$ is equivalent to $F^*$.
	This establishes (i) of the statement.
	As $F$ is equivalent to $F^*$, due to Lemma \ref{lemma:equivalence-connection-H-G}, $F$ is a $d$-cut of $G$.
	
	If $\DD \neq \emptyset$, the algorithm picks every set $T \in \DD$ one by one.
	If there is $\beta > 0$ such that $T$ is $\beta$-suitable with respect to $F^*$, then our backtracking enumeration algorithm chooses a collection of set $X'$ of $\beta$ vertices from $\widehat{I_T}$ one by one and constructs a unique set of edges $E(X', T)$.
	Then, the enumeration algorithm recursively calls with parameters $J \cup E(X', T)$ and $\DD \setminus \{T\}$.
	Note that $|X'| = \beta \leq d$ as $T \subseteq A^*$.
	This recursive algorithm enumerates the subsets of $\widehat{I_T}$ with $\beta$ vertices, and the depth of the recursion is $|\DD|$. 
	Hence, this is a standard backtracking enumeration algorithm with $n^{\beta}$ branches, and the depth of the recursion is at most $k^{d}$ as $\DD \subseteq {{S}\choose{\leq d}}$.
	The $d$-cuts are outputted only in the leaves of this recursion tree when $\DD = \emptyset$.
	Hence, the delay between two consecutive distinct $d$-cuts is at most $k^{\OO(d)}$.
	As every internal node involves $(k^d)^{\OO(1)}$-time computation, it follows that the delay between outputting two distinct equivalent $d$-cuts is $k^{\OO(d)}$.
	
	\medskip
	
	It remains to establish that $\{F \in \CC(F^*) \mid F^* \in \FF(H)\}$ is a partition of $\FF(G)$.
	Let $F \in \FF(G)$.
	If $F \subseteq E(\widetilde{H})$, then due to Lemma \ref{lemma:vc-equivalence-of-d-cuts}, $F$ is a $d$-cut of $H$, and $F \in \CC(F)$.
	If $F \not\subset E(\widetilde{H})$, then 
	it follows from Lemma \ref{lemma:vc-equivalence-of-d-cuts} that there exists a unique $d$-cut $F'$ of $G$ such that
	\begin{itemize}
		\item $F'$ is equivalent to $F$, and
		\item for every $T \in \YY$, if $F \cap L_T \not\subset E(\widetilde{H})$, then there is $r > 0$ such that $T$ is $r$-suitable with respect to $F'$.
	\end{itemize}
	
	By the description of our enumeration algorithm (Algorithm \ref{alg:two}), given $F' \in \FF(H)$, if $T$ is $r$-suitable with respect to $F'$ for some $r >0$, then $F \cap L_T$ was enumerated such that $F' \cap L_T = F \cap L_T$.
	Therefore, $F \in \CC(F')$.
	Due to the uniqueness provided by Lemma \ref{lemma:vc-equivalence-of-d-cuts}, it follows that $\{\CC(F') \mid F' \in \FF(H)\}$ is a partition of $\FF(G)$.   

	As we have already established that the delay is $k^{\OO(d)}$ and the properties (i) and (ii) are satisfied by our enumeration algorithm, this completes the proof of the lemma.
	The proof argumentations for enumeration of maximal $d$-cuts are exactly the same.
\end{proof}

\begin{algorithm}[h]
\SetKwFunction{enumeqv}{{\sc EnumEqv}}
\enumeqv{$J$, $\DD$}{

	$F^* = E(A^*, B^*)$\;
	\If{$\DD = \emptyset$}{
		Output $J$\;
	}
	\For{every $T \in \DD$}{
		$P \leftarrow$ the vertices of ${I_T}$ whose neighborhood is $T$ and $E(P, T) = F^* \cap L_T$\;
		Let $|P| = \beta$\;
		\If{If $T$ is not $\beta$-suitable with respect to $F^*$}{
			\enumeqv{$J \cup E(P, T)$, $\DD \setminus \{T\}$}\;
		}
		\Else{
		Recall $P$, the collection of 
		$\beta$ vertices from $\widehat{I_T}$\;
		\For{every $X' \in {\widehat{{I_T}}\choose{\beta}}$ such that $X' = P$ or $X' \cap V(G - H) \neq \emptyset$}{
			$E(X', T) = \{uv \mid u \in X', v \in T\}$\;
			\enumeqv{$J \cup E(X', T)$, $\DD \setminus \{T\}$}\;
		}
		}
	}
}
\caption{Solution Lifting Algorithm for {\enumdcut}}
\label{alg:two}
\end{algorithm}

{\ThmVCOne*}

\begin{proof}
Recall that $(G, S, k)$ is the input instance such that $|S| = k \leq 2{\vc}(G)$ and $S$ is a vertex cover of $G$.
	We divide the situation into two cases.

	\begin{description}
			\item[Case (i):] First, we discuss a trivial case when $G$ has at most $k^d$ vertices.
	Since $G$ has at most $k^d$ vertices, the kernelization algorithm returns the same instance $(H, S, k)$ such that $H = G$.
	Given a $d$-cut $F'$ of $H$, the solution-lifting algorithm just outputs $F = F'$.
	Clearly, $F$ is a $d$-cut of $G$.
	Kernelization algorithm runs in polynomial time and property (ii) of Definition \ref{defn:poly-delay-enum-kernel} is satisfied.
	This gives us a polynomial-delay enumeration kernelization with at most $k^d$ vertices.
	
		\item[Case (ii):] Otherwise, $G$ has more than $k^d$ vertices.
	The kernelization algorithm invokes the marking scheme ${\sf MarkVC}(G, S, k)$ as described earlier this section and constructs the graph $H$ such that $S_0 \cup S$ is a vertex cover of $H$.
	It follows from Observation \ref{obs:size-H-bound} that $H$ has $\OO(d^3 k^{d+1})$ vertices such that $|S_0 \cup S|$ is $\OO(d^3 k^d)$.
	Output the instance $(H, S_0 \cup S, |S_0 \cup S|)$.
	This completes the description of the kernelization algorithm that runs in polynomial-time.
	
	\medskip
	
	Given a $d$-cut $F'$ of $H$, the solution-lifting algorithm invokes Lemma \ref{lemma:vc-soln-lifting-algo-2} and enumerates a collection of $d$-cuts of $G$ with $k^{\OO(d)}$-delay.
	As $n > k^d$, it follows that the delay is $n^{\OO(1)}$.
	Hence, our solution lifting algorithm has polynomial delay.
	Moreover, by Lemma \ref{lemma:vc-soln-lifting-algo-2}, the property (ii) of Definition \ref{defn:poly-delay-enum-kernel} is satisfied.
	
	\end{description}
	As $k \leq 2{\vc}$, {\enummaxdcut}, when parameterized by ${\vc}$ (the vertex cover number) admits a polynomial-delay enumeration kernel with $\OO(d^3 {\vc}^{d+1})$ vertices.
\end{proof}


\fi

\iflong
\section{Parameterization by the Neighborhood Diversity}
\label{sec:neighborhood-diversity-modular-width}
In this section, we consider the neighborhood diversity (${\ndd}$) of the input graph as the parameter.
We assume that the neighborhood diversity of the input graph $G$ is $k$, i.e. ${\ndd}(G) = k$.
Let $G$ be a graph, a neighborhood decomposition with minimum number of modules can be obtained in polynomial-time \cite{Lampis12}.
Let $\UU = \{X_1,\ldots,X_k\}$ be a minimum-size neighborhood decomposition of $G$.
A module $X$ is an {\em independent module} if $G[X]$ is an independent set and is a {\em clique module} if $G[X]$ is a clique.
Let $\GG$ be the {\em quotient graph} for $\UU$, i.e. $\{X_1,\ldots,X_k\}$ is the set of nodes of $\GG$.
Two distinct nodes $X_i$ and $X_j$ are {\em adjacent} in $\GG$ if and only if a vertex of the module $X_i$ and a vertex of the module $X_j$ are adjacent to each other in $G$.
Since $X_i$ and $X_j$ are modules, if a vertex of $X_i$ is adjacent to a vertex in $X_j$, then vertex of $X_i$ is adjacent to every vertex of $X_j$.

A module is {\em trivial} if it is an independent module and is not adjacent to any other module.
Formally, a module $X$ of $G$ is a trivial module if $X$ is an independent set of $G$ and $N_G(X) = \emptyset$.
An independent module $X$ is {\em nice} if it is not a trivial module and has at most $d$ neighbors outside and not a trivial module.
Formally, $X$ is a nice module of $G$ if $X$ is independent and $1 \leq |N_G(X)| \leq d$.

Observe that there can be at most one trivial module in $G$.
Moreover, a trivial module is an isolated node in the quotient graph $\GG$.
If a module $X$ is nice, then it is adjacent to at most $r \leq d$ modules $X_1,\ldots,X_r$ such that $1 \leq \sum_{i=1}^r |X_i| \leq d$.
If a nice module $X$ is adjacent to exactly $r$ modules $X_1,\ldots,X_r$, we say that $\{X_1,\ldots,X_r\}$ is a {\em small module-set}.
If $r = 1$, then for simplicity, we say that $X_1$ is a {\em small module}.
The set of vertices $\cup_{i=1}^r X_i$ spanned by this small module-set $\{X_1,\ldots,X_r\}$ is a {\em small subset}.

\subsection{Marking Scheme and Kernelization Algorithm}
\label{sec:nd-marking-scheme}
Let $(G, \UU, k)$ denote the input instance.
We perform the following procedure ${\sf MarkND}(G, \UU, k)$ that is exploited by all the kernelizations in this section.

\subparagraph{Procedure ${\sf MarkND}(G, \UU, k)$:}
\begin{itemize}
	\item Initialize $H^* := G$ and $S_1 := S_0 := \emptyset$.
	\item If $X$ is a trivial module of $H^*$, then mark an arbitrary vertex of $X$.
	\item Let $X$ be a nice module of $G$.
	If $X$ has at most $3d+2$ vertices, then mark all vertices of $X$.
	Otherwise, if $X$ has more than $3d+2$ vertices, then mark $3d + 2$ vertices of $X$.
	\item If a nice module $X$ has $3d+2$ marked vertices but $|X| > 3d+2$, then choose arbitrary $d$ marked vertices $u_1,\ldots,u_d \in X$.
		For every $i \in [d]$, attach a pendant clique $C_i$ with $2d + i$ vertices into $u_i$.
	Set $S_0 := S_0 \cup C_i$ and $S_1 := S_1 \cup \{u_1,\ldots,u_d\}$.
	We refer to Figure \ref{fig:nd-attach-clique} for an illustration.
	
	Replace the module $X$ by $X \setminus \{u_1,\ldots,u_d\}$.
		Additionally, for every $i \in [d]$, create modules $\{u_i\}$ and $C_i$.
	
	\item If a module $X$ of $G$ is neither trivial nor nice, then mark $\min\{2d+1, |X|\}$ vertices of $X$.
\end{itemize}

\begin{figure}[t]
\centering
	\includegraphics[scale=0.3]{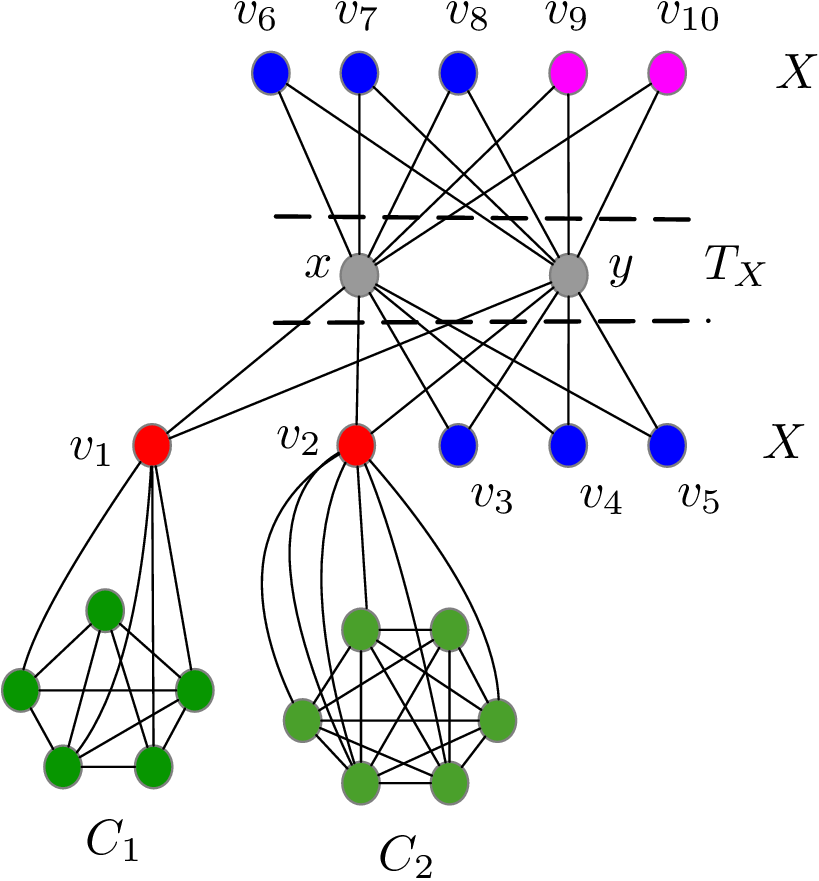}	
	\caption{An illustration of attaching pendant cliques when $d = 2$ and $X$ is a nice module. The input graph $G$ contains all but the pendant cliques $C_1$ and $C_2$. The red and blue vertices from $X$ are marked and $N_G(X) = T_X$. The pink vertex is unmarked.}
\label{fig:nd-attach-clique}
\end{figure}

Let $Z$ be the set of all marked vertices and observe that $S_1 \subseteq Z$.
Consider $H = H^*[Z \cup S_0]$ and $\widetilde{H} = H - S_0$.
We prove the following lemma that gives us an upper bound on the number of vertices and some properties of the output graph $H$ which is the kernel.
It is not very hard to observe that every vertex of $S_1$ is a singleton module that has a pendant clique in $S_0$ attached to it.
We call these modules that are  single vertices of $S_1$ as {\em $S_1$-singleton}. 

\begin{lemma}
\label{lemma:nd-graph-H-property}
Let $(G, \UU, k)$ and $H$ be the output graph.
Then, ${\nd}(H) \leq (2d+1)k$ and $H$ has at most $(3d^2 + 3d+2)k$ vertices.
\end{lemma}

\begin{proof}
First we show the bound on the number of vertices.
Observe that for each of the $k$ modules, at most $3d+2$ vertices have been marked.
If a nice module $X$ has more than $3d+2$ vertices, then it is broken into a module $X'$ of $2d+2$ vertices, and $d$ modules that are in $S_1$.
Additionally, $d$ additional clique modules are constructed in $S_0$ each of which have size at most $3d$.
Therefore, for every module $X$, at most $2d+1$ modules are created.
As $G$ has $k$ modules, the number of vertices in $S_0$ is at most $3d^2 k$.
Therefore, the total number of vertices in $H$ is at most $(3d^2 + 3d + 2)k$.
Observe that every single vertex of $S_1$ is a module in $G$, and every connected component of $S_0$ is a module of $G$.
Since every module $X$ of $G$ produces at most $2d+1$ modules in $H$, the number of modules in $H$ is at most $(2d+1){\nd}(G) = (2d+)k$ and every such module is either a clique or an independent set.
Therefore, ${\nd}(H) \leq (2d+1)k$.
\end{proof}

\subsection{Equivalence Classes and Properties of Minimal $d$-Cuts}
\label{sec:nd-eqv-classes}

Now, we move on to establishing some relations between the $d$-cuts of $H$ and the $d$-cuts of $G$.
Similar as before, we prove the following two items.
\begin{itemize}
	\item The set of all $d$-cuts of $H$ and the set of all $d$-cuts of $\widetilde{H}$ are precisely the same.
	\item Since $\widetilde{H}$ is an induced subgraph of $G$, we provide some specific relations among the $d$-cuts of $G$ and the $d$-cuts of $\widetilde{H}$.
	These specific relations are crucial to establish that the information of every $d$-cut of $G$ is available in $\widetilde{H}$, and hence in $H$.
\end{itemize}

Note that the addition of the pendant cliques in $S_0$ uses similar strategy that we have in Section \ref{sec:vc-marking-scheme}.
By similar arguments as in Observation \ref{obs:vc-S0-empty-H-intersection}, we can prove that every edge of $H$ that is incident to the vertices of $S_0$ are disjoint from every $d$-cut of $H$.
Therefore, we can prove the following lemma using arguments similar to Lemma \ref{lemma:vc-d-cuts-without-S0}.
We state the lemma without proof as the argumentations are similar.

\begin{lemma}
\label{lemma:nd-d-cuts-without-S0}
Let $F$ be an edge set. Then, $F$ is a $d$-cut (minimal or maximal $d$-cut, respectively) of $H$ if and only if $F$ is a $d$-cut (minimal or maximal $d$-cut, respectively) of $\widetilde{H}$.
\end{lemma}

Now, we move on to establish the relations between the $d$-cuts of $\widetilde{H}$ and the $d$-cuts of $G$.
If $X$ is a nice module of $G$,
observe that $X \cap V(\widetilde{H})$ is a nice module in $\widetilde{H}$.
For every nice module $X$ of $G$ (respectively of $\widetilde{H}$), let $T_X = N_G(X)$.
By definition, if $X$ be a nice module of $G$, then $X$ is an independent module and $T_X = N_G(X)$ is a small subset.
Moreover, $T_X$ is the union of the vertex subsets of a collection of modules.
A small subset $T_X$ is {\em bad} if $|X| \geq 3d+2$.
A small subset $T_X$ is {\em good} if $|X| \leq 3d+1$.
Let $\YY$ denotes the collection of all small subsets that are bad and $I(\YY)$ denotes the vertices of all the nice modules $X$ such that $T_X \in \YY$.
Formally, $\YY = \{T_X : |X| \geq 3d+2$ and $X$ is nice$\}$ and $I(\YY) = \{u \in X \mid T_X \in \YY\}$.
Observe that if a small subset $T_X$ is good in $G$, then $T_X$ is good in $\widetilde{H}$.
Similarly, if a small subset $T_X$ is bad in $G$, then $T_X$ is bad in $\widetilde{H}$.
We often use `$T_X$ is a bad subset' to mean that $T_X$ is a small subset that is bad.
Moreover, if a module $X$ is nice in $G$, then $X \setminus V(\widetilde{H})$ remains a nice module in $\widetilde{H}$.

For every $T_X \in \YY$, we define $L(T_X)$ as the set of edges with one endpoint in the nice module $X$ and the other endpoint in $T_X$.
We define $L(\YY) = \cup_{T_X \in \YY} L(T_X)$.
Our next lemma proves that the marking scheme preserves all $d$-cuts of $G - I(\YY)$.

\begin{lemma}
\label{lemma:nd-G-minus-I-Y}
A set of edges $F$ is a $d$-cut of $G - I(\YY)$ if and only if $F$ is a $d$-cut of $\widetilde{H} - I(\YY)$.
\end{lemma}

\begin{proof}
For the simplicity of presentation, let $\vec{H} = \widetilde{H} - I(\YY)$ and $\vec{G} = G - I(\YY)$.

First we give the backward direction ($\Leftarrow$) of the proof.
Let $F$ be a $d$-cut of $\vec{H}$ with the corresponding vertex bipartition $(A', B')$.
We construct a $d$-cut $(A, B)$ of $\vec{G}$ as follows.
Initialize $A := A'$ and $B := B'$.
We consider every vertex $u$ that is in $\vec{G}$ but not in $\vec{H}$ one by one.
If $u \in X$ for some independent module $X$ of $\vec{G}$, then $N_{\vec{G}}(X)$ has at least $d+1$ vertices.
In particular, if $N_{\vec{G}}(X)$ does not contain a nice module $Y$ of $G$.
Therefore, $\widehat{T} = N_{\vec{G}}(X)$ has at least $d+1$ vertices.
Since $u \notin V(\vec{H})$, $X$ must be a module of $G$ with at least $3d+3$ vertices.
Due to Lemma \ref{lem:neighborhood-break-up}, $\widehat{T} \subseteq A'$ or $\widehat{T} \subseteq B'$.
If $\widehat{T} \subseteq A'$, then we set $A := A \cup \{u\}$.
If $\widehat{T} \subseteq B'$, then we set $B := B \cup \{u\}$.
Otherwise, $u \in X$ for a clique module $X$ of $\vec{G}$.
As $X$ is a clique module of $\vec{G}$, it must be that $|X| \geq 2d+2$.
Our marking scheme has marked $2d+1$ vertices from $X$.
Hence, $X \cap \vec{H}$ has $2d+1$ vertices.
Due to Lemma \ref{lem:neighborhood-break-up}, $X \subseteq A'$ or $X \subseteq B'$.
If $X \subseteq A'$, then we set $A := A \cup \{u\}$.
If $X \subseteq B'$, then we set $B := B \cup \{u\}$.
Observe that the above construction of $(A, B)$ ensures us that $E_{\vec{G}}(A, B) = E_{\vec{H}}(A', B') = F$.
Therefore, $F$ is a $d$-cut of $\vec{H}$.

Now, we move on to prove the forward direction ($\Rightarrow$).
Let $F = E_{\vec{G}}(A, B)$ be a $d$-cut of $\vec{G}$.
It is sufficient to prove that every edge $uv \in F$ is an edge of $\vec{H}$, then we are done.
Consider an edge $uv \in F$ such that $uv \notin E(\vec{H})$.
The first case occurs when $u, v \in X$ for a clique module $X$ of $\vec{G}$.
Then $u \notin V(\vec{H})$ or $v \notin V(\vec{H})$ (or both).
Our marking scheme has marked $2d+1$ vertices from $X$, implying that $X$ has more than $2d+1$ vertices in $\vec{G}$.
Since $X$ is a clique with at least $2d+1$ vertices, due to Lemma \ref{lem:neighborhood-break-up}, $X \subseteq A$ or $X \subseteq B$.
Then, $uv \notin F$ which leads to a contradiction.
The second case occurs when $u \in X$ for an independent module $X$ of $\vec{G}$ such that $u$ is not a vertex of $\vec{H}$.
Observe that $\widehat{T} = N_{\vec{G}}(X)$ has at least $d+1$ vertices.
Since $u \notin V(\vec{H})$, $|X| \geq 3d+3$.
Due to Lemma \ref{lem:neighborhood-break-up}, $\widehat{T} \subseteq A$ or $\widehat{T} \subseteq B$.
As $F$ is a $d$-cut and $|\widehat{T}| \geq d+1$, it must be that either $X \cup \widehat{T} \subseteq A$ or $X \cup \widehat{T} \subseteq B$.
In such a case, $uv \notin F$ which leads to a contradiction.
As the above two cases are mutually exhaustive, this completes the proof.
\end{proof}

Let $X$ be a nice module in $G$ (respectively in $\widetilde{H}$) such that $T_X = N_G(X)$ is a small subset that is bad.
For every $u \in X$, define $J(T_X, u)$ the set of edges with one endpoint $u$ and other endpoint in $T_X$.
More formally, $J(T_X, u) = \{uv \mid v \in T_X\}$.
A bad subset $T_X \in \YY$ is {\em broken} by an edge set $F$ if there exists $u \in X$ such that $J(T_X, u) \not\subset F$.
A bad subset $T_X \in \YY$ is {\em occupied} by an edge set $F$ if it is not broken by $F$.
Observe that if a bad subset $T_X \in \YY$ is occupied by $F$, then for every $u \in X$, either $J(T_X, u) \subseteq F$ or $J(T_X, u) \cap F = \emptyset$.
By similar arguments as Observation \ref{obs:broken-case-understanding}, we can prove the following observation.

\begin{observation}
\label{obs:nd-broken-subset-understanding}
If $F$ is a $d$-cut of $G$ (respectively of $\widetilde{H}$) and $T_X \in \YY$ be a bad subset, then $T_X$ is occupied by $F$.
\end{observation}

With the previous definitions and the above observation in hand, we move on to define `equivalence' between two edge sets of $G$ and two edge sets of $\widetilde{H}$ as follows.

\begin{definition}
\label{defn:nd-eqv-d-cuts}
Two edge sets $F_1$ and $F_2$ are {\em equivalent} if the following properties are satisfied.
\begin{itemize}
	\item $F_1 \setminus L(\YY) = F_2 \setminus L(\YY)$, 
	\item for every $T_X \in \YY$, $|F_1 \cap L(T_X)| = |F_2 \cap L(T_X)|$, and
	\item for every $T_X \in \YY$, $T_X$ is occupied by $F_1$ if and only if $T_X$ is occupied by $F_2$.
\end{itemize}	
\end{definition}

Observe that the above definition is an equivalence relation of $d$-cuts.
This definition is not only applicable for two edge sets of $G$ and two edge sets of $\widetilde{H}$, but also applicable between an edge set of $G$ and an edge set of $\widetilde{H}$.
From now onward, for every edge set $F \subseteq E(G)$ (respectively, $F \subseteq E(\widetilde{H})$), we consider a partition $F = F_b \uplus F_g$ such that $F_b = F \cap (\cup_{T_X \in \YY} L(T_X))$ and $F_g = F \setminus L(\YY)$.
Our next lemma proves that an edge set equivalent to a give $d$-cut is also a $d$-cut.

\begin{lemma}
\label{lemma:nd-eqv-connection-H-and-G}
Let $F' \subseteq E(\widetilde{H})$ be a $d$-cut (minimal or maximal $d$-cut respectively) of $\widetilde{H}$.
Then, every $F \subseteq E(G)$ that is equivalent to $F'$ is a $d$-cut (minimal or maximal $d$-cut) of $G$.
\end{lemma}

\begin{proof}
Proof arguments have similarities with the arguments of Lemma \ref{lemma:equivalence-connection-H-G}.
We provide a short sketch here highlighting the differences.
Consider a $d$-cut $F'$ of $\widetilde{H}$ and the partition $F' = F_b' \uplus F_g'$ such that $F_b' = \cup_{T_X \in \YY} (F' \cap L(T_X))$ and $F_g' = F' \setminus L(\YY)$.
The first case occurs when $F_b' = \emptyset$ and $F_g' \neq \emptyset$.
If $F_b' = \emptyset$, then $F_g'$ is a $d$-cut of $\widetilde{H} - I(\YY)$.
Then, due to Lemma \ref{lemma:nd-G-minus-I-Y}, $F_g'$ is a $d$-cut of $G - I(\YY)$.
As $F_b' = \emptyset$, for every marked vertex $u$ from a nice module $X$ such that $T_X \in \YY$, $J(T_X, u) \cap F' = \emptyset$.
Since $F$ is equivalent to $F'$, it must be that for every $v \in X$, $J(T_X, v) \cap F = \emptyset$.
Hence, $F$ is a $d$-cut of $G$. 

The second case occurs when $F_g' = \emptyset$ and $F_b' \neq \emptyset$.
Since $F_b' = \cup_{T_X \in \YY} (F' \cap L(T_X))$, we focus on every $T_X \in \YY$ such that $F' \cap L(T_X) \neq \emptyset$.
Since $F$ is equivalent to $F'$, due to Observation \ref{obs:nd-broken-subset-understanding} it must be that $T_X$ is occupied by $F$ and $|F \cap L(T_X)| = |F' \cap L(T_X)|$.
Hence, for every $u$ in the nice module $X$, $J(T_X, u) \subseteq F$ or $J(T_X, u) \cap F = \emptyset$.
We focus on the number of vertices $v \in X \cap V(\widetilde{H})$ such that $J(T_X, v) \subseteq F'$.
If there are $r$ vertices in $v \in X \cap V(\widetilde{H})$ such that $J(T_X, v) \subseteq F'$, then there are $r$ vertices $w \in X$ such that $J(T_X, w) \subseteq F$.
By arguments similar to Lemma \ref{lemma:equivalence-connection-H-G}, we can prove that every vertex $u$ has at most $d$ edges of $F$ incident to $u$.
Since there is $w \in X$ for some nice module $X$ such that $J(T_X, w) \subseteq F$, the deletion of $F$ increases the number of components in $G$.
Therefore, $F$ is a $d$-cut of $G$.

Finally, the last case occurs when $F_g', F_b' \neq \emptyset$.
As $F$ is equivalent to $F'$, it must be that $F' \setminus L(\YY) = F \setminus L(\YY) = F_g'$.
Moreover, there is a nice module $X$ such that $T_X \in \YY$ and $w \in X$ satisfying $E(T_X, X) \cap  F \neq \emptyset$.
Therefore, there is $w \in X$ such that $J(T_X, w) \subseteq F$.
By arguments similar to Case (iii) of Lemma \ref{lemma:equivalence-connection-H-G} proof, we can argue that $F$ is a $d$-cut of $G$.
\end{proof}

\subparagraph*{Special Property of Minimal $d$-Cuts.}
Now, we illustrate a special property of minimal $d$-cut that we will need to design the solution-lifting algorithm.
Consider a nice module $X$ with more than $3d+3$ vertices such that $T_X \in \YY$.
Let $F = F_b \uplus F_g$ be a minimal $d$-cut of $G$ (respectively, a minimal $d$-cut of $\widetilde{H}$) such that $F_b = F \cap (\cup_{T_X \in \YY} L(T_X))$.
As $F$ is a $d$-cut of $G$ (respectively a $d$-cut of $\widetilde{H}$), due to Observation \ref{obs:nd-broken-subset-understanding}, $T_X$ is occupied by $F$.
We can use this observation to prove the following observation that is crucial for every minimal $d$-cut of $G$ (respectively of $\widetilde{H}$).

\begin{observation}
\label{obs:nd-special-minimal-d-cut-property}
Let $F$ be a minimal $d$-cut of a connected graph $G$ (respectively, a connected graph $\widetilde{H}$).
Consider the partition $F = F_b \uplus F_g$ such that $F_b = F \cap (\cup_{T_X \in \YY} L(T_X))$.
If $F_b \neq \emptyset$, then $F_g = \emptyset$.
\end{observation}

 \begin{proof}
 We give the proof for the graph $G$. 
 Let $F = F_b \uplus F_g$ be a minimal $d$-cut of a connected graph.
 Then, $F \neq \emptyset$.
 Assume for the sake of contradiction that $F_g, F_b \neq \emptyset$.
 Furthermore, we assume that the vertex bipartition for $F$ is $(A, B)$.
 Since $F_b \neq \emptyset$, there is $T_X \in \YY$ such that $F \cap L(T_X) \neq \emptyset$.
 As $T_X \in \YY$ is a bad subset, due to Observation \ref{obs:broken-case-understanding}, for every $u \in X$, either $J(T_X, u) \subseteq F$ or $J(T_X, u) \cap F = \emptyset$.
 Since $F \cap L(T_X) \neq \emptyset$, there exists $u \in X$ such that $J(T_X, u) \subseteq F$ and $J(T_X, u) \neq \emptyset$.
 Then, $T_X \subseteq A$ and $u \in B$.
 But, $F_g \neq \emptyset$.
 Therefore, moving $u$ from $B$ to $A$ preserves the property that the resulting partition $(A \cup \{u\}, B \setminus \{u\})$ corresponds to a nonempty $d$-cut that contains $F_g (\neq \emptyset)$ and the edge set of which is a proper subset of $F$.
 This contradicts the minimality of $F$.
 Therefore, if $F_g \neq \emptyset$, then $F_b = \emptyset$.
 \end{proof}

\subsection{Defining the Notion of Suitable to Avoid Duplicate Enumeration}
\label{sec:nd-duplicate-challenge}

It is possible that there are two distinct $d$-cuts $F_1$ and $F_2$ of $H$ (hence, $d$-cuts of $\widetilde{H}$) such that there is a $d$-cut $F$ of $G$ that is equivalent to both $F_1$ and $F_2$.
We need to ensure in our solution-lifting algorithm that $F$ should be outputted only when one of $F_1$ or $F_2$ is given but not both.
To circumvent this challenge, we define the notion of suitable bad subsets.
Before that, we illustrate a small difference in the modular decompositions of $H$ and $\widetilde{H}$.
Recall from {\sf MarkND}($G, \UU, k$) that if $X$ is a nice module of $G$ with more than $3d+2$ vertices, then $X \cap V(\widetilde{H})$ has been broken into a module $\widetilde{X}$ of $2d+2$ vertices, and each of the $d$ other marked vertices from $X$ have become singleton modules that are in $S_1$.
We use $S_1(X)$ to denote these $d$ singleton modules.
Observe that for every singleton module $z$ in $S_1(X)$, there is a pendant clique in $S_0$ attached to $z$ and $N_{\widetilde{H}}(z) = T_X$.

\begin{definition}
\label{defn:nd-suitable-bad-subsets}
Let $Y$ be a nice module in $G$ such that $T_Y \in \YY$ is a bad subset and $F'$ is a $d$-cut of $H$.
Then, $T_Y$ is said to be {\em $r$-suitable} for some $r > 0$ with respect to $F'$ if
\begin{itemize}
	\item $F' \cap L(T_Y) \neq \emptyset$,
	\item there is a set of $r$ singleton modules $v_1,\ldots,v_r \in S_1(Y)$ in $H$ such that for every $i \in [r]$, $J(T_Y, v_i) \subseteq F'$, and the pendant clique of $S_0$ attached to $v_i$ has exactly $2d+i$ vertices, and
	\item for every $w \in Y \setminus S_1(Y)$, $J(T_Y, w) \cap F' = \emptyset$.
\end{itemize}
We refer to Figure \ref{fig:nd-suitable-subsets} for an illustration.
\end{definition}

\begin{figure}[t]
\centering
	\includegraphics[scale=0.3]{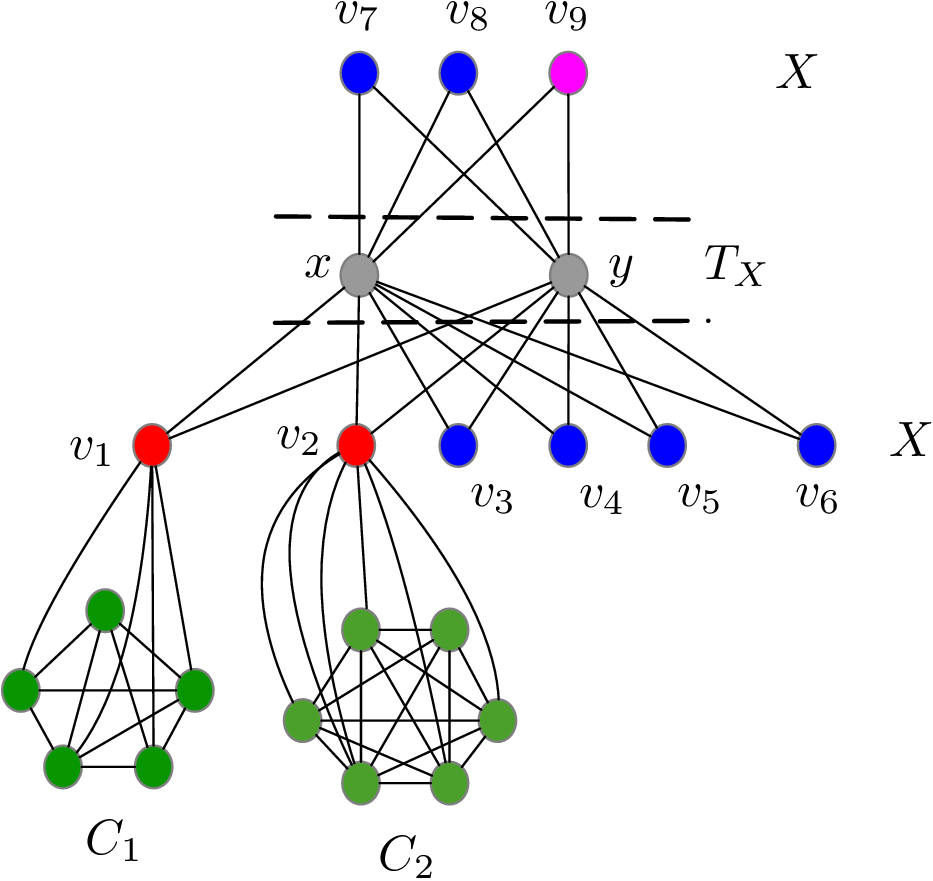}
	\caption{An illustration of suitable subsets for $d=2$.
	The graph $G$ contains all the vertices except the two pendant cliques $C_1, C_2$. The graph $H$ contains all but the pink vertices $u_7$ and $u_8$. 
	The set $T_X$ is 1-suitable with respect to the $d$-cut $F_1 = E_H(A_1, B_1)$ of $H$ such that $B_1 = C_1 \cup \{u_1\}$ and $A_1 = V(H) \setminus B_1$.
	In addition, $T$ is 2-suitable with respect to the $d$-cut $F_2 = E_H(A_2, B_2)$ of $H$ such that $B_2 = C_1 \cup C_2 \cup \{u_1, u_2\}$ and $A_2 = V(H) \setminus B_2$.
	Finally, $T$ is unsuitable with respect to the $d$-cut $F_3 = E_H(A_3, B_3)$ such that $B_3 = C_2 \cup \{u_2, u_3\}$ and $A_3 = V(H) \setminus B_3$.
	We also provide an illustration of Lemma \ref{lemma:nd-equivalence-existence-d-cuts}.
	Consider the $d$-cut $F = E_G(A, B)$ of $G$ such that $B = \{u_2, u_7\}$ and $A = V(G) \setminus B$.
	Then, $F_2$ is a $d$-cut of $G$  equivalent to $F$ and $T$ is 2-suitable with respect to $F_2$.}
\label{fig:nd-suitable-subsets}
\end{figure}

If a bad subset $T_Y \in \YY$ is not $r$-suitable for any $r > 0$ with respect to the $d$-cut $F'$, then $T_Y$ is {\em unsuitable} with respect to $F'$.
Using this notion of $r$-suitable bad subsets, next lemma proves that for every $d$-cut of $G$, there is a $d$-cut in $H$ satisfying some special properties.

\begin{lemma}
\label{lemma:nd-equivalence-existence-d-cuts}
Let $F$ be an edge set of $G$.
Then, the following statements hold true.
\begin{enumerate}[(i)]
	\item If $F$ is a $d$-cut (minimal or maximal $d$-cut, respectively) of $G$, and $F \subseteq E(\widetilde{H})$, then $F$ is a $d$-cut (minimal or maximal $d$-cut, respectively) of $H$.
	\item If $F$ is a $d$-cut (maximal $d$-cut, respectively) of $G$, and $F \not\subset E(\widetilde{H})$, then $H$ has a unique $d$-cut (maximal $d$-cut, respectively) $F'$ such that
	\begin{itemize}
		\item $F'$ is equivalent to $F$, and
		\item for every nice module $X$ of $G$, if $T_X \in \YY$ and $F \cap L(T_X) \not\subset \widetilde{H}$, then $T_X$ is $r$-suitable for some $r > 0$ with respect to $F'$.
	\end{itemize} 
	\item If $F$ is a minimal $d$-cut of $G$ containing an edge not in $\widetilde{H}$, then there is a unique nice module $X$ of $G$ such that $T_X \in \YY$ and $H$ has a unique minimal $d$-cut $F'$ that is equivalent to $F$ and $T_X$ is 1-suitable with respect to $F'$.
\end{enumerate}
\end{lemma}

\begin{proof}
Let $F \subseteq E(G)$.
We prove the items in the given order.
\begin{enumerate}[(i)]
	\item Suppose that $F$ is a $d$-cut of $G$ and $F \subseteq E(\widetilde{H})$.
	Since $\widetilde{H}$ is a subgraph of $G$, $F$ is a $d$-cut of $\widetilde{H}$.
	Due to Lemma \ref{lemma:nd-d-cuts-without-S0}, $F$ is a $d$-cut of $H$.
	The proof arguments for minimal and maximal $d$-cuts are exactly the same.
	
	\item We give the proof for $d$-cuts.
	The arguments for maximal $d$-cuts are exactly the same.
	Let $F$ be a $d$-cut of $G$ containing an edge not in $\widetilde{H}$.
	Consider the partition $F = F_b \uplus F_g$ such that $F_g = F \setminus L(\YY)$.
	As $F$ contains an edge not in $\widetilde{H}$, it must be that $F_b \neq \emptyset$.
	We construct a $d$-cut $F'$ of $H$ satisfying the desired properties.
	
	Initialize $F' := F$.
	In particular, there is a nice module $X$ of $G$ such that $T_X \in \YY$ and $F \cap L(T_X)$ contains an edge not in $\widetilde{H}$.
	Then, there are $r > 0$ vertices $w \in X$ such that $J(T_X, w) \subseteq F$.
	Exactly $3d+2$ vertices from $X$ are marked and $S_1(X)$ contains $d$ vertices $v_1,\ldots,v_d$ such that there is a pendant clique in $S_0$ with $2d+i$ vertices attached to $v_i$.
	We choose $v_1,\ldots,v_r$ from $S_1(X)$ and set $F' := (F' \setminus L(T_X)) \cup (\bigcup\limits_{i \in [r]} J(T_X, v_i)$.
	
	We repeat this procedure for every nice module $X$ of $G$ such that $T_X \in \YY$ and $F \cap L(T_X)$ contains an edge not in $\widetilde{H}$.
	As two distinct nice modules are pairwise disjoint, the replacement procedure involves edge sets that are pairwise disjoint from each other. 
	This replacement procedure keeps $F_g$ unchanged and ensures us that for every nice module $X$ of $G$, if $T_X \in \YY$ and $F \cap L(T_X)$ has an edge not in $\widetilde{H}$, then $|F' \cap L(T_X)| = |F \cap L(T_X)|$.
	Therefore, $F'$ is equivalent to $F$.
	It is not very hard to observe that for every vertex $u \in V(\widetilde{H})$, at most $d$ edges of $F'$ are incident to $u$.
	As $F_b' \neq \emptyset$ and there is $v \in X \cap V(\widetilde{H})$ such that $J(T_X, v) \subseteq F'$, the deletion of $F'$ increases the number of connected components of $\widetilde{H}$.
	Therefore, $F'$ is a $d$-cut of $\widetilde{H}$.
	
	\item Let $F$ be a minimal $d$-cut of $G$ containing an edge not in $\widetilde{H}$.
	Consider the partition $F = F_b \uplus F_g$ such that $F_g = F \setminus L(\YY)$.
	As $F$ contains an edge not in $\widetilde{H}$, it must be that $F_b \neq \emptyset$.
	Then, due to Observation \ref{obs:nd-broken-subset-understanding}, $F_g = \emptyset$.
	In particular, there is a nice module $X$ of $G$ and an unmarked vertex $u \in X$ (that is not in $\widetilde{H}$) such that $J(T_X, u) \subseteq F$.
	Due to the minimality of $F$, $J(T_X, u) = F$ and $X$ is unique.
	Then, exactly $3d+2$ vertices from $X$ are marked and there is a marked vertex $v \in X \cap V(\widetilde{H})$ that has a pendant clique in $S_0$ with exactly $2d+1$ vertices attached to $v$.
	We consider the edge set $F' = J(T_X, v)$.
	Note that $F'$ is a minimal $d$-cut of $\widetilde{H}$ and is equivalent to $F$.
	Due to Lemma \ref{lemma:nd-d-cuts-without-S0}, $F'$ is a minimal $d$-cut of $H$.
	As the pendant clique of $S_0$ attached to $v$ has $2d+1$ vertices, $T_X$ is 1-suitable with respect to $F'$. 
\end{enumerate}
This completes the proof of the lemma.
\end{proof}

\subsection{Designing the Enumeration Kernelizations}
\label{sec:nd-enum-kernels}

Now, we are ready to use the notion of suitable as defined in the previous section, and describe our solution-lifting algorithm that will be central to prove our result.

\begin{lemma}
\label{lemma:nd-miniml-d-cut-enumeration}
Let $\FF(H)$ be the collection of all minimal $d$-cuts of $H$ and $\FF(G)$ be the collection of all minimal $d$-cuts of $G$.
Then, given $F' \in \CC(F)$ there is an algorithm that enumerates a collection $\CC(F')$ of minimal $d$-cuts of $G$ in polynomial-time such that
\begin{itemize}
	\item every $F \in \CC(F')$ is equivalent to $F'$, and
	\item if $F \in \CC(F')$ such that $F \neq F'$, then $F'$ contains an edge not in $H$.
\end{itemize} 
Furthermore, $\{\CC(F') \mid F' \in \FF(H)\}$ is a partition of $\FF(G)$.
\end{lemma}

\begin{proof}
Let $\FF(H)$ be the collection of all minimal $d$-cuts of $H$ and $\FF(G)$ be the collection of all minimal $d$-cuts of $G$.
Given a minimal $d$-cut $F'$ of $H$, due to Lemma \ref{lemma:nd-d-cuts-without-S0}, it follows that $F'$ is a minimal $d$-cut of $\widetilde{H}$.
We consider the partition $F' = F_b' \uplus F_g'$ such that $F_b' = F' \cap (\cup_{T_X \in \YY} L(T_X))$.
Due to Observation \ref{obs:nd-special-minimal-d-cut-property}, $F_b' = \emptyset$ or $F_g' = \emptyset$ or both.
If $F_g' = F_b' = \emptyset$, then $\emptyset$ is the unique minimal $d$-cut of $H$.
Our enumeration algorithm just outputs $\emptyset$ that is clearly equivalent to $F'$ and is a unique minimal $d$-cut of $G$.

We consider the case $F_g' \neq \emptyset$.
Due to Observation \ref{obs:nd-special-minimal-d-cut-property}, $F_b' = \emptyset$.
Our enumeration algorithm just outputs $F = F'$ that is same as $F_g'$.
As $F_g' \subseteq E(\widetilde{H}) - I(\YY)$, $F_g'$ is a $d$-cut of $\widetilde{H} - I(\YY)$.
Due to Lemma \ref{lemma:nd-G-minus-I-Y}, $F_g'$ is a $d$-cut of $G - I(\YY)$.
Since $F (= F')$ is equivalent to $F'$ and $F'$ is a minimal $d$-cut of $\widetilde{H}$, due to Lemma \ref{lemma:nd-eqv-connection-H-and-G}, $F$ is a $d$-cut of $G$.

Finally, we consider when $F_b' \neq \emptyset$.
Due to Observation \ref{obs:nd-special-minimal-d-cut-property}, $F_g' = \emptyset$.
In fact, due to the minimality of $F'$ and $F_b' \neq \emptyset$, there is a unique nice module $X$ of $G$ such that $J(T_X, u) = F'$ for some $u \in V(\widetilde{H}) \cap X$.
We check if $T_X$ is 1-suitable with respect to $F'$.
We can check this in polynomial-time by looking at the pendant clique of $S_0$ attached to $u$.
If the size of the attached pendant clique has exactly $2d+1$ vertices then $T_X$ is 1-suitable with respect to $F'$, otherwise $T_X$ is unsuitable.
If $T_X$ is 1-suitable with respect to $F'$, then output $F'$ and $J(T_X, w)$ for every $w \in X \setminus V(\widetilde{H})$.
Observe that every outputted edge set is a minimal $d$-cut that is equivalent to $F'$.
In fact, if an outputted $d$-cut is not $F'$ itself, then it contains an edge not in $H$.
If $T_X$ is not suitable, then just output $F'$ itself.

Note that in each of the above cases, the collection of $d$-cuts can be outputted in polynomial-time.
We move on to establish that $\{\CC(F') \mid F' \in \FF(H)\}$ is a partition of $\FF(G)$.
Let $F \in \FF(G)$.
If $F \subseteq E(\widetilde{H})$, then due to Lemma \ref{lemma:nd-equivalence-existence-d-cuts}, $F$ is a $d$-cut of $H$, and $F \in \CC(F)$.
If $F$ contains an edge not in $\widetilde{H}$, then due to Lemma \ref{lemma:nd-equivalence-existence-d-cuts}, there is a unique nice module $X$ of $G$ and $T_X \in \YY$ such that $H$ has a unique minimal $d$-cut $F'$ equivalent to $F$ and $T_X$ is 1-suitable with respect to $F'$.
From the description of our enumeration algorithm, note that $F$ is outputted only when $F'$ is given.
Therefore, $\{\CC(F') \mid F' \in \FF(H)\}$ is a partition of $\FF(G)$.
\end{proof}

\begin{lemma}
\label{lemma:nd-all-d-cut-enumeration}
Let $\FF(H)$ be the collection of all $d$-cuts (respectively, maximal $d$-cuts) of $H$ and $\FF(G)$ be the collection of all $d$-cuts (respectively, maximal $d$-cuts) of $G$.
Then, given $F' \in \CC(F)$ there is an algorithm that enumerates a collection $\CC(F')$ of $d$-cuts (respectively, maximal $d$-cuts) of $G$ in $n^{\OO(1)}$-delay such that
\begin{itemize}
	\item every $F \in \CC(F')$ is equivalent to $F'$, and
	\item if $F \in \CC(F')$ such that $F \neq F'$, then $F'$ contains an edge not in $H$.
\end{itemize} 
Furthermore, $\{\CC(F') \mid F' \in \FF(H)\}$ is a partition of $\FF(G)$.
\end{lemma}

\begin{proof}
Let $\FF(H)$ be the collection of all $d$-cuts of $H$ and $\FF(G)$ be the collection of all $d$-cuts of $G$.
We describe the proof for $d$-cuts of $H$.
The arguments for maximal $d$-cuts are exactly the same.

First, we describe the enumeration algorithm.
Let $F'$ be a $d$-cut of $H$.
Due to Lemma \ref{lemma:nd-d-cuts-without-S0}, $F'$ is a $d$-cut of $\widetilde{H}$.
If $F' = \emptyset$, our algorithm only outputs $F = F'$ that is the only edge set equivalent to $F'$.
Due to Lemma \ref{lemma:nd-eqv-connection-H-and-G}, $F$ is a $d$-cut of $G$.
Otherwise, we consider the partition $F' = F_b' \uplus F_g'$.
If $F_b' = \emptyset$, then our algorithm just outputs $F = F'$.
Observe that this is the only edge set equivalent to $F'$.
Due to Lemma \ref{lemma:nd-eqv-connection-H-and-G}, $F$ is a $d$-cut of $G$.

We consider the case when $F_b' \neq \emptyset$.
For every nice module $X$ of $G$, we look at $T_X = N_G(X)$ one by one such that $T_X \in \YY$.
There are at most $k$ distinct nice modules and hence, at most $k$ distinct small subsets. 
\begin{enumerate}
	\item We initialize $\DD$, the collection of all small subsets $T_X \in \YY$ that are bad and $F' \cap L(T_X) \neq \emptyset$ and $J = F_g'$.
	
	\item As a base case, if $\DD = \emptyset$, then we return $J$.
	Otherwise, we perform the following steps.
	
	\item We look at every $T_X \in \YY$ one by one $F(T_X) = F' \cap L(T_X) \neq \emptyset$.
	Due to the marking scheme {\sf MarkND}($G, \UU, k$), exactly $3d+2$ vertices from $X$ are marked but $X$ has some unmarked vertex in $G$.
	We check if $T_X$ is $r$-suitable for some $r > 0$ with respect to $F'$.
	We can check this in polynomial-time as follows.
	First, we look at $P \subseteq X \cap V(\widetilde{H})$ such that $E(P, T_X) = F(T_X)$.
	Since $F(T_X) \neq \emptyset$, $|P| = h > 0$.
	If there are pendant cliques of $S_0$ attached to the vertices of $P$ and the sizes of those pendant cliques are $\{2d+1,\ldots,2d+h\}$, then $T_X$ is $h$-suitable with respect to $F'$.
	Otherwise, $T_X$ is not suitable.
	If $T_X$ is $r$-suitable for some $r > 0$, then we enumerate $P$ and all subsets $P'$ of $h$ vertices from $X$ that contains a vertex not in $H$.
	For each such enumerated subset $P'$ of $h$ vertices (including $P' = P$ itself), we set $F^* = E(P', T_X)$ and recursively call the algorithm with $J := J \cup F^*$, and $\DD := \DD \setminus \{T_X\}$.
	If $T_X$ is not suitable, then we just recursively call the algorithm with $J := J \cup F(T_X)$, and $\DD := \DD \setminus \{T_X\}$.	
\end{enumerate}

The above algorithm is a backtracking enumeration algorithm and it enumerates edge sets that are equivalent to $F'$.
Due to Lemma \ref{lemma:nd-eqv-connection-H-and-G}, every enumerated edge set is a $d$-cut of $G$.
Observe that the depth of the search tree is $k = {\nd}(G)$, and there are $n$ branches.
In every node, the time taken is polynomial, and we output the $d$-cuts only in the leaves.
Moreover, the time take between two consecutive leaves is polynomial.
As $k < n$, it follows that the delay is $n^{\OO(1)}$.

We are yet to establish that $\{\CC(F') \mid F' \in \FF(H)\}$ is a partition of $\FF(G)$.
Consider $F \in \FF(G)$.
If $F \subseteq E(\widetilde{H})$, then due to Lemma \ref{lemma:nd-equivalence-existence-d-cuts}, $F$ is a $d$-cut of $H$.
Therefore, $F \in \CC(F)$.
If $F$ contains an edge not in $\widetilde{H}$, then due to Lemma \ref{lemma:nd-equivalence-existence-d-cuts}, $H$ has a unique $d$-cut $F'$ such that for every nice module $X$ of $G$, if $T_X \in \YY$ and $F \cap T_X$ has an edge not in $\widetilde{H}$, then $T_X$ is $r$-suitable with respect to $F'$ for some $r > 0$.
Based on the description of the enumeration algorithm, when $F'$ is given, $F$ is outputted.
Therefore, $F \in \CC(F')$.
As $F'$ is unique, this completes the proof that $\{\CC(F') \mid F' \in \FF(H)\}$ is a partition of $\FF(G)$.
\end{proof}

With the above mentioned lemmas in hand, we are ready to prove our result in this section.

{\ThmNDComb*}

\begin{proof}
Each of the enumeration kernelizations have two parts, kernelization algorithm, and solution-lifting algorithm.
\begin{itemize}
	\item We give the proof for {\enummindcut}.
	Given $(G, \UU, k)$, our kernelization algorithm invokes {\sf MarkND}($G, \UU, k)$ and outputs the graph $H$.
	Due to Lemma \ref{lemma:nd-graph-H-property}, ${\ndd}(H) \leq (2d+1)k$ and $H$ has $\OO(d^2 k)$ vertices.
	Given a minimal $d$-cut $F'$ of $H$, our solution-lifting algorithm invokes Lemma \ref{lemma:nd-miniml-d-cut-enumeration} and outputs a collection of $d$-cuts in polynomial-time.
	Due to Lemma \ref{lemma:nd-miniml-d-cut-enumeration}, the property (ii*) of fully-polynomial enumeration kernel is satisfied.
	As $k = {\ndd}(G)$, {\enummindcut} parameterized by ${\ndd}(G)$ admits a fully-polynomial enumeration kernel with $\OO(d^2 {\ndd})$ vertices.
	
	\item We give the proof for {\enumdcut}. The proof for {\enummaxdcut} is exactly the same.
	Given $(G, \UU, k)$, our kernelization algorithm invokes marking scheme {\sf MarkND}($G, \UU, k$) and outputs $H$.
	Due to Lemma \ref{lemma:nd-graph-H-property}, ${\ndd}(H) \leq (2d+1)k$ and $H$ has $\OO(d^2 k)$ vertices.
	Given a $d$-cut $F'$ of $H$, the solution-lifting algorithm invokes Lemma \ref{lemma:nd-all-d-cut-enumeration} to output a collection of $d$-cuts of $G$ with polynomial-delay that satisfies (ii) of Definition \ref{defn:poly-delay-enum-kernel}.
	As $k = {\ndd}(G)$, the kernelization algorithm and solution-lifting algorithm together constitute a polynomial-delay enumeration kernel with $\OO(d^2 {\ndd})$ vertices.
\end{itemize}
This completes the proof of the theorem.
\end{proof}

\section{Parameterization by the Clique Partition Number}
\label{sec:clique-partition}

In this section, we consider the clique partition number, i.e. ${\clp}$ of the input graph as the parameter.
Since the clique partition number of $G$ is the complement of the chromatic number of the complement graph of $G$, for any $k \geq 3$, it is NP-complete to decide whether ${\clp}(G) = k$.
Therefore, we assume that a clique partition $\{C_1,\ldots,C_k\}$ of $G$ with the minimum number of cliques is given to us with the input graph.
We denote $(G, \CC, k)$ the input instance where $\CC$ is a clique-partition of $G$ into $k$ cliques.
Clearly, by our assumption ${\clp}(G) = k$.
We call a clique $C \in \CC$ {\em large} if $|C| \geq 2d+1$.
If a clique $C \in \CC$ is not a large clique, we call it a {\em small} clique.

\begin{reduction rule}
\label{rule:splitting-small-cliques-partition}
Let $C = \{u_1,\ldots,u_i\}$ be a small clique.
Then, replace $C$ by $i$ distinct cliques $\{u_1\},\ldots,\{u_i\}$.
\end{reduction rule}

Observe that the above reduction rule does not change the graph.
The collection of cliques in $\CC$ gets updated and $|\CC|$ increases to at most $2dk$.
We first apply Reduction Rule \ref{rule:splitting-small-cliques-partition} exhaustively.
We do not apply this above reduction rule anymore after that.
When the above reduction rule is not applicable, it holds that any small clique has just one vertex.
The next reduction rule is the following.

\begin{reduction rule}
\label{rule:clique-partition-merging-large-cliques}
If there are two cliques $C_i, C_j \in \CC$ such that there is $u \in C_i$ having more than $d$ neighbors in $C_j$ or there is $u \in C_j$ having more than $d$ neighbors in $C_i$, then make every vertex of $C_i$ adjacent to every vertex of $C_j$.
Replace $C_i$ and $C_j$ from $\CC$ by a new clique $C' = C_i \cup C_j$.
Reduce $k$ by 1.
\end{reduction rule}

We apply Reduction Rule \ref{rule:clique-partition-merging-large-cliques} only when Reduction Rule \ref{rule:splitting-small-cliques-partition} is not applicable.
When the above reduction rule is applicable for a pair of cliques $C_i, C_j \in \CC$, then either $C_i$ or $C_j$ (or both) is (or are) a large clique(s). 
Observe that for any $d$-cut $(A, B)$ and for any (large) clique $C$, it holds that $C \subseteq A$ or $C \subseteq B$.
Hence, it is not possible that $C_i \subseteq A$ and $C_j \subseteq B$.
It means that either $C_i \cup C_j \subseteq A$ or $C_i \cup C_j \subseteq B$.
We apply Reduction Rules \ref{rule:splitting-small-cliques-partition} exhaustively followed by the Reduction Rule \ref{rule:clique-partition-merging-large-cliques} exhaustively to obtain the instance $(G_1, \CC_1, k_1)$.
After that we do not apply the above mentioned reduction rules anymore.
Let $(G_1, \CC_1, k_1)$ be the obtained instance after exhaustive application of Reduction Rule \ref{rule:clique-partition-merging-large-cliques} and let $\CC_1$ be the set of obtained cliques.
Note that for every two distinct cliques $C_i, C_j \in \CC_1$, any $u \in C_i$ or $u \in C_j$ can have at most $d$ neighbors in $C_j$ or in $C_i$, respectively.
But it is possible that some $u \in C_i$ may have more than $d$ neighbors outside $C_i$.
If a vertex $u \in C_i$ has at most $d$ neighbors outside $C_i$, then we say that $u$ is a {\em low} vertex.
On the other hand if $u \in C_i$ has more than $d$ neighbors outside $C_i$, then we say that $u$ is a {\em high} vertex.
%
%
\begin{lemma}
\label{lemma:clique-partition-step-1-preservation}
A set of edges $F \subseteq E(G)$ is a $d$-cut of $G$ if and only if $F$ is a $d$-cut of $G_1$.
\end{lemma}

\begin{proof}
Let $F = (A, B)$ be a $d$-cut of $G$.
Observe that Reduction Rule \ref{rule:splitting-small-cliques-partition} does not change the graph.
We apply Reduction Rule \ref{rule:clique-partition-merging-large-cliques} only after an exhaustive application of Reduction Rule \ref{rule:splitting-small-cliques-partition}.
Consider two cliques $C, D$ such that $E(C, D)$ is not a $d$-cut.
Then, it must be that $D$ (or $C$) is a large clique.
It implies that $C \cup D \subseteq A$ or $C \cup D \subseteq B$.
Hence, no edge $uv \in E(G_1) \setminus E(G)$ can be present in a $d$-cut of $G$.
It means that $F$ is a $d$-cut of $G_1$.
By similar argument as above, we can prove that any $d$-cut of $G_1$ is also a $d$-cut of $G$.
\end{proof}

%
%

After we obtain the instance $(G_1, \CC_1, k_1)$, we invoke a marking procedure ${\sf MarkCP}(G_1, \CC_1, k_1)$ that works as follows.
\begin{enumerate}
	\item Suppose there is a clique $C \in \CC_1$ and $i$ distinct large cliques $C_1, C_2,\ldots,C_i \in \CC_1$ such that $C$ has a low vertex $u$ having neighbors only in each of the cliques $C_1,\ldots,C_i$.
	Formally, we assume that there is $u \in C$ such that $u$ is a low vertex, $N_G(u) \setminus C \subseteq \bigcup\limits_{j=1}^i C_j$, and $\sum\limits_{j=1}^i |N_G(u) \cap C_j| = \ell \leq d$.
	Then we mark $u$ and the vertices in $\bigcup\limits_{j=1}^i (N_G(u) \cap C_j)$.
	Repeat this marking scheme for every $i \leq d$ cliques and every $\ell \leq d$.
	\item If there is a clique $C \in \CC_1$ and $i$ distinct large cliques $C_1, C_2,\ldots,C_i \in \CC_1$ such that $C$ has a high vertex $u$ satisfying $\sum\limits_{j = 1}^i |N_G(u) \cap C_j| > d$ but $\sum\limits_{j = 1}^{i-1} |N_G(u) \cap C_j| \leq d$ then mark $u$ and all the vertices in $\bigcup\limits_{j=1}^{i} (N_G(u) \cap C_j)$.
	Let $\bigcup\limits_{j=1}^{i} (N_G(u) \cap C_j)$ has $\ell$ vertices.
	Informally, if for a clique $C \in \CC_1$ and $i$ distinct large cliques $C_1,\ldots,C_i \in \CC_1$, it happens that $u \in C$ is a high vertex but has at most $d$ cumulative neighbors in $C_1,\ldots,C_{i-1}$ but has more than $d$ cumulative neighbors in $C_1,\ldots,C_i$, then mark $u$ and all its neighbors from $C_1,\ldots,C_i$.
	Repeat this process for every $i \leq d+1$ and observe that $\ell \leq 2d$.
\end{enumerate}

Informally, the above marking scheme has marked $u \in C$ and all the neighbors of $u$ outside $C$ in other large cliques if there are at most $d$ distinct neighbors.
Otherwise, it has marked $u \in C$ and at most $2d$ neighbors in other large cliques. 
After performing ${\sf MarkCP}(G_1, \CC_1, k_1)$, we apply the following reduction rule that removes some unnecessary unmarked vertices from large cliques.

\begin{reduction rule}
\label{rule:clique-partition-marking-vertices}
Let $C$ be a large clique in $\CC_1$ and $X \subseteq C$ be the unmarked vertices of $C$.
Remove $\min\{|X|, |C| - 2d - 1\}$ unmarked vertices from $C$.
\end{reduction rule}

Let $(G', \CC', k')$ be the obtained after applying the marking scheme ${\sf MarkCP}(G_1, \CC_1, k_1)$ followed by an exhaustive application of Reduction Rule \ref{rule:clique-partition-marking-vertices} on $(G_1, \CC_1, k_1)$.

Let $\CC'$ be the set of cliques such that $\bigcup\limits_{C \in \CC'} V(C) = V(G')$.
The following observation is simple to see.

\begin{obs}
\label{obs:d-cut-property-clique-partition}
Let $F = (A', B')$ be a $d$-cut of $G'$ and $\CC' = \{C_1',\ldots,C_r'\}$ such that $r \leq (2d+1)k$.
Then, there is a partition $I \uplus J = \{1,\ldots,r\}$ such that
$A' = \bigcup_{i \in I} C_i'$ and $B' = \bigcup_{j \in J} C_j'$.
\end{obs}

\noindent
We prove the following two lemmas that are among the critical ones to prove our final theorem statement.


\begin{lemma}
\label{lemma:clique-partition-d-cut-preservation}
Let $(G, \CC, k)$ be an instance of {\enumdcut} (or {\enummindcut} or {\enummaxdcut}) such that $V(G)$ is partitioned into a set $\CC$ of $k$ cliques. 
There is an algorithm that runs in $k^{d+1}n^{\OO(1)}$-time and constructs $(G', \CC', k')$ such that $k' \leq k$ and $|V(G')|$ is $\OO(k^{d+2})$.
\end{lemma}

\begin{proof}
Let us first describe the algorithm.
Given $(G, \CC, k)$, at the first phase, the algorithm performs Reduction Rules \ref{rule:splitting-small-cliques-partition} exhaustively followed by Reduction Rule \ref{rule:clique-partition-merging-large-cliques} exhaustively to obtain $(G_1, \CC_1, k_1)$ such that $V(G_1)$ can be partitioned into a set $\CC_1$ of $k_1 \leq 2dk$ cliques.
In the second phase, the algorithm performs the marking scheme ${\sf MarkCP}(G_1, \CC_1, k_1)$.
In the third phase, the algorithm applies Reduction Rule \ref{rule:clique-partition-marking-vertices} on $(G_1,\CC_1, k_1)$ exhaustively to obtain $(G', \CC', k')$ such that $V(G')$ is partitioned into a set $\CC'$ of $k'$ cliques.
By construction, $k' = k_1 \leq 2dk$.
Furthermore, the subroutine ${\sf MarkCP}(G_1, \CC_1, k_1)$ takes $k^{d+1} n^{\OO(1)}$-many operations.
This completes the description of the polynomial-time algorithm to construct the output instance $(G', \CC', k')$.

Next, we justify that $|V(G')|$ is $\OO(k^{d+2})$.
Observe that by construction, if $C \in \CC$ is a small clique of $G$, then $C$ is broken into $|C|$ different small cliques in $\CC'$ and the resultant cliques are small cliques of $G'$.
Otherwise, consider the case when $C \in \CC$ is a large clique.
The algorithm preserves the fact that either $C$ remains a large clique in $\CC'$ or is replaced by another large clique $C'$ in $\CC'$.
If we justify that every large clique of $\CC'$ has $\OO(k^{d+1})$ vertices, then we are done.
Consider an arbitrary large clique $\hat C \in \CC'$.
For every marked vertex $u \in \hat C$, we associate a set $\DD_u$ of at most $d+1$ large cliques from $\CC' \setminus \{\hat C\}$ such that $u$ has some neighbors in every $C' \in \DD_u$ and all those neighbors are marked neighbors.
Observe that for every marked vertex $u \in \hat C$, the algorithm has marked at least one and at most $2d$ neighbors $v_1,\ldots,v_i$ ($i \leq 2d$) that are outside $C$.
Hence, for every marked vertex $u \in \hat C$, we can associate a set $\DD_{u}$ of at most $d+1$ cliques in $\CC'$.
On the other hand, let $\DD$ be a set of at most $d+1$ large cliques $C_1,\ldots,C_i \in \CC'$ with vertices $u_j \in C_j$ for $j \leq i$.
Then, the marking procedure ${\sf MarkCP}$ marks $d$ neighbors of $u_j$ in the clique $\hat C$.
Hence, the number of marked vertices in $\hat C$ is at most $\sum\limits_{i=1}^{d+1} d{{k}\choose{i}}$.
It implies that every large clique $\hat C \in \CC'$ has $\OO(k^{d+1})$ vertices.
This completes the justification that $|V(G')|$ is $\OO(k^{d+2})$.
\end{proof}

The following lemma illustrates a bijection between the collection of $d$-cuts of $G'$ to the collection of $d$-cuts of $G_1$.

\begin{lemma}
\label{lemma:d-cut-preservation-clique-partition}
Let $\CC_1 = \{\hat C_1,\ldots,\hat C_r\}$ and $\CC' = \{C_1',\ldots,C_r'\}$ be the clique-partition of $G_1$ and $G'$ respectively and $I \uplus J = \{1,\ldots,r\}$ be a partition of $\{1,\ldots,r\}$.
Then, $\hat F = (\bigcup\limits_{i \in I} \hat C_i, \bigcup\limits_{j \in J} \hat C_j)$ is a $d$-cut of $G_1$ if and only if $F' = (\bigcup\limits_{i \in I} C_i', \bigcup\limits_{j \in J} C_j')$ is a $d$-cut of $G'$.
\end{lemma}

\begin{proof}
The forward direction ($\Rightarrow$) is easy to see because $G'$ is a subgraph of $G_1$.

Let us give the proof of the backward direction ($\Leftarrow$).
Suppose that for a partition $I \uplus J = \{1,\ldots,r\}$, $(\bigcup\limits_{i \in I} C_i', \bigcup\limits_{j \in J} C_j')$ is a $d$-cut of $G'$.
We assume for the sake of contradiction that $\hat F = (\bigcup\limits_{i \in I} \hat C_i, \bigcup\limits_{j \in J} \hat C_j)$ is not a $d$-cut of $G_1$.
Then there exists $u \in \bigcup\limits_{i \in I} \hat C_i$ such that $u$ has $d+1$ neighbors in $\bigcup\limits_{j \in J} \hat C_j$.

Let us consider the case when $u$ is a marked vertex, implying that $u \in \bigcup\limits_{i \in I} C_i'$.
Then, there is $v \in N_{G_1}(u)$ such that $v$ is an unmarked neighbor and $v \notin V(G')$.
Then, there exists $t \in \{1,\ldots,r\}$ such that $v \in \hat C_t \setminus C_t'$.
Recall that at most $2d$ neighbors of $u$ were marked by the marking procedure.
Since there is a neighbor $v \in \hat C_t$ is unmarked, $u$ had overall at most $d$ neighbors in some other large cliques $\hat C_1,\ldots,\hat C_p$.
It means that $\sum\limits_{\ell = 1}^p (|N_{G_1}(u) \cap \hat C_{\ell}|) \leq d$ but $\sum\limits_{\ell = 1}^p (|N_{G_1}(u) \cap \hat C_{\ell}|) + |N_{G_1}(u) \cap \hat C_t| > d$.
In particular, $u$ is a high vertex.
But then more than $d$ marked neighbors of $u$ are present in $G'$.
As at most $d$ marked neighbors of $u$ are already present in $\bigcup\limits_{j \in J} \hat C_j$, the same set of marked neighbors of $u$ were also present in $\bigcup\limits_{j \in J} C_j'$.
But, it is also true that there is an unmarked neighbor $v \in \hat C_t$.
But then $u$ has other marked neighbors in $C_t'$ that are also present in $\bigcup\limits_{j \in J} C_j'$.
But then $u$ has a total of more than $d$ marked neighbors present in $\bigcup\limits_{j \in J} C_j'$.
This leads to a contradiction to the hypothesis that $F'$ is a $d$-cut of $G'$.

The other case is when $u$ is unmarked vertex of $\hat C_p$ for some $p \leq r$.
Let us consider these $d+1$ neighbors of $u$ that are distributed across $j \geq 2$ distinct cliques from $\{\hat C_i \mid i \in J\}$.
A reason $u$ is unmarked because there is some other marked vertex $\hat u \in \hat C_p$ such that $\hat u$ also has at least $d+1$ neighbors in the same set of $j \geq 2$ distinct cliques from $\{\hat C_i \mid i \in J\}$.
As $\hat u$ is marked in $G_1$, it must be that $\hat u \in G'$.
In fact in $\hat u \in \bigcup\limits_{i \in I} C_i'$ has at least $d+1$ neighbors in $\bigcup\limits_{j \in J} C_j'$.
Then it leads to a contradiction to the hypothesis that $F'$ is also not a $d$-cut of $G'$.
This completes the proof.
\end{proof}

The above lemma provides us a proper bijection from the set of all $d$-cuts of $G_1$ to all the $d$-cuts of $G'$.
Using this characterization (Lemma \ref{lemma:d-cut-preservation-clique-partition}) and Lemma \ref{lemma:clique-partition-d-cut-preservation}, we can prove our result of this section.
Kernelization algorithms exploit the procedure {\sf MarkCP} and the reduction rules above.

{\ThmCLP*}

\begin{proof}
We prove this for {\enumdcut}.
The proof arguments for {\enummaxdcut} and {\enummindcut} will be similar.
Let $(G, \CC, k)$ be an input instance to {\enumdcut} and $\{C_1,\ldots,C_k\}$ be the set of $k$ cliques satisfying $V(G) = C_1 \uplus \cdots \uplus C_k$.
If $|V(G)| \leq k^{d+2}$, then we output $(G, \CC, k)$ as the output instance.
So, we can assume without loss of generality that $|V(G)| > k^{d+2}$.
We invoke Lemma \ref{lemma:clique-partition-d-cut-preservation} to obtain $(G', \CC', k')$.
Since the algorithm by Lemma \ref{lemma:clique-partition-d-cut-preservation} runs in $k^{d+1}n^{\OO(1)}$-time, and we invoke this algorithm only when $n > k^{d+2}$, this reduction algorithm runs in polynomial-time.
This completes the description of the reduction algorithm and $|V(G')|$ is $\OO(k^{d+2})$.

The solution lifting algorithm works as follows.
Let $F'$ be a $d$-cut of $(G', \CC', k')$ such that $\CC'$ be the set of cliques for the clique partition of $G'$.
It follows from Observation \ref{obs:d-cut-property-clique-partition} that for any clique $C' \in \CC'$, $E(G'[C']) \cap F' = \emptyset$.
Let $I \uplus J = \{1,\ldots,k'\}$ be a partition such that $F' = E(\bigcup\limits_{i \in I} C_i', \bigcup\limits_{j \in J} C_j')$.
The solution lifting algorithm first computes $\hat F = E(\bigcup\limits_{i \in I} \hat C_i, \bigcup\limits_{j \in J} \hat C_j)$.
Due to Lemma \ref{lemma:d-cut-preservation-clique-partition}, it holds that $E(\bigcup\limits_{i \in I} \hat C_i, \bigcup\limits_{j \in J} \hat C_j)$ is a $d$-cut of $G_1$.
After that the solution lifting algorithm outputs $F = \hat F$ as the $d$-cut of $G$.
Again from Lemma \ref{lemma:clique-partition-step-1-preservation}, $F$ is a $d$-cut of $G$.
Clearly, both the reduction algorithm and the solution lifting algorithms run in polynomial-time.
Hence, it is a bijective enumeration kernel with $\OO(k^{d+2})$ vertices.
\end{proof}

\fi

\section{Conclusions and Future Research}
\label{sec:conclusion}

We have initiated the study of (polynomial-delay) enumeration kernelization for variants of {\dcut} problem. We have provided (polynomial-delay) enumeration kernels when parameterized by the size of vertex cover, the size of neighborhood diversity and clique partition number of the graph.
In addition, we have also explained why we are unlikely to obtain a polynomial-delay enumeration kernel when parameterized by vertex-integrity or treedepth.
The negative result of the Proposition \ref{prop:no-poly-kernel-treedepth} can also be extended when we consider a combined parameter -- treedepth, the number of edges in the cut and maximum degree of the graph.
Obtaining our positive results has required proving some properties that are more specific to $d$-cuts of a graph for any arbitrary $d \geq 2$ in addition to the properties specific to the parameter we consider (see Theorem 7 of \cite{GomesS21}).
We believe that a lot of follow-up research works can be done for {\enumdcut} and the other enumeration variants. In particular, we have realized that there are some challenges to obtain (polynomial-delay) enumeration kernels when we consider structurally smaller parameters such as feedback vertex set size or cluster vertex deletion set size.
An interesting future research direction would be to explore the above-mentioned other structural parameters, the status of which are open even for $d=1$, i.e. for the enumeration variant of the {\MC} problem.
Another interesting research direction would be to identify a natural parameterized problem whose decision version admits polynomial kernel but whose  enumeration version does not admit a polynomial-delay kernelization of polynomial size.
Obtaining such a negative result would likely need to rely on a new lower bound framework.



\begin{thebibliography}{10}

\bibitem{AravindKK17}
N.~R. Aravind, Subrahmanyam Kalyanasundaram, and Anjeneya~Swami Kare.
\newblock On structural parameterizations of the matching cut problem.
\newblock In Xiaofeng Gao, Hongwei Du, and Meng Han, editors, {\em
  Combinatorial Optimization and Applications - 11th International Conference,
  {COCOA} 2017, Shanghai, China, December 16-18, 2017, Proceedings, Part {II}},
  volume 10628 of {\em Lecture Notes in Computer Science}, pages 475--482.
  Springer, 2017.

\bibitem{AravindKK22}
N.~R. Aravind, Subrahmanyam Kalyanasundaram, and Anjeneya~Swami Kare.
\newblock Vertex partitioning problems on graphs with bounded tree width.
\newblock {\em Discret. Appl. Math.}, 319:254--270, 2022.

\bibitem{AravindS21}
N.~R. Aravind and Roopam Saxena.
\newblock An {FPT} algorithm for matching cut and d-cut.
\newblock In Paola Flocchini and Lucia Moura, editors, {\em Combinatorial
  Algorithms - 32nd International Workshop, {IWOCA} 2021, Ottawa, ON, Canada,
  July 5-7, 2021, Proceedings}, volume 12757 of {\em Lecture Notes in Computer
  Science}, pages 531--543. Springer, 2021.

\bibitem{BentertFNN19}
Matthias Bentert, Till Fluschnik, Andr{\'{e}} Nichterlein, and Rolf
  Niedermeier.
\newblock Parameterized aspects of triangle enumeration.
\newblock {\em J. Comput. Syst. Sci.}, 103:61--77, 2019.

\bibitem{BodlaenderJK14}
Hans~L. Bodlaender, Bart M.~P. Jansen, and Stefan Kratsch.
\newblock Kernelization lower bounds by cross-composition.
\newblock {\em {SIAM} J. Discret. Math.}, 28(1):277--305, 2014.

\bibitem{Bonsma09}
Paul~S. Bonsma.
\newblock The complexity of the matching-cut problem for planar graphs and
  other graph classes.
\newblock {\em J. Graph Theory}, 62(2):109--126, 2009.

\bibitem{ChenHLLP21}
Chi{-}Yeh Chen, Sun{-}Yuan Hsieh, Ho{\`{a}}ng{-}Oanh Le, Van~Bang Le, and
  Sheng{-}Lung Peng.
\newblock Matching cut in graphs with large minimum degree.
\newblock {\em Algorithmica}, 83(5):1238--1255, 2021.

\bibitem{Chvatal84}
Vasek Chv{\'{a}}tal.
\newblock Recognizing decomposable graphs.
\newblock {\em J. Graph Theory}, 8(1):51--53, 1984.

\bibitem{CreignouKPSV19}
Nadia Creignou, Markus Kr{\"{o}}ll, Reinhard Pichler, Sebastian Skritek, and
  Heribert Vollmer.
\newblock A complexity theory for hard enumeration problems.
\newblock {\em Discret. Appl. Math.}, 268:191--209, 2019.

\bibitem{CreignouKMMOV19}
Nadia Creignou, Ra{\"{\i}}da Ktari, Arne Meier, Julian{-}Steffen M{\"{u}}ller,
  Fr{\'{e}}d{\'{e}}ric Olive, and Heribert Vollmer.
\newblock Parameterised enumeration for modification problems.
\newblock {\em Algorithms}, 12(9):189, 2019.

\bibitem{CreignouMMSV17}
Nadia Creignou, Arne Meier, Julian{-}Steffen M{\"{u}}ller, Johannes Schmidt,
  and Heribert Vollmer.
\newblock Paradigms for parameterized enumeration.
\newblock {\em Theory Comput. Syst.}, 60(4):737--758, 2017.

\bibitem{CyganFKLMPPS15}
Marek Cygan, Fedor~V. Fomin, Lukasz Kowalik, Daniel Lokshtanov, D{\'{a}}niel
  Marx, Marcin Pilipczuk, Michal Pilipczuk, and Saket Saurabh.
\newblock {\em Parameterized Algorithms}.
\newblock Springer, 2015.

\bibitem{Damaschke06}
Peter Damaschke.
\newblock Parameterized enumeration, transversals, and imperfect phylogeny
  reconstruction.
\newblock {\em Theor. Comput. Sci.}, 351(3):337--350, 2006.

\bibitem{GomesS21}
Guilherme de~C.~M.~Gomes and Ignasi Sau.
\newblock Finding cuts of bounded degree: Complexity, {FPT} and exact
  algorithms, and kernelization.
\newblock {\em Algorithmica}, 83(6):1677--1706, 2021.

\bibitem{Diestel-Book}
Reinhard Diestel.
\newblock {\em Graph Theory, 4th Edition}, volume 173 of {\em Graduate texts in
  mathematics}.
\newblock Springer, 2012.

\bibitem{DowneyF13}
Rodney~G. Downey and Michael~R. Fellows.
\newblock {\em Fundamentals of Parameterized Complexity}.
\newblock Texts in Computer Science. Springer, 2013.

\bibitem{FominLSZ19}
Fedor~V Fomin, Daniel Lokshtanov, Saket Saurabh, and Meirav Zehavi.
\newblock {\em {Kernelization: Theory of Parameterized Preprocessing}}.
\newblock Cambridge University Press, 2019.

\bibitem{GolovachKKL22}
Petr~A. Golovach, Christian Komusiewicz, Dieter Kratsch, and Van~Bang Le.
\newblock Refined notions of parameterized enumeration kernels with
  applications to matching cut enumeration.
\newblock {\em J. Comput. Syst. Sci.}, 123:76--102, 2022.

\bibitem{JansenS23}
Bart M.~P. Jansen and Bart van~der Steenhoven.
\newblock Kernelization for counting problems on graphs: Preserving the number
  of minimum solutions.
\newblock In Neeldhara Misra and Magnus Wahlstr{\"{o}}m, editors, {\em 18th
  International Symposium on Parameterized and Exact Computation, {IPEC} 2023,
  September 6-8, 2023, Amsterdam, The Netherlands}, volume 285 of {\em LIPIcs},
  pages 27:1--27:15. Schloss Dagstuhl - Leibniz-Zentrum f{\"{u}}r Informatik,
  2023.

\bibitem{KomusiewiczKL20}
Christian Komusiewicz, Dieter Kratsch, and Van~Bang Le.
\newblock Matching cut: Kernelization, single-exponential time fpt, and exact
  exponential algorithms.
\newblock {\em Discret. Appl. Math.}, 283:44--58, 2020.

\bibitem{KMS25}
Christian Komusiewicz, Diptapriyo Majumdar, and Frank Sommer.
\newblock Polynomial-size enumeration kernelizations for long path enumeration.
\newblock In {\em Graph-Theoretic Concepts in Computer Science, 27th
  International Workshop, {WG} 2025, Otzenhausen, Germany, 2025, Proceedings},
  2025.
\newblock To appear, full version available at
  https://doi.org/10.48550/arXiv.2502.21164.

\bibitem{KratschL16}
Dieter Kratsch and Van~Bang Le.
\newblock Algorithms solving the matching cut problem.
\newblock {\em Theor. Comput. Sci.}, 609:328--335, 2016.

\bibitem{Lampis12}
Michael Lampis.
\newblock Algorithmic meta-theorems for restrictions of treewidth.
\newblock {\em Algorithmica}, 64(1):19--37, 2012.

\bibitem{LeL19}
Ho{\`{a}}ng{-}Oanh Le and Van~Bang Le.
\newblock A complexity dichotomy for matching cut in (bipartite) graphs of
  fixed diameter.
\newblock {\em Theor. Comput. Sci.}, 770:69--78, 2019.

\bibitem{LeT22}
Van~Bang Le and Jan~Arne Telle.
\newblock The perfect matching cut problem revisited.
\newblock {\em Theor. Comput. Sci.}, 931:117--130, 2022.

\bibitem{LokshtanovM0Z24}
Daniel Lokshtanov, Pranabendu Misra, Saket Saurabh, and Meirav Zehavi.
\newblock Kernelization of counting problems.
\newblock In Venkatesan Guruswami, editor, {\em 15th Innovations in Theoretical
  Computer Science Conference, {ITCS} 2024, January 30 to February 2, 2024,
  Berkeley, CA, {USA}}, volume 287 of {\em LIPIcs}, pages 77:1--77:23. Schloss
  Dagstuhl - Leibniz-Zentrum f{\"{u}}r Informatik, 2024.

\bibitem{LokshtanovPRS17}
Daniel Lokshtanov, Fahad Panolan, M.~S. Ramanujan, and Saket Saurabh.
\newblock Lossy kernelization.
\newblock In Hamed Hatami, Pierre McKenzie, and Valerie King, editors, {\em
  Proceedings of the 49th Annual {ACM} {SIGACT} Symposium on Theory of
  Computing, {STOC} 2017, Montreal, QC, Canada, June 19-23, 2017}, pages
  224--237. {ACM}, 2017.

\bibitem{LuckePR22}
Felicia Lucke, Dani{\"{e}}l Paulusma, and Bernard Ries.
\newblock Finding matching cuts in {H}-free graphs.
\newblock In Sang~Won Bae and Heejin Park, editors, {\em 33rd International
  Symposium on Algorithms and Computation, {ISAAC} 2022, December 19-21, 2022,
  Seoul, Korea}, volume 248 of {\em LIPIcs}, pages 22:1--22:16. Schloss
  Dagstuhl - Leibniz-Zentrum f{\"{u}}r Informatik, 2022.

\bibitem{PatrignaniP01}
Maurizio Patrignani and Maurizio Pizzonia.
\newblock The complexity of the matching-cut problem.
\newblock In Andreas Brandst{\"{a}}dt and Van~Bang Le, editors, {\em
  Graph-Theoretic Concepts in Computer Science, 27th International Workshop,
  {WG} 2001, Boltenhagen, Germany, June 14-16, 2001, Proceedings}, volume 2204
  of {\em Lecture Notes in Computer Science}, pages 284--295. Springer, 2001.

\bibitem{Savage82}
Carla~D. Savage.
\newblock Depth-first search and the vertex cover problem.
\newblock {\em Inf. Process. Lett.}, 14(5):233--237, 1982.

\bibitem{Thilikos21}
Dimitrios~M. Thilikos.
\newblock Compactors for parameterized counting problems.
\newblock {\em Comput. Sci. Rev.}, 39:100344, 2021.

\bibitem{Thurley07}
Marc Thurley.
\newblock Kernelizations for parameterized counting problems.
\newblock In Jin{-}yi Cai, S.~Barry Cooper, and Hong Zhu, editors, {\em Theory
  and Applications of Models of Computation, 4th International Conference,
  {TAMC} 2007, Shanghai, China, May 22-25, 2007, Proceedings}, volume 4484 of
  {\em Lecture Notes in Computer Science}, pages 703--714. Springer, 2007.

\bibitem{Wasa16arXiv}
Kunihiro Wasa.
\newblock Enumeration of enumeration algorithms.
\newblock {\em CoRR}, abs/1605.05102, 2016.

\end{thebibliography}

\end{document}